\documentclass[11pt, twoside, a4paper]{article}
\usepackage{latexsym}
\usepackage{theorem}
\usepackage{amssymb,amsfonts,amsmath,amscd}
\usepackage{bbm}
\usepackage{xypic}
\usepackage{bm}

\def\cA{{\cal A}}

\def\cD{{\cal D}}

\def\cH{{\cal H}}

\def\cK{{\cal K}}

\def\cO{{\cal O}}

\def\cP{{\cal P}}

\def\cT{{\cal T}}

\def\bR{\mathbb{R}}
\def\bC{\mathbb{C}}
\def\bN{\mathbb{N}}

\def\ulp{\underline {p}}

\def\ulx{\underline {x}}

\def\ul0{\underline {0}}

\def\fF{{\mathfrak F}}

\usepackage{mathrsfs}

\def\sD{{\mathscr D}}

\def\sS{{\mathscr S}}

\DeclareMathOperator{\supp}{supp}

\newtheorem{theorem}{Theorem}[section]
\newtheorem{lemma}[theorem]{Lemma}

\newtheorem{proposition}[theorem]{Proposition}
\newtheorem{definition}[theorem]{Definition}
\newenvironment{proof}[1][Proof]{\begin{trivlist}
\item[\hskip \labelsep {\bfseries #1}]}{\end{trivlist}}

\newcommand{\eins}{ {1\hspace*{-3pt}{\rm l}}}

\newcommand{\qed}{\hfill$\square$}

\newlength{\fdagwidth}
\newlength{\diagupwidth}
\newlength{\stepback}
\newcommand{\fdag}[2][\diagup]
{\text{$#2$\settowidth{\fdagwidth}{$#2$}\settowidth{\diagupwidth}{$#1$}\setlength{\stepback}
{0.5\fdagwidth}\hspace{-\stepback}\hspace{-0.5\diagupwidth}$#1$\hspace{\stepback}\hspace{-0.5\diagupwidth}}}

\newcommand{\Larry}{\mbox{\Large $\upsilon$}}

\title{}
\author{} 
\date{\today}    

\renewcommand{\theenumi}{\alph{enumi}} 
\renewcommand{\labelenumi}{(\theenumi)}

\begin{document}

{\huge
 \begin{center}
 {\sf Dirac Field on Moyal-Minkowski Spacetime and Non-Commutative Potential
 Scattering}
 \end{center}
 }
 ${}$\\
{\Large
 \begin{center}
 \bf Markus Borris and Rainer Verch\\[22pt]
 {\normalsize \sf
  Institut f\"ur Theoretische Physik, Universit\"at Leipzig\\
  Postfach 100 920, D-04009 Leipzig, Germany
  \\[16pt]
  }
 \end{center}
 }
${}$ \\
${}$ \hfill {\sl Dedicated to Klaus Fredenhagen} \\
${}$ \hfill {\sl on the occasion of his 61st birthday}
\\[10pt]
\noindent
{\small {\bf Abstract.} \quad The quantized free Dirac field is considered on Minkowski spacetime (of general dimension).
The Dirac field is coupled to an external scalar potential whose support is finite in time and which acts by a
Moyal-deformed multiplication with respect to the spatial variables. The Moyal-deformed multiplication corresponds
to the product of the algebra of a Moyal plane described in the setting of spectral geometry.
It will be explained how this leads to an interpretation of the Dirac field as a quantum field theory on 
Moyal-deformed Minkowski spacetime (with commutative time) in a setting of Lorentzian spectral geometries of which
some basic aspects will be sketched. The scattering transformation will be shown to be unitarily implementable
in the canonical vacuum representation of the Dirac field. Furthermore, it will be indicated how the functional derivatives of
the ensuing unitary scattering operators with respect to the strength of the non-commutative potential induce, in
the spirit of Bogoliubov's formula, quantum field operators (corresponding to observables) depending on the elements of the non-commutative algebra of
Moyal-Minkowski spacetime. 
 }
\\[10pt]
\section{Introduction}

There are several theoretical arguments leading to the hypothesis that spacetime coordinates may,
at extremely short length scales, no longer be described by continuous, mutually commuting numbers,
but by non-commutative quantities, so that spacetime coordinates of events become subject to 
uncertainty relations characteristic of quantum physics. In a somewhat more technical sense, this
means that the commutative algebra containing the coordinate functions of a spacetime manifold
in the sense of general relativity is to be replaced by a non-commutative algebra. There is a fair number
of publications devoted to giving arguments to this effect and we shall make no attempt at reviewing
this issue; for the best motivation, in our opinion, see \cite{DFR}.

For more than two decades by now, various approaches to non-commutative spacetimes have been investigated.
In the context of physically motivated investigations, the procedure consists in replacing either
the numerical coordinate functions on spacetime by elements of a non-commutative algebra, together with
(more or less, physically motivated) commutation relations for a set of gene\-ra\-ting elements, or in replacing
the spacetime symmetry group by a suitable deformation (e.g., a Hopf algebra), or both. Some of these
non-commutative spacetime models correspond to Lorentzian spacetime structure, some others to Riemannian.
It is then attempted to formulate physical theories (of matter or of gauge fields) over such spacetimes,
patterned either after classical field theory, after quantum mechanics, or quantum field theory.
(We refrain from listing the quite extended literature that is available on this matter and refer
instead to the review \cite{Szabo} for comments and references.) 
Some of the results obtained along that route offer interesting and promising perspectives. 
However, it is quite difficult, at present, to compare all these various approaches
\cite{BuchSummers,BaDoFrPi2,GrosseLechner,GrWu2,RivVTWu}. Their interpretation is often
problematic, in particular when the underlying non-commutative spacetime geometry corresponds to 
Riemannian metric signature. It would be very much desirable to have a mathematical and conceptual framework
which allows to stage a discussion as to which of the various proposals for non-commutative spacetime
geometries appear more favourable than others. 

While it is very likely that such a discussion won't lead to definite conclusions in the absence of experimental
evidence for non-commutative spacetime structure, it would still be a valuable step to have a mathematical
and conceptual framework broad enough so that the various approaches can be systematically compared.
Concerning Riemannian non-commutative geometries, it appears that the general approach by spectral geometry,
due to Alain Connes, provides such a framework in principle
\cite{Connes,Connes2,Connes3,GVF}. It incorporates many of the known examples
of Riemannian non-commutative manifolds. Moreover, the structural results which have been obtained in this
approach --- among them, but not only, the result that a spectral geometry with a commutative {``function''}- algebra
actually corresponds to a Riemannian manifold --- lend further to its support. Furthermore, one can give in this
approach a certain reformulation of the current standard model of elementary particle physics
\cite{ConnesLott}. This reformulation
is, however, not a full quantum field theory, and cannot account for all processes of elementary particle physics
in the same way as a quantum field theory. The conclusion may be --- and we shall in fact take this point of
view which we share, amongst others, with \cite{DFR} --- that a combination of non-commutative spacetime geometry and
quantum field theory is a most promising candidate for describing processes at extremely short distances
and high energies, up to (and possibly including) Planck scale. 

Now, the closest connection to the physical geometry of spacetime in quantum field theory is seen when spacetime
has Lorentzian signature. In keeping with what we just mentioned, it would then be of interest to have a framework
of spectral geometry corresponding to Lorentzian metric structure which includes, at least in approximation, those of
the known non-commutative
examples of Lorentzian spacetimes for which a good physical interpretation can be given. 
This is by no means a humble request as it is not at all straightforward to generalize the setting of spectral geometry
from Riemannian to Lorentzian signature. While some proposals have been made \cite{Hawkins,KopfPaschke,Moretti,Strohmaier},
they still seem to lack certain important ingredients and results. What appears to be lacking is a concept
of covariance (see \cite{PaschkeVerch} for discussion), and structural results like in the Riemannian case.
Endeavour in this regard shall be brought to the fore elsewhere \cite{PRV}, where a new approach to Lorentzian
spectral geometry will be developed.  

Accepting for the time being that the framework to appear in \cite{PRV} yields a viable general
setting for non-commutative Lorentzian geometry, the next question is if there is a general method of constructing
quantum field theories over such Lorentzian spectral geometries. To wit, the ways of assigning
quantum field theories with non-commutative spacetime models proposed up to now are, by and large, ad hoc, and not based on a systematic
general method. It is therefore of importance to see if there is a systematic way to assign quantum field theories to
abstractly described Lorentzian spectral geometries, and to identify the meaning of their observables.

In this publication, we attempt some first steps along these lines. We shall present a very superficial sketch of 
some elements of the
setting of Lorentzian spectral geometry anticipated to be set out in detail in \cite{PRV}. We discuss an
abstract way of assigning a quantum field theory to any Lorentzian spectral geometry. In the case of a ``classical''
spacetime, with the usual commutative algebra of coordinate functions, this method amounts to assigning the quantized
linear Dirac field to that spacetime. Part of our discussion will address the construction and meaning of observables in this setting.
Much of this  will actually be carried out at a much more concrete level by means of an example: The Dirac field on
the Moyal-deformed non-commutative version of Minkowski spacetime.  
It will turn out that the spectral geometrical (``spectral triple'') data of Moyal-Minkowski spacetime agree with those
of ordinary Minkowski spacetime except that the commutative algebra of functions on Minkowski spacetime --- which may be taken
to be Schwartz-type functions, $\mathscr{S}(\mathbb{R}^n)$ --- is replaced by the algebra of Schwartz-type functions
$\mathscr{S}(\mathbb{R}^n)_\star$ where the pointwise multiplication is replaced by the Moyal product (with commutative time).
This indeed fits into the setting of spectral geometry as was demonstrated in detail in \cite{GGISV}. One can assign to
these spectral data abstractly a quantized Dirac field essentially by second quantization (or, more specifically, CAR-quantization)
of the other spectral data, given by a Hilbert space carrying a Dirac operator etc. These quantizations agree for both usual
Minkowski spacetime and Moyal-Minkowski spacetime since they depend on the same spectral data. A difference is only noticable when
the algebras of coordinate functions are brought into play. A very natural way to let them enter is via a scattering process.

Let us explain this at the level of the Dirac field on usual Minkowski spacetime. Consider the coupling of the Dirac field
to an external scalar potential $V = V(x^0,\underline{x})$ where $x^0$ is the time-coordinate and $\underline{x}$ denotes the
spatial coordinates with respect to some chosen Lorentz frame. Assume that the potential is of Schwartz type and has finite
support in time, i.e.\ it is different from zero only for
$x^0$-coordinates lying in some finite interval. Then the scattering of the Dirac field is described  
by a unitary S-matrix, $S_V$, in the vacuum representation of the free Dirac field \cite{Araki2,Thaller}. Re-writing the scalar potential $V$ as
$c(x^0,\underline{x}) = V(x^0,\underline{x})$, and accordingly, $S_c = S_V$, one can form the functional derivative
$\Phi(c) = -i\left.d/d\lambda\right|_{\lambda = 0} S_{\lambda c}$ of the scattering matrix with respect to the strength of
the scalar scattering potential. This is a special case of ``Bogo\-liu\-bov's formula'' \cite{Bogoliubov}, whose content is,
roughly speaking, the idea that observable quantum fields can be obtained from scattering matrices by functional differentiation
with respect to the interaction strength. In the case of the Dirac field on Minkowski spacetime, one finds
that $\Phi(c) = : \bm\psi^+\bm\psi:(c)$, i.e.\ the Wick-ordered operator standing for ``absolute square of field strength'', which
is an observable quantum field, as opposed to the quantized Dirac field operators $\bm\psi(f)$ labelled by spinor fields $f$,
which aren't directly observable.

For the Dirac field on Moyal-Minkowski spacetime, one can proceed in a similar fashion, but now replacing the ordinary scalar potentials
by ``non-commutative'' scalar potentials, where $V = V(x^0,\underline{x})$ is coupled to spinor fields not by pointwise multiplication
at each spacetime point, but by Moyal-multi\-pli\-ca\-tion. (By the nature of Moyal-multiplication, this yields a non-local interaction
of the Dirac field 
with the external potential.) We will show that, under certain, general conditions, there will then again be unitary S-matrices
$S_c^M =S^M_V$ ($V = V_c$) describing scattering by such non-commtuative potentials in the vacuum representation of the free Dirac field. Furthermore,
we also show that the corresponding functional derivatives $\Phi(c) = -i\left.d/d\lambda\right|_{\lambda = 0} S^M_{\lambda c}$ 
exist as (essentially) selfadjoint operators. These operators are now in a natural way labelled by the elements $c$ of the
non-commutative algebra $\mathscr{S}(\mathbb{R}^n)_\star$ of Schwartz functions endowed with Moyal-multiplication. In principle,
this method of assigning operators (to be interpreted as observables) to elements in the --- possibly non-commutative --- algebra of
spacetime coordinate functions could be carried over to other models of non-commutative spacetimes, whenever it is possible
to have a well-posed scattering process in the indicated sense (which seems to imply restrictions on the degree of non-commutativity
of time-like coordinates, at least asymptotically).    
 
Let us now describe the contents and organization of the present article in more detail.
In Sec.\ 2 we present the theory of the Dirac field. First, the classical Dirac field coupled to an
external scalar potential, together with the theory of
solutions, will be summarized, on $n = 1 + s$ dimensional Minkowski spacetime, where $n$ is even or fulfills
the relations $n = 3,9$ mod $8$. This then implies the existence of a ``self-dual'' charge-conjugation.
These considerations draw mainly on material in \cite{Dimock,Coquereaux2,GVF}. We also present the abstract CAR ($C^*$-algebraic)
quantization of the Dirac field
and summarize the connection between the scattering transformation induced by a scalar scattering potential and
the $C^*$-algebraic Bogoliubov transformations which they induce, following mainly Araki's works
\cite{Araki2,Araki}. The scattering transformations will be considered both in covariant form (following ideas in \cite{BFV}), and
in the Hamiltonian form at the level of Cauchy-data, since the interplay between both formulations will be useful later on.

Sec.\ 3 is a short section recapitulating the basics on the vacuum-repre\-sen\-ta\-tion of the Dirac field, both in covariant description
and in the Hamiltonian, or Cauchy-data, formulation; and citing results on the unitary implementability of the 
Bogoliubov transform describing scalar potential scattering from \cite{Palmer}. The latter is mainly included for comparison
with the non-commutative case treated later.

In Section 4 we discuss the non-commutative algebra $\mathscr{S}(\mathbb{R}^n)_\star$ of Schwartz functions with
the Moyal product (with commutative time). Our presentation draws heavily on \cite{GGISV} with some small alterations
adapted to our setting.

Sec.\ 5 contains the main conceptual considerations. In this longer section, we give a sketch of some of the ingredients
of the approach to Lorentzian spectral geometry expected to appear in \cite{PRV}, illustrating the
main points by the example of Moyal-deformed Minkowski spacetime with commutative time. We will discuss the general route
of associating to a Lorentzian spectral geometry a quantum field theory, where the observables depend on the elements of
the (non-commutative) ``function-''algebra in the spectral geometric data, via the procedure of abstract CAR quantization and Bogoliubov's formula, as indicated above, in some detail. (Incidentally,
a formally similar set-up appears in \cite{GVultra}, but
in this reference, the quantum field transformations are not
linked to any dynamical process like scattering.)
 Also some speculations about a possible general structure
of quantum field theories over Lorentzian spectral geometries will appear in Sec.\ 5.

In Section 6 we investigate the solution properties of the Dirac equation with a non-commutative scalar potential
(we investigate two such potentials obtained by Moyal-multiplying scalar functions with spinor fields). Since the 
time coordinate, in a chosen Lorentz frame,
 is still commutative, we can formulate the Cauchy-problem and establish its well-posedness. There are unique
advanced and retarded fundamental solutions with respect to the chosen Lorentzian frame. This is implied by (in fact, equivalent to) 
a uniquely solvable initial value problem in the Hamiltonian formulation of the Dirac equation, which we solve by constructing
the Dyson series for the time-dependent interaction Hamiltonians (at the one-particle level). Correspondingly,
we construct the one-particle scattering transformations, which induce Bogoliubov-transformations on the CAR-algebra
of the free Dirac field, describing the scattering of the field by the non-commutative potential. This discussion parallels
the discussion of the usual scalar potential scattering in many formal respects, but at several points, different
arguments are required due to the non-local character of the non-commutative potential with respect to spatial coordinates.

The main results will be presented in Sec.\ 7. It will be proved that the Bogoliubov-transformations describing non-commutative
potential scattering are unitarily implementable in the vacuum presentation of the free quantized Dirac field on Minkowski spacetime.
In order to prove this, we make significant use of an earlier result by Langmann and Mickelsson \cite{LangmannMickelsson} who developed a sufficient
criterion for unitary implementability that can be applied in the case considered here. We show that this criterion is fulfilled.
Furthermore, one can differentiate the Bogoliubov transformation with respect to the strength of the scattering potential as 
mentioned above. This leads to a derivation on the CAR-algebra, which we show to be induced by an essentially selfadjoint ope\-ra\-tor
$\Phi(c)$ in the vacuum-representation of the Dirac field. This is the precise form of the relation
$\Phi(c) = -i\left.d/d\lambda\right|_{\lambda = 0} S^M_{\lambda c}$.

Finally, in Sec.\ 8, we derive the action of the derivative of the commutative and non-commutative scattering
Bogoliubov transformations with respect to the potential strength on the generating elements of the Dirac field algebra (the field operators), drawing on the results of Sec.\ 6. Together with
the relation $\Phi(c) = : \bm\psi^+\bm\psi:(c)$, which will be proved in Appendix A, this result finally hints at
the operational meaning of $\Phi(c)$, and illustrates how an interpretation of quantum field observables on
a non-commutative spacetime may be reached at in more general situations.

There is a short conclusion and outlook in Sec.\ 9.

\section{The Dirac Field}\label{S_DiracField}

We start our discussion by summarizing some of the essentials on Dirac spinors 
and Dirac representations as far as
required for describing the quantized Dirac field on $n=1+s$ dimensional 
Minkowski spacetime. In doing so, we proceed
quite leisurely; most of our presentation relies on
\cite{Dimock}, \cite{Araki}, \cite{Araki2}, \cite{BGP}, \cite{Thaller}, \cite{GVF}, \cite{Coquereaux2}.
We refer to these references for proofs of the
statements appearing in this section.

Minkowski spacetime of dimension $n=1+s$ will be described as $\bR^n$ with the 
Minkowskian metric
\[
\eta=(\eta_{\mu\nu})_{\mu,\nu=0}^s=\text{diag}(1,-1,\ldots,-1)
\]
where the entry $-1$ appears $s$ times.
[The opposite signature convention would in some respects suit the NCG context 
better, but we find it convenient to stick
to the convention which is more common in QFT.]
For any given $n=1+s\in\bN$, $s\geq 1$, we set
\begin{equation}\label{E_1}
N=N(n)=\begin{cases}2^{n/2}&:n\text{ even}\\2^{(n-1)/2}&:n\text{ odd}\end{cases}.
\end{equation}
Then we refer to a collection $(\gamma_0,\gamma_1,\ldots,\gamma_s)$ of $N\times 
N$-matrices as a set of
{\it Dirac matrices} if the relations
\begin{eqnarray*}
\gamma_\mu\gamma_\nu+\gamma_\nu\gamma_\mu=2\eta_{\mu\nu}\eins\quad 
(\mu,\nu=0,1,\ldots,s)\nonumber\\
\gamma_0^*=\gamma_0,\quad\gamma_k^*=-\gamma_k\quad (k=1,\ldots,s)
\end{eqnarray*}
are fulfilled.
A set of Dirac matrices thus corresponds to an irreducible Dirac representation 
of the complexified Clifford algebra
$\bC l_{1,s}$; it exists for all $n\geq 2$.

We shall from now on restrict ourselves to dimensions
\begin{equation}\label{E_3}
n\text{ even or }n=3,9\text{ mod }8.
\end{equation}
For these values of $n$, it is possible to find a {\it charge conjugation 
operator} $C:\bC^N\rightarrow \bC^N$
for the Dirac matrices $(\gamma_0,\gamma_1,\ldots,\gamma_s)$; this means that 
$C$ is an antilinear involution ($C^2=\eins$)
satisfying
\begin{equation}\label{E_ChConjC}
C\gamma_\mu=-\gamma_\mu C,
\end{equation}
whence, restriction to dimensions $n$ with~(\ref{E_3}) has the advantage that 
one can quantize Dirac fields with quite
arbitrary (real) potentials in the ``self-dual formalism'' at the level of 
spinor fields only, without need to use a
``doubled'' system of spinor (and co-spinor) fields. The resulting 
simplification is convenient later when discussing
the quantized Dirac field on Moyal-deformed spacetime. 

Let $(\gamma_0,\gamma_1,\ldots,\gamma_s)$ be a set of Dirac matrices with charge 
conjugation $C$. Then we denote by
\begin{equation}\label{E_DiracOpWithmV}
D_V=(-i\fdag\partial+m)+V
\end{equation}
the Dirac operator (with fixed constant mass $m>0$) with potential term $V\in 
C^\infty(\bR^n,\bR)$. $D_V$ acts on
$f\in C^\infty(\bR^n,\bC^N)$ according to
\begin{equation*}
(D_Vf)(x)=(-i\fdag\partial+m)f(x)+V(x)f(x)\quad (x\in\bR^n),
\end{equation*}
i.e. $V$ acts as a (scalar) multiplication operator, and writing $f^A(x)$ for 
the components of $f(x)$ regarded as a
column vector, the explicit definition of $(-i\fdag\partial+m)$ is given by
\begin{equation}\label{E_DefDiracOpNoPot}
((-i\fdag\partial+m)f)^A(x)=-i\gamma^{\mu A}_{\ \ B}\frac{\partial}{\partial 
x^\mu}f^B(x)+mf^A(x)\quad (x\in\bR^n),
\end{equation}
where $\gamma^\mu=\eta^{\mu\nu}\gamma_\nu$ with $(\eta^{\mu\nu})=\text{diag}(1,-
1,\ldots,-1),$ and with
$\gamma^{\mu A}_{\ \ B}$ denoting the matrix entries of $\gamma^\mu$. We also 
make use of the summation convention so
that doubly appearing indices are understood as being summed over.

On $C^\infty_0(\bR^n,\bC^N)$ we can introduce the sesquilinear form
\begin{equation}\label{E_SesqLinFormLRangle}
\langle f,h \rangle=\int_{\bR^n}\gamma_{0AB}\bar f^B(x)h^A(x)d^nx\quad (f,h\in 
C^\infty_0(\bR^n,\bC^N)),
\end{equation}
where $\gamma_{0AB}$ are the matrix elements of $\gamma_0$.
The charge conjugation $C$ is a skew conjugation for this sesquilinear form, 
that is,
\begin{equation}\label{E_CSkewForSesqu}
\langle Cf,Ch \rangle=-\langle h,f \rangle\quad (f,h\in C^\infty_0(\bR^n,\bC^N)).
\end{equation}

Since we will also use the description of solutions to the Dirac equation in 
terms of their Cauchy data, we have cause
to introduce also the following objects. Let $t\in\bR$ and define the $x^0=t$ 
hyperplane
\begin{equation*}
\Sigma_t=\{x=(x^0,x^1,\ldots,x^s)\in\bR^n:x^0=t\}
\end{equation*}
in $n=1+s$ dimensional Minkowski spacetime. We introduce the Hilbert space 
$\cD_t=L^2(\Sigma_t,\bC^N)$ with canonical
scalar product
\begin{equation*}
(v,w)_\cD=\int_{\Sigma_t}\bar v^A(\ulx)\delta_{AB}w^B(\ulx)d^s\ulx\quad 
(v,w\in\cD_t).
\end{equation*}
Each $\cD_t$ is canonically isomorphic to $L^2(\bR^s,\bC^N)$. Note that the 
charge conjugation $C$ induces a
conjugation, denoted by the same symbol $C$, on each $\cD_t$, i.e. it holds that
\begin{equation*}
(Cv,Cw)_\cD=(w,v)_\cD\quad (v,w\in\cD_t).
\end{equation*}
For a subset $G$ of $n$ dimensional Minkowski spacetime we define, following 
usual convention, $J^\pm(G)$ as the causal
future($+$)/past($-$) set of $G$, defined as consis\-ting of all points that can 
be reached from $G$ by smooth
future/past directed causal curves. We say that an open subset $G$ of $n$ 
dimensional Minkowski spacetime is
{\it hyperbolic} if for each pair of points $x,y\in G$ the set $J^+(x)\cap 
J^-(y)$ is a subset of $G$.
Examples of hyperbolic subsets are neighbourhoods $G$ of $\Sigma_t$ of the form 
$G=\{(x^0,x^1,\ldots,x^s):t_+>x^0>t_-\}$
where $t_+>t$ and $t_-<t$, or sets $G$ of the form $G=\text{int}(J^+(x)\cap J^-
(y))$ where $y$ lies in the open interior
of $J^+(x)$.

Now we collect some well-known results (well-known mainly in the context of the 
quantized Dirac field on curved
spacetimes) on the existence and uniqueness of advanced and retarded fundamental 
solutions for the Dirac operator $D_V$.

\begin{proposition}[\cite{Dimock},\cite{BGP}]\label{Prop1}
\hfill
\begin{enumerate}
\item\label{I_a_Prop1}
$\langle D_Vf,h\rangle=\langle f,D_Vh\rangle\quad(f,h\in 
C^\infty_0(\bR^n,\bC^N))$
\item
There is a unique pair of linear maps
\[
R^\pm_V:C^\infty_0(\bR^n,\bC^N)\rightarrow C^\infty(\bR^n,\bC^N)
\]
having the properties
\begin{eqnarray*}
&&\quad\quad D_VR^\pm_Vf=f=R^\pm_VD_Vf\text{ and }\nonumber\\
&&\supp R^\pm_Vf \subset J^\pm(\supp f)\quad(f\in C^\infty_0(\bR^n,\bC^N)).
\end{eqnarray*}
$R^\pm_V$ is called {\em advanced}(+)/{\em retarded}(-) 
{\em fundamental solution} of $D_V$.
\item\label{I_c_Prop1}
$CR^\pm_V=R^\pm_VC$
\item
Writing $R_V=R^+_V-R^-_V$, the form
\begin{equation} \label{E_propagatorV}
(f,h)_V=\langle f,iR_Vh\rangle
\end{equation}
is a sesquilinear form on $C^\infty_0(\bR^n,\bC^N)$, and $C$ is a conjugation 
for this form:
\begin{equation*}
(Cf,Ch)_V=(h,f)_V=\overline{(f,h)_V},\quad (f,h\in C^\infty_0(\bR^n,\bC^N)).
\end{equation*}
\item\label{I_e_Prop1}
For each $t\in\bR$ it holds that
\begin{equation*}
(f,h)_V=(P_tR_Vf,P_tR_Vh)_{\cD},\quad (f,h\in C^\infty_0(\bR^n,\bC^N)),
\end{equation*}
where $P_t:C^\infty(\bR^n,\bC^N)\rightarrow C^\infty(\Sigma_t,\bC^N)$ is the map 
given by
\begin{equation*}
P_t:\varphi\mapsto \varphi(t,\cdot)
\end{equation*}
for $\varphi:(x^0,\ulx)\mapsto \varphi(x^0,\ulx)$ in $C^\infty(\bR^n,\bC^N)$, 
$\ulx=(x^1,\ldots,x^s)$. Hence, $(\cdot,\cdot)_V$
is positive-semidefinite on $C^\infty_0(\bR^n,\bC^N)$.
\item The Cauchy-problem for the Dirac-equation $D_V\varphi = 0$ is
well-posed: Given any Cauchy-hyperplane $\Sigma_t$ and Cauchy-data $w \in
\mathscr{S}(\Sigma_t,\mathbb{C}^N)$, there is a unique 
$\varphi \in C^\infty(\mathbb{R}^n,\mathbb{C}^N)$ such that
$$ D_V \varphi = 0 \quad \text{and} \quad P_t\varphi = \varphi|_{\Sigma_t} = w\,.$$
Furthermore, the solution $\varphi$ fulfills the causal propagation property in the sense
that
$$ {\rm supp} \, \varphi \subset J({\rm supp}\,w) \,.$$
\item\label{I_f_Prop1}
Let $E_V$ be the subspace of all $f\in C^\infty_0(\bR^n,\bC^N)$ so that 
$(f,f)_V=0$, and let $\cK_V$ be the Hilbert space
arising as completion of $C^\infty_0(\bR^n,\bC^N)/E_V$ with respect to the 
scalar product induced by $(\cdot,\cdot)_V$ (which
will be denoted by the same symbol). The quotient map 
$C^\infty_0(\bR^n,\bC^N)\rightarrow C^\infty_0(\bR^n,\bC^N)/E_V$ will
be written
\begin{equation*}
f\mapsto [f]_V.
\end{equation*}
Then for each $t\in\bR$, the map
\begin{equation}
Q_{V,t}:[f]_V\mapsto P_tR_Vf
\end{equation}
extends to a unitary map from $\cK_V$ onto $\cD_t$.
\item\label{I_g_Prop1}
Let $G$ be a hyperbolic subset of $n$ dimensional Minkowski spacetime, and 
suppose that $V_1$ and $V_2$ are real-valued,
$C^\infty$, and that $V_1=V_2$ on $G$. Then
\begin{equation}
R^\pm_{V_1}f=R^\pm_{V_2}f\text{ on }G\text{ for all }f\in C^\infty_0(G,\bC^N).
\end{equation}
\end{enumerate}
\end{proposition}

\begin{proof}[Sketch of proof]
\hfill
\begin{enumerate}
\item
This is a straightforward calculation.
\item
This is proved using the same argument as for Theorem 2.1 in~\cite{Dimock}, 
which applies also in the presence of a
real scalar potential $V$, together with the existence and uniqueness result for 
fundamental solutions of hyperbolic wave
operators, which can be found (in far greater generality than needed here) 
in~\cite{BGP}.
\item
This is a consequence of the uniqueness of the $R^\pm_V$ together with 
$CD_V=D_VC$.
\item
The only non-obvious part $(Cf,Ch)_V=(h,f)_V$ of the claim follows easily from 
(\ref{I_c_Prop1}),
equation~(\ref{E_CSkewForSesqu}) and the relation
\[
\langle R_Vf,h\rangle=-\langle f,R_Vh\rangle,
\]
which is shown in the proof of Theorem 2.1 in~\cite{Dimock}.
\item
The argument is the same as for Proposition 2.4 (d) in~\cite{Dimock}.
\item
This statement is proved analogously to Thm.\ 2.3 in \cite{Dimock}. It is
proved there for the case that the Cauchy-data are $C^\infty_0$. However,
existence and uniqueness of a distributional solution is proved in Prop.\ 2.4
in \cite{Dimock} for distributional Cauchy-data. The smoothness of the solution
in case of Cauchy-data that are of Schwartz type can be proved by making use of the
causal propagation property of the solutions (i.e.\ ${\rm supp}\,\varphi 
\subset J({\rm supp}\,w)$) in combination with a partition of unity argument.
\item
In view of~(\ref{I_e_Prop1}), what remains to be checked is the surjectivity of 
$Q_{V,t}$. To see this, let $w \in C_0^\infty(\Sigma_t,\mathbb{R}^N)$, and let
$\varphi \in C^\infty(\mathbb{R}^n,\mathbb{C}^N)$ be the solution of $D_V\varphi=0$
having Cauchy-data $w$ on $\Sigma_t$, i.e.\ $P_t \varphi = w$. We will construct
some $f \in C_0^\infty(\mathbb{R}^n,\mathbb{R}^N)$ so that
$P_t R_V f = w$. To this end, we take two further Cauchy-hyperplanes, $\Sigma_\pm$,
with 
$$ \Sigma_\pm = \{ x = (x^0,\ldots,x^s) : x^0 = t \pm 1\} \,. $$
Then we can consider the open sets
$$ G^{\pm} = {\rm int}(J^\pm(\Sigma_\mp)) = \{ x = (x^0,\ldots,x^s) : \pm x^0 > t \mp 1\}\,.$$
The sets $G^\pm$ form an open covering of $\mathbb{R}^n$. Let $\chi_\pm$ be a $C^\infty$
partition of unity of $\mathbb{R}^n$ subordinate to the covering. It is easy to see that 
the functions $\chi_\pm$ can be chosen in such a way that they depend only on $x^0$, and
we will assume that this choice has been made (although this is not relevant at this point; see
however the proof of Prop.\ \ref{Prop1forMoyal} (g) later). Then one has 
$$ D_V(\chi_+ \varphi) = - D_V(\chi_- \varphi) \,,$$
and owing to the support properties of $\chi_\pm$, one concludes that both $D_V(\chi_\pm \varphi)$
have support contained in $\overline{G^+ \cap G^-} = \{ (x^0,\ldots,x^s): t + 1 \ge x^0 \ge t - 1\}$.
One the other hand, since we also have ${\rm supp}\,\varphi \subset J({\supp}\,w)$ and since ${\rm supp}\, w$ was
assumed to be compact, this implies that both $D_V(\chi_\pm \varphi)$ are $C_0^\infty$.
Setting now $f = D_V(\chi_+ \varphi)$, it holds that $f \in C_0^\infty(\mathbb{R}^n,\mathbb{C}^N)$,
and moreover, we see that $D_V(R_V f - \varphi) = 0$. However, we also have that
$R_V f = R_V^+f - R_V^-f$, and owing the support properties of $R_V^\pm$, on $\mathbb{R}^n \backslash G^- =
\{ (x^0,\ldots,x^s) : x^0 > t +1\}$ it holds that $R_Vf = R_V^+f = \chi_+ \varphi = \varphi$. This means
that $P_\tau(R_V f - \varphi) = 0$ for all real $\tau > t +1$ and hence, since $(R_Vf -\varphi)$ is a $C^\infty$ 
solution of the Dirac equation with $C_0^\infty$ Cauchy-data, one actually concludes that 
$ (R_V f - \varphi) = 0$ on all of $\mathbb{R}^n$. Hence we have shown that there is some $f \in C_0^\infty(\mathbb{R}^n,
\mathbb{C}^N)$ with $R_Vf = \varphi$, implying $P_t R_Vf = P_t\varphi$. This shows that the range of 
$Q_{V,t}$ is dense, and by its isometric property, $Q_{V,t}$ is actually surjective.    
\item
The spacetime region $G$, endowed with the Minkowski metric (and standard spin 
structure) is a globally hyperbolic spacetime. Given a smooth real-valued $V:G\rightarrow \bR$ as potential function, 
one can define the ``intrinsic''
Dirac operator of $G$, $D_{V|G}:C^\infty(G,\bC^N)\rightarrow C^\infty(G,\bC^N)$ 
by $D_{V|G}f=D_Vf$,
$f\in C^\infty(G,\bC^N)$, which is nothing but the canonical restriction of 
$D_V$ onto $G$. According to~\cite{Dimock}
(cf.\ also~\cite{BGP}), there are unique advanced/retarded fundamental solutions
$R^\pm_{V|G}:C^\infty_0(G,\bC^N)\rightarrow C^\infty(G,\bC^N)$ for $D_{V|G}$. 
Now, if $V_1=V_2=V$ on $G$, then the
appropriate restrictions of $R^\pm_{V_1}$ and $R^\pm_{V_2}$ onto 
$C^\infty_0(G,\bC^N)$ (more precisely, the maps
$f\mapsto \left.(R^\pm_{V_j})\right|_G$, $f\in C^\infty_0(G,\bC^N)$, $j=1,2$) 
have the same properties as the map
$R^\pm_{V|G}$. Hence, by the uniqueness statement, these restrictions must be 
equal to $R^\pm_{V|G}$.\qed
\end{enumerate}
\end{proof}

Starting from $(\cK_V,C)$, the Hilbert space $\cK_V$ with conjugation $C$, one 
can form, following~\cite{Araki}, the
corresponding self-dual CAR-algebra $\fF(\cK_V,C)$. It is defined as follows: 
One introduces a $*$-algebra generated by
symbols $B(\xi)=B_V(\xi)$, $\xi\in\cK_V$, subject to the relations
\begin{eqnarray*}
B(\xi)^*&=&B(C\xi),\nonumber\\
B(\xi_1)^*B(\xi_2)+B(\xi_2)B(\xi_1)^*&=&2(\xi_1,\xi_2)_V\eins,\nonumber\\
\xi&\mapsto &B(\xi)\text{ is complex linear},
\end{eqnarray*}
where $\eins$ is an algebraic unit. One can show that the resulting $*$-algebra 
admits a unique $C^*$-norm, and
$\fF(\cK_V,C)$ is the completion of that $*$-algebra with respect to the $C^*$-
norm. Therefore, $\fF(\cK_V,C)$ is a
$C^*$-algebra. Wri\-ting $\Psi(f)=\Psi_V(f)=B_V([f]_V)$ for $f\in 
C^\infty_0(\bR^n,\bC^N)$, $\fF(\cK_V,C)$ is generated
by ``abstract field operators'' $\Psi(f)$, which are $\bC$-linear and obey the 
relations
\begin{eqnarray*}
\Psi(f)^*&=&\Psi(Cf),\nonumber\\
\{\Psi(f)^*,\Psi(h)\}&=&2(f,h)_V\eins,\nonumber\\
\Psi(D_Vf)&=&0.
\end{eqnarray*}
The construction of $\fF(\cK_V,C)$ can also be carried out, in analogous manner, 
for ``local subspaces'' of $\cK_V$.
For this purpose, let $G$ be a hyperbolic subset of $n$ dimensional Minkowski 
spacetime. For $f\in C^\infty_0(G,\bC^N)$,
we introduce the equivalence class
\begin{equation*}
[f]^G_V=\{f+h:h\in C^\infty_0(G,\bC^N),R_Vh=0\}.
\end{equation*}
As before, the space of the $[f]^G_V$ carries a scalar product $(\cdot,\cdot)_V$ 
in the same fashion as
$C^\infty_0(\bR^n,\bC^N)/E_V$ (denoted by the same symbol as there is no danger 
of confusion). The resulting Hilbert-space
completion will be denoted by $\cK^G_V$. Again, the charge conjugation $C$ 
induces a conjugation on $\cK^G_V$ as well.
Whence, we can form the self-dual CAR-algebra $\fF(\cK^G_V,C)$, which is the 
$C^*$-algebra
ge\-nerated by symbols $B^G_V([f]^G_V)$, $[f]^G_V\in\cK^G_V$, obeying relations 
akin to those fulfilled by the $B([f]_V)$
above.

\begin{lemma}\label{Lemma2}
Suppose that $G$ is a hyperbolic neighbourhood of a Cauchy hyperplane in $n$ 
dimensional Minkowski spacetime. Moreover,
suppose that $V_1$ and $V_2$ are two smooth, real-valued potentials which 
coincide on the region $G$. Then
\begin{enumerate}
\item
The map
\begin{equation*}
u^G_{V_1,V_2}:[f]^G_{V_1}\mapsto [f]_{V_2},\quad f\in C^\infty_0(G,\bC^N)
\end{equation*}
extends to a unitary between $\cK^G_{V_1}$ and $\cK_{V_2}$ commuting with the 
charge conjugation $C$.
\item
There is a $*$-algebra isomorphism
\begin{equation*}
\alpha^G_{V_1,V_2}:\fF(\cK^G_{V_1},C)\rightarrow \fF(\cK_{V_2},C)
\end{equation*}
induced by
\begin{equation*}
\alpha^G_{V_1,V_2}\left(B^G_{V_1}([f]^G_{V_1})\right)=B_{V_2}([f]_{V_2}),\quad 
f\in C^\infty_0(G,\bC^N).
\end{equation*}
\end{enumerate}
\end{lemma}

\begin{proof}
\hfill
\begin{enumerate}
\item
\sloppy In view of~(\ref{I_g_Prop1}) of Proposition~\ref{Prop1}, 
$R^G_{V_1}f=R_{V_2}f$ on $G$ for all
$f\in C^\infty_0(G,\bC^N)$. Using the definition of $(\cdot,\cdot)_V$, this 
implies that the map $u^G_{V_1,V_2}$ is
isometric. To show that the map is surjective, let $h\in 
C^\infty_0(\bR^n,\bC^N)$. Since $G$ is an open neighbourhood of a
Cauchy surface, there is some $f\in C^\infty_0(G,\bC^N)$ such that $R_{V_2}(f-
h)=0$ (\cite{BGP}) and hence
$[f]_{V_2}=[h]_{V_2}$.
\item
This is a straightforward consequence of the fact that $u^G_{V_1,V_2}$ is a 
unitary intertwining the action of $C$,
see~\cite{Araki2} (or~\cite{BratteliRobinson1},\cite{BratteliRobinson2}).\qed
\end{enumerate}
\end{proof}

We will now make use of the Lemma. Suppose that a smooth scalar (real) potential 
$V$ is given on $n$ dimensional
Minkowski spacetime, having support contained in the time-slice 
$\{(x^0,x^1,\ldots,x^s):\lambda_-<x^0<\lambda_+\}$ for
some real numbers $\lambda_-<\lambda_+$. Then one can consider the regions
\begin{eqnarray*}
G_+&=&\{(x^0,x^1,\ldots,x^s):x^0>\lambda_++\frac{1}{2}\}\text{ and}\nonumber\\
G_-&=&\{(x^0,x^1,\ldots,x^s):x^0<\lambda_--\frac{1}{2}\}.
\end{eqnarray*}
They form hyperbolic neighbourhoods of the Cauchy hyperplanes 
\begin{eqnarray*}
\Sigma_+&=&\{(x^0,x^1,\ldots,x^s):x^0=\lambda_++1\}\text{ and}\nonumber\\
\Sigma_-&=&\{(x^0,x^1,\ldots,x^s):x^0=\lambda_--1\},
\end{eqnarray*}
respectively. Lemma~\ref{Lemma2} then warrants the following $C^*$-algebraic 
isomorphisms:
\begin{eqnarray*}
\alpha_{0\pm}=\alpha^{G_\pm}_{0,0}:\fF(\cK^{G_\pm}_0,C)&\rightarrow 
&\fF(\cK_0,C)\nonumber\\
B^{G_\pm}_0([f]^{G_\pm}_0)&\mapsto &B_0([f]_0),\quad f\in 
C^\infty_0(G_\pm,\bC^N),
\end{eqnarray*}
\begin{eqnarray*}
\alpha_{V\pm}=\alpha^{G_\pm}_{0,V}:\fF(\cK^{G_\pm}_0,C)&\rightarrow 
&\fF(\cK_V,C)\nonumber\\
B^{G_\pm}_0([f]^{G_\pm}_0)&\mapsto &B_V([f]_V),\quad f\in 
C^\infty_0(G_\pm,\bC^N).
\end{eqnarray*}
Since these maps are isomorphisms, they can be combined into an automorphism
\begin{eqnarray}
&&\beta_V:\fF(\cK_0,C)\rightarrow \fF(\cK_0,C),\nonumber\\
&&\beta_V=\alpha_{0,-}\circ \alpha_{V,-}^{-1}\circ \alpha_{V,+}\circ 
\alpha_{0,+}^{-1}.\label{E_DefBetaV}
\end{eqnarray}
This isomorphism is reminiscent of a similar object defined in Section 4 
of~\cite{BFV}, and it has similar properties. Its
significance is that it describes the scattering of the quantized Dirac field by 
the classical potential $V$ at the level of
a $C^*$-algebraic Bogoliubov transformation. In order to see this more clearly, 
we will discuss how $\beta_V$ relates to the
perhaps more familiar scattering formalism in terms of time-evolution on the 
Cauchy data.

For this purpose, let us first revisit $\beta_V$. We have
\begin{equation}\label{E_betaVUVf0}
\beta_V(B_0([f]_0))=B_0(U_V[f]_0),
\end{equation}
where $U_V$ is the unitary given by
\begin{equation*}
U_V=u_{0,-}\circ u_{V,-}^{-1}\circ u_{V,+}\circ u_{0,+}^{-1}
\end{equation*}
and where, similarly as for the isomorphisms above, we have used the 
abbreviations
\begin{equation*}
u_{0,\pm}=u^{G_\pm}_{0,0},\quad u_{V,\pm}=u^{G_\pm}_{0,V}.
\end{equation*}
The action of the succession of unitaries on the right hand side of the defining 
equation of $U_V$ can be described as
follows:
\begin{equation}\label{E_fSuccUnitar}
\xymatrix{
[f]_0 \ar@{|->}[r]^{u_{0,+}^{-1}}& [f^{G_+}]^{G_+}_0 \ar@{|->}[r]^{u_{V,+}}& 
[f^{G_+}]_V
\ar@{|->}[r]^{u_{V,-}^{-1}}& [f^{G_-}]^{G_-}_0 \ar@{|->}[r]^{u_{0,-}}& [f^{G_-
}]_0
}
\end{equation}
In this chain of mappings, $f^{G_+}$ is any element in $C^\infty_0(G_+,\bC^N)$ 
such that $R_0(f-f^{G_+})=0$, and $f^{G_-}$
is any element in $C^\infty_0(G_-,\bC^N)$ such that $R_V(f^{G_+}-f^{G_-})=0$.\\

Turning to the description of the quantized Dirac field in terms of its Cauchy 
data, we recall that
$\cD_0=L^2(\Sigma_0,d^sx)$, where $\Sigma_0$ is the $x^0=0$ Cauchy hyperplane. 
We have also seen that the charge conjugation
$C$ acts as a complex conjugation on $\cD_0$. Hence one can associate to $\cD_0$ 
and $C$ the CAR-algebra $\fF(\cD_0,C)$ with
generators $B_{\cD_0}(v)$, $v\in\cD_0$, linear in $v$, and with the relations 
\begin{equation} \label{E_CAR_D_0}
B_{\cD_0}(v)^*=B_{\cD_0}(Cv)\,, \quad
\{B_{\cD_0}(v)^*,B_{\cD_0}(w)\}=2(v,w)_{\cD}\eins\,.
\end{equation}
 Writing $Q_0$ for $Q_{V,0}$ in the case of $V=0$,
$Q_0:\cK_0\rightarrow \cD_0$, $[f]_0\mapsto P_0R_0f$ is a unitary intertwining 
the actions of $C$ on the respective
Hilbert spaces. Consequently~(\cite{Araki}), there is a canonical isomorphism 
$\varrho:\fF(\cK_0,C)\rightarrow \fF(\cD_0,C)$ of
CAR-algebras induced by
\begin{equation}\label{E_varrhoDef}
\varrho(B_0([f]_0))=B_{\cD_0}(Q_0([f]_0)).
\end{equation}
On $\fF(\cD_0,C)$, we can introduce two types of time evolutions, one 
corresponding to a vanishing potential $V=0$ in the
Dirac equation (the ``free'' dynamics), and another corresponding to a non-
vanishing $C^\infty$ potential term $V$ in the
Dirac equation (the ``interacting'' dynamics). These dynamical evolutions will 
be defined on the Cauchy-data space $\cD_0$.
To this end, we define on $C^\infty_0(\bR^s,\bC^N)$ the operators
\begin{eqnarray}
(H_0f)(x^1,\ldots,x^s)&=&\left(i\gamma^0\gamma^k\frac{\partial}{\partial 
x^k}+\gamma^0m\right)f(x^1,\ldots,x^s)\label{E_DefH0}\\
(H_V(t)f)(x^1,\ldots,x^s)&=&\left(i\gamma^0\gamma^k\frac{\partial}{\partial 
x^k}+\gamma^0m+\gamma^0V(t)\right)f(x^1,\ldots,x^s),\nonumber
\end{eqnarray}
where $f(x^1,\ldots,x^s)$ is regarded as column vector on which the $\gamma$-
matrices act by matrix multiplication. These
operators are symmetric with respect to the scalar product $(\cdot,\cdot)_{\cD}$, 
and even essentially selfadjoint under very
general conditions on (the real-valued) $V(t)$ (e.g. see Theorem 1.1 
of~\cite{Thaller}, and Theorem X.69 of~\cite{ReedSimon2}).
Moreover, it is easy to check that the operators anti-commute with the charge 
conjugation $C$,
\begin{equation}\label{E_CH0HV}
CH_0=-H_0C,\quad CH_V(t)=-H_V(t)C.
\end{equation}
There is hence a continuous unitary group $T_t$, $t\in\bR$, on $\cD_0$ such that
\begin{equation*}
\left.\frac{1}{i}\frac{d}{dt}T_t\right|_{t=0}v=H_0v,\quad v\in 
C^\infty_0(\bR^s,\bC^N).
\end{equation*}
There is also a continuous family of unitarities $T^{(V)}_{s,t}$, $s,t\in\bR$, so 
that
\begin{eqnarray}
T^{(V)}_{r,s}\circ T^{(V)}_{s,t}&=&T^{(V)}_{r,t},\quad 
T^{(V)}_{t,t}=\eins\quad\text{and}\nonumber\\
\left.\frac{1}{i}\frac{d}{ds}\right|_{s=t}T^{(V)}_{s,t}v&=&H_V(t)v,\quad v\in 
C^\infty_0(\bR^s,\bC^N).\label{E_EvolEqTst}
\end{eqnarray}
Let us indicate that the existence of the family $T^{(V)}_{s,t}$ with the said 
properties is implied by the well-posedness of
the Cauchy problem for the Dirac equation: For each $v\in 
C^\infty_0(\bR^s,\bC^N)\subset \cD_t$ there is a unique solution
$\varphi\in C^\infty(\bR^n,\bC^N)$ to the Dirac equation
\begin{equation*}
D_V\varphi=0
\end{equation*}
having initial data $v$ on $\Sigma_t$, i.e.
\begin{equation*}
P_t\varphi=\varphi|_{\Sigma_t}=v.
\end{equation*}
The solution property is equivalent to
\begin{equation}\label{E_PtvarphiHV}
\frac{1}{i}\frac{d}{dt}P_t\varphi=H_V(t)P_t\varphi.
\end{equation}
On the other hand, the uniqueness statement implies that there is a map 
$T^{(V)}_{t,t'}:P_{t'}\varphi\mapsto P_t\varphi$
with the properties $T^{(V)}_{t,t'}\circ T^{(V)}_{t',t''}=T^{(V)}_{t,t''}$ and 
$T^{(V)}_{t,t}=\eins$. And
$\left.\frac{1}{i}\frac{d}{ds}\right|_{s=t}T^{(V)}_{s,t}=H_V(t)$ on 
$C^\infty_0(\bR^s,\bC^N)$
then follows from~(\ref{E_PtvarphiHV}). The unitarity of $T^{(V)}_{t,t'}$ is 
implied by Proposition~\ref{Prop1}~(\ref{I_e_Prop1}).
We note also that
\begin{equation}\label{E_CT_ttCommute}
CT_t=T_tC\text{ and }CT^{(V)}_{t,t'}=T^{(V)}_{t,t'}C
\end{equation}
on account of~(\ref{E_CH0HV}). Therefore, $T_t$ and $T^{(V)}_{t,t'}$ give rise 
to CAR-algebra automorphisms $\tau_t$ and
$\tau^{(V)}_{t,t'}$ of $\fF(\cD_0,C)$ induced by
\begin{eqnarray*}
\tau_t(B_{\cD_0}(v))&=&B_{\cD_0}(T_tv),\nonumber\\
\tau^{(V)}_{t,t'}(B_{\cD_0}(v))&=&B_{\cD_0}(T^{(V)}_{t,t'}v).
\end{eqnarray*}
As before, we will now assume that the potential $V\in C^\infty(\bR^n,\bR)$ has 
support contained in the set
$\{x^0,x^1,\ldots,x^s):\lambda_-<x^0<\lambda_+\}$ for some real numbers 
$\lambda_-<\lambda_+$. The
{\em scattering operator} for the Dirac equation at the level of Cauchy 
data on $\Sigma_0$ is the operator
\begin{equation}\label{E_TscLim}
T^{(V)}_{sc}=\lim_{\stackrel{t'\rightarrow \infty}{\scriptscriptstyle 
t\rightarrow -\infty}}T_t^{-1}
\circ T^{(V)}_{t,t'}\circ T_{t'}
\end{equation}
on $\cD_0$. The restriction on the time-support of $V$ implies that the 
limit~(\ref{E_TscLim}) is reached as soon as
$t'>\lambda_+$ and $t<\lambda_-$, so that
\begin{equation}\label{E_Tsclambdapm}
T^{(V)}_{sc}=T_t^{-1}\circ T^{(V)}_{t,t'}\circ 
T_{t'},\quad\text{for }t'>\lambda_+,t<\lambda_-.
\end{equation}
We denote by $\tau^{(V)}_{sc}$ the corresponding {\em scattering 
morphism} on $\fF(\cD_0,C)$ given by
\begin{equation}\label{E_DeftauScMor}
\tau^{(V)}_{sc}(B_{\cD_0}(v))=B_{\cD_0}(T^{(V)}_{sc}v).
\end{equation}
Now we want to demonstrate that
\begin{equation}\label{E_GammaBetaVTauSc}
\varrho\circ \beta_V=\tau^{(V)}_\text{sc}\circ \varrho.
\end{equation}
Thus we aim at showing
\begin{equation}\label{E_Q0UV}
Q_0\circ U_V=T^{(V)}_{sc}\circ Q_0.
\end{equation}
To prove this, we write the action of $Q_0^{-1}\circ T^{(V)}_\text{sc}\circ Q_0$ 
on an element $[f]_0\in\cK_0$ in the
following form:
\begin{equation}\label{E_fSuccUnitar2}
\xymatrix{
[f]_0 \ar@{|->}[r]^{Q_0}& P_0R_0f \ar@{|->}[r]^{T_{t'}}& P_{t'}R_0f \ar@{|->}[r]^{(*1)}& P_{t'}R_0h^{G_+}\\
\ \ \ \ \ar@{|->}[r]^{(*2)}& P_{t'}R_Vh^{G_+}
\ar@{|->}[r]^{T^{(V)}_{t,t'}}& P_tR_Vh^{G_+}
\ar@{|->}[r]^{(*3)}& P_tR_Vh^{G_-}\\
\ \ \ \ \ar@{|->}[r]^{(*4)}& P_tR_0h^{G_-} \ar@{|->}[r]^{T_t^{-1}}& P_0R_0h^{G_-}
\ar@{|->}[r]^{Q_0^{-1}}& [h^{G_-}]_0
}
\end{equation}
In this succession of maps, at $(*1)$ an element $h^{G_+}\in 
C^\infty_0(G_+,\bC^N)$ is chosen so that $R_0h^{G_+}=R_0f$.
At $(*2)$, it is used that $R_Vh^{G_+}=R_0h^{G_+}$ on $G_+$ because of the 
support properties of the functions $V$ and
$h^{G_+}$, cf. Proposition~\ref{Prop1}~(\ref{I_g_Prop1}). At $(*3)$, an element 
$h^{G_-}\in C^\infty_0(G_-,\bC^N)$ is
chosen so that $R_Vh^{G_+}=R_Vh^{G_-}$. Then at $(*4)$, it is again used that 
$R_0h^{G_-}=R_Vh^{G_-}$ on $G_-$ owing to
the support properties of $V$ and $h^{G_-}$. Comparing~(\ref{E_fSuccUnitar}) 
and~(\ref{E_fSuccUnitar2}), one can see that
the specifications of $f^{G_\pm}$ and $h^{G_\pm}$ are such that one can may even 
choose (starting from the same given $f$)
$f^{G_\pm}=h^{G_\pm}$, and this then proves the relation~(\ref{E_Q0UV}). 
Summarizing, we have proved

\begin{lemma}\label{L_betaV_tauscV_intertw}
The morphism $\beta_V$ of $\fF(\cK_0,C)$ defined in~(\ref{E_DefBetaV}) and the 
scattering morphism $\tau^{(V)}_{sc}$
describing the potential scatte\-ring of the quantized Dirac field at the level 
of the Cauchy-data CAR-algebra
$\fF(\cD_0,C)$ are intertwined by the CAR-algebra isomorphism 
$\varrho:\fF(\cK_0,C)\rightarrow \fF(\cD_0,C)$ defined
in~(\ref{E_varrhoDef}), i.e. it holds that
\begin{equation*}
\varrho\circ \beta_V=\tau^{(V)}_\text{sc}\circ \varrho.
\end{equation*}
\end{lemma}

One advantage of working with $\beta_V$ is that it can be associated to 
localization in spacetime: It acts trivially
outside of $J(\supp V)$. This is our next assertion.

\begin{proposition}
Let $f\in C^\infty_0(\bR^n,\bC^N)$ have support causally disjoint from $\supp V$, 
i.e. $\supp f\cap J(\supp V)=\emptyset$.
Then
\begin{equation*}
\beta_V(\Psi_0(f))=\Psi_0(f).
\end{equation*}
\end{proposition}

\begin{proof}
According to Proposition~\ref{Prop1}~(\ref{I_g_Prop1}), if $\supp f\cap J(\supp 
V)=\emptyset$, then $R_0f=R_Vf$ on
$\bR^n\setminus J(\supp V)$, since $\bR^n\setminus J(\supp V)$ is a hyperbolic 
region in $n$ dimensional Minkowski spacetime.
On the other hand, $\supp f\cap J(\supp V)=\emptyset$ is equivalent to $J(\supp 
f)\cap \supp V=\emptyset$. Now, $R_Vf$ is a
solution to $(D+V)R_Vf=0$, and $\supp R_Vf\subset J(\supp f)$, thus $\supp 
R_Vf\cap\supp V=\emptyset$, implying that $DR_Vf=0$.
Consequently, $R_0f$ and $R_Vf$ are both solutions to the Dirac equation with 
vanishing potential $V=0$, and coincide in the
neighbourhood of a Cauchy surface for $n$ dimensional Minkowski spacetime (which 
is implied by $R_0f=R_Vf$ on
$\bR^n\setminus J(\supp V)$ and $\supp R_0f\cup \supp R_Vf\subset J(\supp f)$). 
This implies that $R_0f=R_Vf$ on all of $n$
dimensional Minkowski spacetime. Now consider the map
\begin{equation}\label{E_fSuccUnitar_1}
\xymatrix{
U_V:[f]_0 \ar@{|->}[r]& [f^{G_+}]^{G_+}_0 \ar@{|->}[r]& [f^{G_+}]_V \ar@{|->}[r]& [f^{G_-}]^{G_-}_0 \ar@{|->}[r]& [f^{G_-}]_0
}.
\end{equation}
In this succession of mappings, $f^{G_+}$ is any element in 
$C^\infty_0(G_+,\bC^N)$ so that $R_0(f-f^{G_+})=0$, and $f^{G_-}$
is any element in $C^\infty_0(G_-,\bC^N)$ so that $R_V(f^{G_+}-f^{G_-})=0$. 
However, since $R_0f=R_Vf$, it holds that
$R_V(f-f^{G_+})=R_0f-R_Vf^{G_+}=R_0f-R_0f^{G_+}=0$ on $G_+$, hence $R_V(f-
f^{G_+})=0$ on $G_+$, and hence $R_V(f-f^{G_+})=0$
everywhere on $n$ dimensional Minkowski spacetime. Furthermore, $R_Vf^{G_-
}=R_Vf^{G_+}$, from which $R_V(f-f^{G_-})=0$ obtains.
On the other hand, we also have $R_Vf^{G_-}=R_0f^{G_-}$ on $G_-$, and 
$R_Vf=R_0f$, thus $R_0f=R_0f^{G_-}$ on $G_-$, and hence
everywhere on $n$ dimensional Minkowski spacetime. This shows that 
$[f]_0=[f^{G_-}]_0$ and therefore, $U_V[f]_0=[f]_0$. In view
of~(\ref{E_betaVUVf0}), this yields the claimed proposition.\qed
\end{proof}

\section{Scattering of the Dirac field in the vacuum representation and 
implementability of the scattering transformation}
\label{SectScattMink}

\sloppy The Hamilton operator $H_0$ defined in~(\ref{E_DefH0}) is essentially 
selfadjoint on $C^\infty_0(\bR^s,\bC^N)
\subset L^2(\bR^s,\bC^N)$ (see Theorem 1.1 of~\cite{Thaller}). Therefore, its 
selfadjoint extension, again denoted by $H_0$,
possesses a spectral decomposition, and we denote by $p_+$ the spectral 
projection of $H_0$ corresponding to the spectral
interval $(0,\infty)$. Since the mass term $m$ in the Dirac equation has been 
assumed to be strictly greater than $0$, $p_+$
projects in fact on the spectral values in $[m,\infty)$ and the orthogonal 
projector $p_-=\eins-p_+$ coincides with the spectral
projector of the spectral interval $(-\infty,-m]$. Owing to $CT_t=T_tC$ for all 
$t\in\bR$, it holds that
\begin{equation*}
Cp_+=p_-C.
\end{equation*}
Thus, $p_+$ is a basis projection in the sense of~\cite{Araki}. To this basis 
projection one can associate a pure, quasifree
state $\omega^{p_+}$ on $\fF(\cD_0,C)$ whose two-point function is given by
\begin{equation*}
\omega_2^{p_+}(B_{\cD_0}(u)^*B_{\cD_0}(w))=(u,p_+w)_\cD,\quad u,w\in\cD_0.
\end{equation*}
The state can be pulled back by $\varrho$ to a pure, quasifree state
\begin{equation*}
\omega^\text{vac}=\omega^{p_+}\circ \varrho
\end{equation*}
on $\fF(\cK_0,C)$. This state is actually just the usual ($\tilde\cP_+^\uparrow 
(n)$-invariant) vacuum state on $\fF(\cK_0,C)$.
Writing
\begin{equation*}
e_+=Q_0^{-1}p_+Q_0,
\end{equation*}
its GNS-representation $(\cH^\text{vac},\pi^\text{vac},\Omega^\text{vac})$ can 
be realized as follows:
\begin{equation*}
\cH^\text{vac}=\mathcal{F}_+(e_+(\cK_0)),
\end{equation*}
is the Fermionic Fock space over the one-particle Hilbert space $e_+(\cK_0)$ 
($e_+$ projects on the ``positive frequency''
solutions of the Dirac equation), $\Omega^\text{vac}=(1,0,0,\ldots)$ the Fock 
vacuum vector,
\begin{equation*}
\pi^\text{vac}(\Psi_0(f))= 
A(e_+C[f]_0)+A^+(e_+[f]_0)
,
\end{equation*}
where $A(\chi)$ and $A^+(\chi)$ denote, respectively, the Fermionic annihilation 
and creation operators of a $\chi$ in the
one-particle Hilbert space. We will sometimes use the notation
\begin{equation*}
{\bm\psi}(f)=\pi^\text{vac}(\Psi_0(f)),\quad f\in C^\infty_0(\bR^n,\bC^N),
\end{equation*}
for the field operators of the quantized Dirac field in the vacuum 
representation.

For several reasons, it is important to investigate the question of unitary 
implementability of the scattering transformation
in the vacuum representation. In the situation at hand, this is the question if 
there exists a unitary operator
$S_V:\cH^\text{vac}\rightarrow \cH^\text{vac}$ such that
\begin{equation}\label{E_ExUnitarSV}
S_V\pi^\text{vac}(\Psi_0(f))S_V^{-1}=\pi^\text{vac}(\beta_V(\Psi_0(f))),\quad 
f\in C^\infty_0(\bR^n,\bC^N).
\end{equation}
This issue has been investigated for the Dirac field on Minkowski spacetime by 
several authors in various publications that
have appeared over the last decades. The result is that there is such an 
operator, or ``S-matrix'', provided that the
potential $V$ is sufficiently regular and sufficiently fast decaying. A 
sufficient condition to this end, which is convenient
for comparison with developments presented later in this article, is the 
following

\begin{proposition}
If $V$ is in $\sS(\bR^n,\bR)$ (the class of Schwartz functions) and if $V$ has 
compact support with respect to the
time-coordinate $x^0$, then there is a unitary operator $S_V$ on 
$\cH^\text{vac}$ implementing the potential scattering
morphism $\beta_V$ in the vacuum representation, i.e. 
relation~(\ref{E_ExUnitarSV}) holds.
\end{proposition}

This is, however, a very specialized version of results which have been obtained 
previously. We make no attempt to review
these results here, but mention the following. It is quite obvious that one may 
generalize the result by dropping the compact
support of $V$ in time, relaxing the smoothness requirement and replacing the 
rapid decay conditions by suitable conditions
of integrability. Furthermore, one can generalize $V$ to a matrix-valued 
function as long as the resulting Hamilton operator
$H_V(t)$ remains essentially selfadjoint and still fulfills 
\begin{equation*}
CH_V(t)=-H_V(t)C.
\end{equation*}
Generalizations of this type have been considered by Palmer~\cite{Palmer}, and 
he has found that the S-matrix $S_V$ implementing
the scattering transformation exists, if $\|\partial^\alpha_t\hat 
V(t,\cdot)\|_{L^q(\bR^s)}$ is integrable over $t\in\bR$ for
all $1\leq q<2+\varepsilon$ and for all $0\leq \alpha<s/2+\varepsilon$. $\hat V$ 
denotes the Fourier transform of $V$ with
respect to the spatial variables $x^1,\ldots,x^s$. We refer to~\cite{Palmer} for 
further details, and also for references to
related, earlier work.

\section{Moyal Minkowski spacetime}\label{S_MoyalMinkST}

As it is usually introduced, $n=1+s$ dimensional Minkowski spacetime gets Moyal-deformed
if one postulates the following
commutation relations between the coordinates:
\begin{equation}\label{E_xmuxnucomm1}
[x_\mu,x_\nu]=i\theta_{\mu\nu}\quad(\mu,\nu=0,\ldots,s)
\end{equation}
with $\theta$ being some antisymmetric, real $(n\times n)$-matrix. Of course, 
this stems from the idea of generalizing
the behaviour of the quantum mechanical position operators $x_\mu$ originating 
in motivations like~\cite{DFR}
(restricting event localization by incorporating the uncertainty principle in 
general relativity)
and~\cite{Szabo} (string theory). Alternatively one can implement the 
relations~(\ref{E_xmuxnucomm1}) by
changing the product structure on the spacetime manifold such that
\begin{equation*}
[x_\mu,x_\nu]_\star=x_\mu\star x_\nu-x_\nu\star 
x_\mu=i\theta_{\mu\nu}\quad(\mu,\nu=0,\ldots,s)
\end{equation*}
is fulfilled between the coordinate chart functions $x_\mu$ of the manifold. 
Thereby $\star$ is the non-commutative
Moyal product. But let us make this last point more precise now.

Let $q,p\in\bN_0$, with $p=2l$ for $l\in\bN_0$, and let $\theta>0$. Then we 
define the $(q+p)\times(q+p)$-matrix
\begin{equation*}
M=M_\theta=\frac{\theta}{2}
\left[ 
\begin{array}{ccc|ccc}
&&&&&\\
&0_{q\times q}&&&0_{q\times p}& \\
&&&&&\\
\hline 
&&&0_{l\times l}&&\eins_{l\times l}\\ 
&0_{p\times q}&&&&\\
&&&-\eins_{l\times l}&&0_{l\times l}\\ 
\end{array}
\right] 
\end{equation*}
having the $2l\times 2l$-dimensional standard symplectic matrix in the lower 
right corner, and zeros everywhere else.
With this notation, we introduce the Moyal product
\begin{equation*}
c\star_{(q,p)}g(x)=\frac{1}{(2\pi)^{q+p}}\iint c(x-Mu)g(x+v)e^{-iu\cdot 
v}d^{q+p}ud^{q+p}v,\quad x\in\bR^{q+p},
\end{equation*}
for (complex-valued) Schwartz functions $c,g\in\sS(\bR^{q+p})$. By $u\cdot v$ we 
denote the standard Euclidean scalar product of vectors
$u,v\in\bR^{q+p}$. One can show, either directly or by adapting the arguments 
of~\cite{GGISV}, that $c\star_{(q,p)}g$ is
again in $\sS(\bR^{q+p})$ and that the product $c\star_{(q,p)}g$ is jointly 
continuous in $c$ and $g$ with respect to the
usual test-function topology on $\sS(\bR^{q+p})$.\\
In the case that $q=0$, $M=M_\theta$ is invertible, and then one has
\begin{equation*}
c\star_{(0,p)}g(x)=\frac{1}{(\pi\theta)^p}\iint c(x-u)g(x+v)e^{-iu\cdot M^{-
1}v}d^pud^pv,
\end{equation*}
which is the usual Moyal product investigated in several references 
(see~\cite{GGISV},\cite{GV}). In the other extreme
case, $p=0$, one finds
\begin{equation*}
c\star_{(q,0)}g(x)=\frac{1}{(2\pi)^q}\iint c(x)g(x+v)e^{-iu\cdot 
v}d^qud^qv=c(x)g(x),
\end{equation*}
i.e.\ the product $c\star_{(q,0)}g$ coincides with the usual pointwise product 
of functions.\\
In the general case, it is straightforward to check that
\begin{equation}\label{E_TensProdqp}
(c\otimes \varphi)\star_{(q,p)}(g\otimes 
\xi)=(c\star_{(q,0)}g)\otimes(\varphi\star_{(0,p)}\xi)
\end{equation}
for $c,g\in\sS(\bR^q)$ and $\varphi,\xi\in\sS(\bR^p)$. Together with the 
continuity of $\cdot\star_{(q,p)}\cdot$ in
both entries and the fact that $\sS(\bR^{q+p})=\sS(\bR^q)\otimes\sS(\bR^p)$ 
topologically, this shows that the product
$\star_{(q,p)}$ is associative and furnishes an algebra product on 
$\sS(\bR^{q+p})$, because these properties are known
for $\star_{(q,0)}$ and $\star_{(0,p)}$. Furthermore, the standard complex 
conjugation induces a $*$-involution on
$\sS(\bR^{q+p})$ with respect to the product $\star_{(q,p)}$. We denote this by 
$c\mapsto c^*=\bar c$. As a $*$-involution,
it has the property
\begin{equation*}
c^*\star_{(q,p)}g^*=(g\star_{(q,p)}c)^*.
\end{equation*}
With the algebra product $\star_{(q,p)}$ and the complex conjugation as a $*$-
involution, $\sS(\bR^{q+p})$ is turned into a
$*$-algebra which we denote by $\sS^M_{\star_{(q,p)}}$. By~(\ref{E_TensProdqp}), 
we have
\begin{equation}\label{E_qpMoyalTensProd}
\sS^M_{\star_{(q,p)}}=\sS^M_{\star_{(q,0)}}\otimes \sS^M_{\star_{(0,p)}},
\end{equation}
which holds also in the topological sense.\\
One can adapt the arguments in~\cite{GGISV} to observe that the product 
$\star_{(q,p)}$ can be extended to much larger spaces of
functions and even distributions. An important case is that one factor in 
$c\star_{(q,p)}g$ is in $\sS(\bR^{q+p})$ and the other
is in $L^2(\bR^{q+p})$. Again we consider this situation first for $q=0$. Using 
Lemma 2.12 of~\cite{GGISV} (resp. reference [43]
therein, which is~\cite{GV} here), it holds that $c\star_{(0,p)}g$ is in 
$L^2(\bR^p)$ if both $c$ and $g$ are in $L^2(\bR^p)$.
One can thus also define the operator of left Moyal multiplication on 
$L^2(\bR^p)$,
\begin{equation*}
{\rm L}_c:g\mapsto c\star_{(0,p)}g,\quad g\in L^2(\bR^p),
\end{equation*}
for $c\in L^2(\bR^p)$.  It is proved in~\cite{GV} that this operator is bounded, 
more precisely, that
\begin{equation}\label{E_LRMoyalBounded}
\|{\rm L}_c g\|_{L^2}\leq \frac{1}{(2\pi\theta)^{p/2}}\|c\|_{L^2}\|g\|_{L^2}.
\end{equation}
The same estimate holds then also for the operator of right multiplication by 
$c\in L^2(\bR^p)$ given by
\begin{equation*}
{\rm R}_c:g\mapsto g\star_{(0,p)}c,\quad g\in L^2(\bR^p),
\end{equation*}
since $\|\overline{{\rm R}_c g}\|_{L^2}=\|{\rm L}_{\bar c}\bar g\|_{L^2}$ and 
$\|\bar g\|_{L^2}=\|g\|_{L^2}$, where the overlining
denotes complex conjugation. For $p=0$, as $c\star_{(q,0)}g=c\cdot g=g\cdot 
c=g\star_{(q,0)}c$ is just the usual pointwise product
of functions, one has
\begin{equation*}
\|c\star_{(q,0)}g\|_{L^2}=\|g\star_{(q,0)}c\|_{L^2}\leq \|c\|_\infty\|g\|_{L^2},
\end{equation*}
where $\|\cdot\|_\infty$ is the supremum norm. This entails that for 
$c=c_q\otimes c_p$ with $c_q\in \sS(\bR^q)$ and
$c_p\in \sS(\bR^p)$, the operators
\begin{equation*}
{\rm L}_c:g\mapsto c\star_{(q,p)}g,\text{ and }{\rm R}_c:g\mapsto 
g\star_{(q,p)}c,\quad g\in L^2(\bR^{q+p}),
\end{equation*}
are bounded operators whose operator norms are not greater than 
$\frac{1}{(2\pi\theta)^{p/2}}\|c_q\|_\infty\|c_p\|_{L^2}$.
Since each $c\in\sS(\bR^{q+p})$ can be approximated by a sequence 
$\sum_{j=1}^Nc_{q,j}\otimes c_{p,j}$ as $N\rightarrow \infty$,
so that for all of the Schwartz norms $\|\cdot\|_s$ there holds 
$\sum_{j=1}^\infty \|c_{q,j}\otimes c_{p,j}\|_s<\infty$, it
follows that ${\rm L}_c$ and ${\rm R}_c$ are bounded operators on 
$L^2(\bR^{q+p})$ for all $c\in\sS(\bR^{q+p})$. Furthermore, we
put on record here the following hermiticity property of ${\rm L}_c$ and ${\rm 
R}_c$.

\begin{lemma}\label{Lemma_HermPropMoyal}
Let $c\in\sS^M_{\star_{(q,p)}}$ and let $\varphi,\psi\in L^2(\bR^{q+p})$. Then
\begin{eqnarray}
(c\star_{(q,p)}\varphi,\psi)_{L^2}&=&(\varphi,c^*\star_{(q,p)}\psi)_{L^2}\label{E_LemmaLRHerm_L}\\
(\varphi\star_{(q,p)}c,\psi)_{L^2}&=&(\varphi,\psi\star_{(q,p)}c^*)_{L^2}\label{E_LemmaLRHerm_R}.
\end{eqnarray}
\end{lemma}

\begin{proof}
Consider first the case $q=0$. Then
\begin{eqnarray}
(c\star_{(q,p)}\varphi,\psi)_{L^2}&=\frac{1}{(\pi\theta)^p}\int 
\overline{c(w)\varphi(v)}
e^{-i(x-w)\cdot M^{-1}(x-v)}\psi(x)d^pwd^pvd^px,\label{E_LemmaLRHerm_1}\\
(\varphi,c^*\star_{(q,p)}\psi)_{L^2}&\!\!\!\!\!\!\!\!=\frac{1}{(\pi\theta)^p}\int 
\overline{\varphi(x)c(y)}\psi(z)
e^{i(x-y)\cdot M^{-1}(x-z)}d^pzd^pyd^px.\label{E_LemmaLRHerm_2}
\end{eqnarray}
Carrying out the substitution $(w,v,x)\mapsto (y,x,z)$, the right hand side 
of~(\ref{E_LemmaLRHerm_1}) becomes
\begin{equation}\label{E_LemmaLRHerm_3}
\frac{1}{(\pi\theta)^p}\int \overline{c(y)\varphi(x)}e^{-i(z-y)\cdot M^{-1}(z-
x)}\psi(z)d^pyd^pxd^pz.
\end{equation}
Thus one can see that~(\ref{E_LemmaLRHerm_3}) coincides 
with~(\ref{E_LemmaLRHerm_2}) upon noticing that, using the anti-symmetry
of $M^{-1}$,
\[
(z-x)\cdot M^{-1}(z-y)=-x\cdot M^{-1}z+x\cdot M^{-1}y-z\cdot M^{-1}y
\]
coincides with
\[
(x-y)\cdot M^{-1}(x-z)=-x\cdot M^{-1}z+y\cdot M^{-1}z-y\cdot M^{-1}x.
\]
This proves~(\ref{E_LemmaLRHerm_L}) in the case $q=0$, 
and~(\ref{E_LemmaLRHerm_R}) is proved analogously. Then we notice
that~(\ref{E_LemmaLRHerm_L}) and~(\ref{E_LemmaLRHerm_R}) are obviously correct 
for $p=0$. Therefore we obtain, using the tensor
product decomposition of $\star_{(q,p)}$ as in~(\ref{E_TensProdqp}),
\begin{eqnarray*}
\lefteqn{\left(\varphi_q\otimes \varphi_p,(c_q\otimes 
c_p)^*\star_{(q,p)}(\psi_q\otimes \psi_p)\right)_{L^2}}\\
&=&\left(\varphi_q\otimes \varphi_p,(c_q^*\star_{(q,0)}\psi_q)\otimes 
(c_p^*\star_{(0,p)}\psi_p)\right)_{L^2}\\
&=&\left((c_q\star_{(q,0)}\varphi_q)\otimes 
(c_p\star_{(0,p)}\varphi_p),\psi_q\otimes \psi_p\right)_{L^2}\\
&=&\left((c_q\otimes c_p)\star_{(q,p)}(\varphi_q\otimes \varphi_p),\psi_q\otimes 
\psi_p\right)_{L^2},
\end{eqnarray*}
whenever $c_q\in\sS^M_{\star_{(q,0)}}$, $c_p\in\sS^M_{\star_{(0,p)}}$ and 
$\varphi_q,\psi_q\in L^2(\bR^q)$,
$\varphi_p,\psi_p\in L^2(\bR^p)$. This implies~(\ref{E_LemmaLRHerm_L}). 
Relation~(\ref{E_LemmaLRHerm_R}) is proved
analogously.\qed
\end{proof}

\section{The Dirac field on Moyal-deformed Minkowski spacetime as a Lorentzian 
spectral geometry --- general discussion}\label{S_DiracMoyalGeneral}

We will now embark on a --- rather informal --- discussion on the setting in which 
we wish to view the quantized Dirac field on
Moyal-deformed Minkowski spacetime.

Assume that $n\geq 2$, $n=1+s$, and assume the restrictions on $n$ made before 
in~(\ref{E_3}). Let $q+p=n$ where $p$ is even.
Let $C^\infty(\bR^n,\bC^N)$, $N=N(n)$ as in~(\ref{E_1}), denote the space of 
smooth spinor fields on flat Minkowski spacetime
$\bR^n=\bR^{1+s}$ as introduced in section~\ref{S_DiracField}. We can introduce 
a scalar product on the spinors given by
\begin{equation}\label{E_ScProdOnSpinorH}
(\psi,\eta)=\int_{\bR^n}\bar \psi^A(x)\delta_{AB}\eta^B(x)d^nx
\end{equation}
for $\psi=(\psi^A)_{A=1}^N$, $\eta=(\eta^A)_{A=1}^N$ in $L^2(\bR^n)\otimes 
\bC^N$. Let $\cH=\cH_n$ denote the Hilbert space of
square-integrable spinors $L^2(\bR^n)\otimes \bC^N$, carrying the scalar 
product~(\ref{E_ScProdOnSpinorH}). Then
$\sS(\bR^n,\bC^N)\cong \sS(\bR^n)\otimes \bC^N$ is a dense subspace of $\cH$. 
The algebra $\sS^M_{\star_{(q,p)}}$ can act from
the left or the right on $\cH$; an explicit representation of the left action is
\begin{equation}\label{E_rhoReprAonH}
({\sf L}_c\psi)^A=c\star_{(q,p)}\psi^A
\end{equation}
for $\psi=(\psi^A)_{A=1}^N$ in $\cH$. We denote by $\cA^M$ the represented 
algebra ${\sf L}_{\sS^M_{\star_{(q,p)}}}$. Thus, we have
a $*$-algebra of bounded linear operators, $\cA^M$, acting on $\cH$, (cf. last 
section), and if $p\neq 0$, then this algebra is
non-commutative. Furthermore, we have the usual Dirac operator $D$ defined 
in~(\ref{E_DiracOpWithmV}), whereas we set $D=D_0$ for
potential $V=0$ here, acting on a dense domain in $\cH$; for convenience, we 
shall take this domain to be $C_0^\infty(\mathbb{R})\otimes\sS(\bR^s) \otimes\bC^N$.

The said data $\cA^M,\cH,D$ are reminiscent of the data of a spectral triple in 
the spectral triple approach to non-commutative
geometry by Connes~\cite{Connes},\cite{Connes2}, and in fact, this is how we 
would like to think of them. There are, however, a few technical
obstructions to doing so, since the original spectral geometry approach 
generalizes compact Riemannian spin geometries, while in
our case $\cA^M$ is a non-commutative deformation of an algebra of functions 
over the non-compact $\bR^n$, and $D$ is the Dirac
operator of a metric of Lorentzian signature. This means that a modified 
structure needs to be provided in order to attain a
spectral geometry generalization of non-compact Lorentzian spin geometries of 
comparable strength as in the compact, Riemannian
case. This endeavour will be carried out elsewhere~\cite{PRV}, we report here 
only about some of the important ingredients in a
rather non-technical manner, and largely tailored to our Moyal spacetime case at 
hand.

We begin by noting that structural elements in addition to $\cA^M,\cH,D$ are 
needed already in the Riemannian spectral geometry
framework. What is required is an anti-unitary involution $C$ on $\cH$, playing 
the role of a charge conjugation, and in our
Moyal-setting, $C$ will in fact be defined as in~(\ref{E_ChConjC}). Additionally, 
one needs an operator $\bm\gamma$ on
$\cH$ which induces an orientation, and in our Moyal-case at hand, 
$\bm\gamma=\gamma_0\gamma_1\cdots\gamma_s$ is the product of
Dirac matrices, acting on $L^2$-spinors in $\cH$ by matrix multiplication from 
the left.

Supposing for a moment (for the purpose of comparison) that 
$\cA^M,\cH,D,C,\bm\gamma$ were describing a compact (non-commutative)
Riemannian spin geometry in the framework of spectral geometry --- which actually 
is not the case --- then the just listed items
would be required to fulfill important structural properties, such as:
\renewcommand{\theenumi}{\roman{enumi}} 
\renewcommand{\labelenumi}{(\theenumi)}
\begin{enumerate}
\item\label{I_SpecGeomCompRiem_1}
$\cA^M$ is a unital pre-$C^*$-algebra
\item\label{I_SpecGeomCompRiem_2}
$D$ is hermitean and elliptic, and $(D-\lambda\eins)^{-n}$ is in a suitable 
Schatten class for $\lambda \not\in\text{spec }D$
\item
a series of (anti-)commutation relations between $\cA^M,D,C$ and $\bm\gamma$
\item
certain ``regularity'' conditions on $\cA^M$ and $D$ (including domain 
conditions)
\end{enumerate}
\renewcommand{\theenumi}{\alph{enumi}} 
\renewcommand{\labelenumi}{(\theenumi)}
(See~\cite{GVF} for a detailed exposition of the required properties.)

Now in the present case, where $\cA^M,\cH,D,C,\bm\gamma$ actually derive from 
Moyal-deformed Minkowski spacetime, several of these
properties, in particular~(\ref{I_SpecGeomCompRiem_1}) 
and~(\ref{I_SpecGeomCompRiem_2}), no longer hold, but need to be replaced
by suitable generalizations. We won't discuss here the appropriateness of the 
generalizations envisaged (see~\cite{PRV}), but only
give a few indications of their nature. $\cA^M$ is not a unital algebra, so one 
needs, as a further datum, a unitization
$\cA^{M,I}\supset \cA^M$, where $\cA^{M,I}$ is a unital pre-$C^*$-algebra. The 
work~\cite{GGISV} contains an extended
discussion on the best choice of $\cA^{M,I}$ in the Riemannian Moyal-algebra 
case (actually, for $q=0$), and since this
discussion concerns mainly topological aspects of the non-commutative space as 
opposed to its metric structure, the results of this
apply here as well.

In~\cite{GGISV}, $\cA^{M,I}$ is constructed as follows. Let $c$ be a $C^\infty$ 
function on $\bR^p$ which is bounded together
with all of its derivatives. Then define the operator (cf.~(\ref{E_rhoReprAonH}))
\begin{equation*}
{\sf L}_c : \psi \mapsto {\sf L}_c\psi
\end{equation*}
for all $\psi\in\cH=L^2(\bR^p)\otimes \bC^N$. This is a bounded operator with 
respect to the operator norm on
$L^2(\bR^p)\otimes \bC^N$. The $*$-algebra generated by these operators is taken 
as $\cA^{M,I}$. Note that
\[
{\sf L}_{c_1}{\sf L}_{c_2}\psi = 
{\sf L}_{c_1 \star_{(0,p)} c_2}\psi
\]
when $c_1\star_{(0,p)}c_2$ is defined, and likewise ${\sf L}_c^* = {\sf 
L}_{c^*}$. One can opt for this choice of $\cA^{M,I}$ also in the case of 
$\star_{(q,p)}$.

Another modification is needed for~(\ref{I_SpecGeomCompRiem_2}). Already in the 
non-compact Riemannian case, $(D-\lambda\eins)^{-n}$
is not compact for resolvent values $\lambda$ of $D$, but this can be remedied 
by requiring that $a(D-\lambda\eins)^{-n}$ is in a
suitable Schatten class for $a\in\cA^M$. However, in the Lorentzian case, $D$ is 
not elliptic, and thus $a(D-\lambda\eins)^{-n}$ is
non-compact. Moreover, $D$ is not hermitean with respect to the $L^2$ scalar 
product.

A way to get around this difficulty is to introduce another element of structure 
in the form of a further linear, bounded operator
$\fdag\beta:\cH\rightarrow \cH$. This operator carries the information of a 
``time-like'' direction and thereby encodes the
Lorentzian metric signature; in our case, $\fdag\beta=\gamma_0$, acting as 
(matrix) multiplication operator on the spinors. The
characteristic properties of $\fdag\beta$, besides $\fdag\beta^2=1$ and suitable 
Clifford relations with $C$ and $\bm\gamma$, are
\begin{equation*}
\fdag\beta D=D^*\fdag\beta\text{ on the $C^\infty$-domain of }D,
\end{equation*}
and that
\begin{equation*}
\langle D\rangle=\sqrt{\frac{1}{2}(D^*D+DD^*)}
\end{equation*}
is an elliptic operator so that $a(\langle D\rangle-\lambda\eins)^{-n}$ is in a 
suitable Schatten class for resolvent values
$\lambda$ of $\langle D\rangle$ and $a\in\cA^M$. (The adjoint $D^*$ is defined 
with respect to the scalar product of $\cH$.)
Whence, the collection of objects
\[
\cA^{M,I}\supset \cA^M,\cH,D,\fdag\beta,C,\bm\gamma
\]
in combination with a list of relations and conditions that will be discussed in 
detail in~\cite{PRV}, can be viewed as a
``Lorentzian spectral triple'' (LOST), i.e. the generalization of spectral 
geometry from Riemannian to Lorentzian signature. As we
have outlined, Moyal-deformed Minkowski spacetime can be fit into this setting.

If one now contends that non-commutative Lorentzian spacetimes are described in 
terms of LOSTs with data
$\cA^{M,I}\supset \cA^M,\cH,D,\fdag\beta,C,\bm\gamma$, one is faced with the 
question as to what a quantum field theory on a LOST
should be, and how such quantum field theories can, on one hand, be constructed, 
and on the other hand, be interpreted. A fairly
immediate idea is this: Since a Hilbert space $\cH$ with a Dirac-operator $D$ 
and a charge conjugation $C$ acting in it are part of
the data describing a LOST, one may define the Dirac field on a LOST as an 
abstract CAR algebra corresponding to these data.

One must remember, however, that the Hilbert space $\cH$ does not play the role 
of the Hilbert space $\cK$ ($=\cK_V$,
$V=0$) in Proposition~\ref{Prop1}, describing the space of equivalence classes 
of smooth, compactly supported elements
in $L^2(\bR^n)\otimes\bC^N$ modulo the kernel of the operator $R=R^+-R^-$ (where 
$R^\pm$ are the advanced/retarded
fundamental solutions of $D$). Nevertheless, the Hilbert space structure of 
$\cH=L^2(\bR^n)\otimes\bC^N$ is used to obtain
a Hilbert space structure on the set of equivalence classes.

In the case of a general LOST, it is at present not clear how to characterize 
advanced and retarded fundamental solutions
of $D$. One of the difficulties is caused by the circumstance that ``advanced'' 
and ``retarded'' refer to localization
properties which are notoriously difficult to capture in non-commutative 
geometry. This notwithstanding let us, for the
time being, suppose that we have a LOST where advanced and retarded fundamental 
solutions of $D$ are given as quadratic
forms on a suitable domain $\sD$ contained in the joint $C^\infty$-domain of $D$ 
and $D^*$.
Abusing notation, we will denote these quadratic forms by
\begin{equation*}
f,h\mapsto (f,R^\pm h),\quad f,h\in\sD.
\end{equation*}
The fundamental solution property amounts to the condition
\begin{equation*}
(D^*f,R^\pm h)=(f,h)=(f,R^\pm Dh)\text{ for all }f,h\in\sD.
\end{equation*}
Guided by the example of the Dirac field on commutative Minkowski spacetime, one 
is led to the assumption that
\begin{equation*}
(f,h)_{(R)}=e^{i\delta}\left[(\fdag\beta f,R^+h)-(\fdag\beta f,R^-h)\right]
\end{equation*}
defines, upon choice of a suitable phase $\delta$, a scalar product on $\sD/\ker 
(\cdot,\cdot)_{(R)}$.
[At present it is not clear if such a property can actually be proved under 
suitable additional ``regularity'' conditions
on LOSTS, or if this is genuinely an extra assumption; but in our Moyal 
spacetime example in the next section we will
see that this property is fulfilled.]
With this assumption, one can define the Hilbert space $\cK_{(R)}$ as the 
completion of $\sD/\ker (\cdot,\cdot)_{(R)}$
with respect to $(\cdot,\cdot)_{(R)}$. Under these circumstances, the 
conjugation $C$ on $\sD$ induces a conjugation $C$
on $\cK_{(R)}$ via $C[f]_{(R)}=[Cf]_{(R)}$.
Thus, one has a Hilbert space $\cK_{(R)}$ with a conjugation $C$ on it. One can 
therefore define the associated
CAR-algebra $\fF(\cK_{(R)},C)$ in a manner completely analogous to the example 
of the free Dirac field on Minkowski
spacetime, cf.\ Section~\ref{S_DiracField}. That is, $\fF(\cK_{(R)},C)$ is 
generated by $B([f]_{(R)})$, $f\in\sD$, which
are linear in $[f]_{(R)}$, and subject to the relations
\begin{eqnarray*}
B([f]_{(R)})^*&=&B(C[f]_{(R)}),\nonumber\\
\{B([f]_{(R)})^*,B([h]_{(R)})\}&=&2([f]_{(R)},[h]_{(R)})_{(R)}\eins,\nonumber\\
B([Df]_{(R)})&=&0.
\end{eqnarray*}
At this stage, one has constructed abstractly a quantum field theory on a non-
commutative geometry described by a LOST
and some additional structure. The quantum field theory was then essentially 
obtained by second quantization. The
question arises how such a quantum field theory should be interpreted.

Regarding this point, let us specialize to the case that $\cA^M$ is the Moyal-
deformed algebra of functions on
Minkowski spacetime, and $\cH=L^2(\bR^n)\otimes\bC^N$, $\sD=C_0^\infty(\mathbb{R})\otimes\sS(\bR^s) \otimes\bC^N$, 
with the Dirac operator as
in~(\ref{E_DefDiracOpNoPot}). This means that $\cH$ and $D$ are the same as in the case 
of commutative, ``undeformed'' Minkowski
spacetime, just the domain $\sD$ has changed, but this does not lead to a 
significant modification.
As will be explained in the next section, there will again be uniquely 
determined advanced and retarded fundamental
solutions $R^\pm$ of $D$. The CAR-algebra $\fF(\cK_{(R)},C)$ one obtains in this 
case coincides
with $\fF(\cK,C)$ defined in Section~\ref{S_DiracField}, except that $\cK_{(R)}$ is larger than $\cK$
owing to the fact that $\mathscr{D}$ is taken larger than it was in the case of
commutative spacetime. This difference would, however, disappear in the vacuum 
representation
of the Dirac field (defined with respect to the time-translations) upon passing 
to von Neumann
algebras in that representation. Thus, the von Neumann algebras of the CAR-
algebras of
the Dirac field, in vacuum representation, constructed either for classical 
Minkowski spacetime,
or for Moyal-Minkowski spacetime, both coincide.

It is therefore worth contemplating if the sketched way of ``abstract'' 
quantization of the
LOST corresponding to Moyal-deformed Minkowski spacetime leads to anything 
different from
the usual quantized Dirac field on usual Minkowski spacetime. We argue that this 
is indeed
the case. One must remember that, in operational terms,
 a quantum field theory --- on a classical spacetime ---
is described by an assignment of observables to spacetime regions and that the 
physical content
of the theory lies mainly in the localization properties of the observables (and 
their
algebraic relations) relative to each other, see~\cite{Haag},~\cite{HaagKastler}
and discussion further below. We 
must, in the
case of Moyal-Minkowski spacetime, specify the observables of the quantum field 
theory
we have defined, and study their localization properties in connection with the 
algebraic
structure of the Moyal-Minkowski-algebra $\cA^M$.

In the vacuum representation $(\cH^{\rm vac},\pi^{\rm vac},\Omega^{\rm vac})$ of 
$\fF(\cK_0,C)$, we have defined the field
operators
$$ \bm\psi(f) = \pi^{\rm vac}(\Psi_0(f))\,,
\quad f \in C^\infty_0(\mathbb{R}^n,\mathbb{C}^N)\,. $$
These operators do not correspond directly to observable quantities since they 
fulfill anticommutativity upon spacelike separation of the test-spinors $f$.
Therefore, one needs to build operators corresponding to observables from the
$\bm\psi(f)$. A 
common choice is to take operators of the form $\bm\psi(f_1)^*\bm\psi(f_2)$ as building blocks 
for observables.
Then $\bm\psi(f_1)^*\bm\psi(f_2)$ commutes with $\bm\psi(h_1)^*\bm\psi(h_2)$ if the
supports of $f_1$ and $f_2$ are spacelike separated from the supports of $h_1$ and $h_2$.

\label{p_wickcomm} Certain operators arising as limits of linear combinations of such
operators have interesting properties.
Among them is the Wick-product $:\bm\psi^+\bm\psi:(c)$
which is indexed by scalar testing
functions $c \in C^\infty_0(\bR^n,\mathbb{R})$. 
One may define
$: \bm\psi^+\bm\psi:(c)$ as follows. 
Take two finite families of spinors, $e_\mu$ and $\eta_\mu$ $(\mu = 1,\ldots, L)$
in $\mathbb{C}^N$, with the property that $\sum_{\mu =1}^L \overline{e^A}_{\mu}\eta^{B}_{\mu} = \frac{1}{4}\gamma_{0}{}^{AB}$
(the matrix entries of $\gamma_0$). Then define, for
$q_1$ and $q_2$ in $C_0^\infty(\mathbb{R}^n,\mathbb{R})$, the operator 
$$ \bm\psi^+\bm\psi (q_1 \otimes q_2) = \sum_{\mu = 1}^L \bm\psi(q_1 e_{\mu})^*\bm\psi(q_2 \eta_{\mu})\,.$$
The map $q_1 \otimes q_2 \mapsto \bm\psi^+\bm\psi(q_1 \otimes q_2)$ defines a real-linear operator-valued
distribution and thus extends to $C_0^\infty(\mathbb{R}^n \times \mathbb{R}^n,\mathbb{R})$. 
Let $j_\epsilon$ be a family of
real-valued functions in $C^\infty_0(\mathbb{R}^n)$
approaching the $\delta$-measure peaked at $0$
for $\epsilon \to 0$, and set, for $q_1,q_2 \in C_0^\infty(\mathbb{R}^n,\mathbb{R})$,
$$ F_\epsilon(x,y) = q_1(x)q_2(y)j_\epsilon(x-y) \quad (x,y \in \mathbb{R}^n)\,. $$
Moreover, denote by $\mathcal{W} \subset \mathcal{H}^{\rm vac}$ the dense subspace generated
by $P \Omega^{\rm vac}$ where $P$ ranges over all polynomials in the $\bm\psi(f)$
with $f \in C_0^\infty(\mathbb{R}^n,\mathbb{C}^N)$
(including the case that $P$ has degree zero, i.e.\ $P$ is a multiple of $\eins$).
With these conventions, we define
$$ :\bm\psi^+\bm\psi:(c)\chi =
\lim_{\epsilon \to 0} \, \bm\psi^+\bm\psi(F_\epsilon)\chi - (\Omega^{\rm 
vac},\bm\psi^+ \bm\psi(F_\epsilon)\Omega^{\rm vac}) \chi$$
for all $\chi \in \mathcal{W}$ and $c(x) = q_1(x)q_2(x)$ $(x \in \mathbb{R}^n)$.
It turns out (see Sec.\ 7 and Appendix A) that $:\bm\psi^+\bm\psi:(c)$ is an essentially selfadjoint
operator on $\mathcal{W}$ which furthermore turns out to be independent of the
choices made for $e_\mu$ and $\eta_\mu$ $(\mu = 1,\ldots,L)$.
The $:\bm\psi^+\bm\psi:(c)$ are local operators in the sense
that $:\bm\psi^+\bm\psi:(c_1)$ commutes with $:\bm\psi^+\bm\psi:(c_2)$ if
the supports of $c_1$ and $c_2$
are spacelike separated.
For $c \ge 0$, $:\bm\psi^+\bm\psi:(c)$ can be interpreted as the observable of 
(squared) field strength density weighted with the function $c$.

An interesting property of $:\bm\psi^+\bm\psi:(c)$, proven
in Appendix A, is
\begin{equation} \label{E_Wickcom}
[:\bm\psi^+\bm\psi:(c),\bm\psi(f)] =
 -i\bm\psi(c R_0f)
\end{equation}
for all $c \in C^\infty_0(\mathbb{R}^n)$ and
all $f \in C^\infty_0(\mathbb{R}) \otimes \mathscr{S}(\mathbb{R}^s) \otimes \mathbb{C}^N$.
On the other hand, we will also show in Sec.~\ref{S_BogForm} that, identifying $c$ with the 
scalar potential in the discussion of potential scattering in Sec.~\ref{S_DiracMoyalSpecial}, there holds
\begin{equation} \label{E_Scatder}
\left. \frac{d}{d\lambda} \right|_{\lambda = 0}
\pi^{\rm vac}(\beta_{\lambda c}(\Psi_0(f)))
 = \left. \frac{d}{d\lambda} \right|_{\lambda = 0}
 S_{\lambda c} \bm\psi(f) S_{\lambda c}^{-1}
 =  \bm\psi(c R_0 f)
\end{equation}
for all $c \in C^\infty_0(\mathbb{R}^n)$ and
$f \in C^\infty_0(\mathbb{R}^n,\mathbb{C}^N)$.

In view of \eqref{E_Wickcom} and \eqref{E_Scatder},
the observables $:\bm\psi^*\bm\psi:(c)$ are identified as
$-i\left.  \frac{d}{d \lambda}\right|_{\lambda = 0} S_{\lambda c}$ where $S_c$ 
is the scattering matrix corresponding to the localized scattering potential $c$. 
This connection between localized observables
and the derivative of the scattering matrix of a localized interaction with 
respect to the interaction strength is, of course, long known, especially in the 
context of perturbative interacting quantum field theory,
and often goes by the name ``Bogoliubov's formula''~\cite{Bogoliubov}.

We now wish to point out that one can obtain in a similar manner observables for 
the quantized Dirac field on Moyal-deformed Minkowski spacetime employing 
Bogoliubov's formula. The precise mathematical discussion
of the considerations we present here will be given in the next section.
In the case of the Dirac field on usual Minkowski spacetime, the scattering 
matrix $S_V \equiv S_c$ was 
constructed for the Dirac operator
$D_V = D + V$ where the potential term was
$Vf = cf$, $cf$ meaning the usual pointwise (and component-wise) multiplication 
of a scalar function $c$ with a spinor-field $f$.
We should now recall that classical Minkowski spacetime is
also described by the structure of a LOST. 
The data for the LOST corresponding to classical Minkowski spacetime coincide 
with the data for the LOST of Moyal-Minkowski spacetime, except that instead of 
the non-commutative algebra $\cA^M = \sS^M_{\star_{(q,p)}}$
we have the commutative algebra $\cA^{\rm Min}=
C^\infty_0(\mathbb{R}^n)$ of scalar functions on
spacetime. The map $c \mapsto c \varphi$,
$\varphi \in \cH = L^2(\mathbb{R}^n,\mathbb{C}^N)$
produces a faithful representation of $\cA^{\rm Min}$
on the Hilbert-space of square-integrable 
spinor fields. For the case of Moyal-Minkowski spacetime, one can regard the 
potential term
$V$ in a similar light, and define, for $\varphi
\in \cH$, for instance
\begin{equation} \label{E_ncpot}
 V \varphi = {\sf L}_c\varphi +
 {\sf R}_c \varphi =
c \star \varphi + \varphi \star c 
\end{equation}
with real-valued $c$ in $\cA^M = \sS^M_{\star_{(q,p)}}$.  
In the next section we will show that in the case
$q=1,p=2l>0$, i.e.\ when the Moyal-deformed Minkowski spacetime has no non-
trivial commutation relations between time- and space-coordinates,
there is a Bogoliubov-transformation $\beta^M_V$
on the CAR-algebra $\fF(\cK = \cK_{(R_0)},C)$ 
describing scattering by the non-commutative potential
$V$ given in \eqref{E_ncpot}. (This needs mild further assumptions on $c$, see 
Sec.\ 6 for details.)  Furthermore, we will
show that this scattering transformation is unitarily implementable in the 
vacuum-representation $(\cH^{\rm vac},\pi^{\rm vac},\Omega^{\rm vac})$, so that 
there is a unitary operator $S_V^M$ with
$$ S_V^M \pi^{\rm vac}(\Psi_0(f))(S_V^M)^{-1}
= \pi^{\rm vac}(\beta_V^M(\Psi_0(f))) $$
for all $f \in C_0^\infty(\mathbb{R}) \otimes \mathscr{S}(\mathbb{R}^s) \otimes \mathbb{C}^N$.
Consequently, one can formally define the derivative
\begin{equation} \label{E_Phic}
 \Phi(c) = -i\left.\frac{d}{d\lambda}\right|_{\lambda = 0} S_{\lambda V}^M
\end{equation}
which, following the ideas underlying Bogoliubov's
formula alluded to just before, would correspond
to an observable quantity. In Sec.\ 7 we will in fact show that
\begin{equation} \label{E_derivationM}
 \left.\frac{d}{d\lambda}\right |_{\lambda = 0} S^M_{\lambda V} \bm\psi(f)
 S^M_{\lambda V}{}^{-1} = [i\Phi(c),\bm\psi(f)] = \bm \psi(VR_0f) 
\end{equation}
for all $f \in C_0^\infty(\mathbb{R}) \otimes \mathscr{S}(\mathbb{R}^s) \otimes \mathbb{C}^N$
with an essentially selfadjoint operator $\Phi(c)$ on $\mathcal{W}$. One may therefore
identify $\Phi(c)$ with the derivative $-i\left.\frac{d}{d\lambda}\right|_{\lambda = 0} S_{\lambda V}^M$,
in the sense that \eqref{E_derivationM} holds.

In the case of usual Minkowski spacetime, the assignment $c \mapsto :\bm\psi^+\bm\psi:(c)$, where
$c$ is a scalar $C_0^\infty$ test-function on spacetime, has the typical properties of an observable
quantum field of Wightman type \cite{StreaterWightman}. The support of the
test-function $c$ limits the localization of the observable $:\bm\psi^+\bm\psi:(c)$, which is reflected
by the relations \eqref{E_Wickcom} and \eqref{E_Scatder} and the fact that the changes of states
which $\beta_{\lambda c}$ induces are localized in the support of $c$. In the algebraic approach to
quantum field theory \cite{Haag,HaagKastler}, one therefore considers the $*$-algebras $\mathfrak{R}(O)$ generated by all
observable quantum field operators $:\bm\psi^+\bm\psi:(c)$ where the support of $c$ is contained in
the spacetime region $O$.\footnote{Two things should be noted here. (1) Actually, $\mathfrak{R}(O)$
would have to be defined as algebraically generated by {\it all} observable quantum field operators
smeared with test-functions supported in $O$; we use $:\bm\psi^+\bm\psi:$ as a placeholder for any
observable quantum field at this point. (2) In the algebraic approach to quantum field theory it is
customary to define $\mathfrak{R}(O)$ as algebraically generated by the bounded functions of quantum
field operators smeared with test-functions supported in $O$; here, our $\mathfrak{R}(O)$ are
algebras of unbounded operators.}
Then one obtains an assignment $O \mapsto \mathfrak{R}(O)$ of spacetime regions to operator algebras
with the two characteristic properties of
\begin{eqnarray*}
    \text{Isotony:} & & O_1 \subset O_2\ \ \Rightarrow\ \ \mathfrak{R}(O_1) \subset \mathfrak{R}(O_2) \\
    \text{Locality:} & & O_1 \perp O_2\ \ \Rightarrow\ \ [F_1,F_2] = 0\quad \text{for} \ \ 
                        F_j \in \mathfrak{R}(O_j) \quad (j =1,2)
\end{eqnarray*}
where $O_1 \perp O_2$ means that the spacetime regions are causally separated, i.e.\ there is no
causal curve joining them.

According to the algebraic approach to quantum field theory, a quantum field theoretical model is
basically characterized by a map $O \mapsto \mathfrak{R}(O)$ with
these properties (see \cite{HaagKastler,Haag,Roberts}), describing especially the localization of observables of the quantum system under consideration
on a ``classical'' spacetime with commutative coordinate functions.    
Let us now discuss some, however vague, ideas how this may be generalized to quantum field theories on
non-commutative spacetimes, where again we stay at the level of the Dirac field on Moyal-Minkowski spacetime.
The scattering by a non-commutative potential furnishes the assignment $c \mapsto \Phi(c)$ of \eqref{E_Phic}.
We interpret $\Phi(c)$ as an observable, and hence we have an assignment of elements $c$ in the
non-commutative algebra $\mathcal{A}^M$ to (unbounded) operators in $\mathcal{H}^{\rm vac}$. The $c$ 
now carries the information about the spacetime localization of $\Phi(c)$, but due to the non-commutativity
of $\mathcal{A}^M$,
this is subject to uncertainties. In particular, in general $\Phi(c_1)$ and $\Phi(c_2)$ won't commute
anymore if the supports of $c_1$ and $c_2$, viewed as test-functions, are causally separated. Therefore,
if one defines the algebras $\mathfrak{R}^M(O)$ as being generated by the $\Phi(c)$ where $c$ has support
in $O$, then the assignment $O \mapsto \mathfrak{R}^M(O)$ is clearly different from 
$O \mapsto \mathfrak{R}(O)$ as defined above for usual Minkowski spacetime, and thus we see that we derive
indeed a different system of observables from the scattering morphisms via Bogoliubov's formula in the
non-commutative case, without an obvious locality structure.  

Nevertheless, one may attempt to mimic the algebraic approach to quantum field theory in a generalized form,
upon forming algebras of observables $\mathfrak{R}^M(\mathcal{P})$ labelled by subsets
$\mathcal{P}$ of $\mathcal{A}^M$, understanding that
$\mathfrak{R}^M(\mathcal{P})$ be generated by
the $\Phi(c)$ with $c \in \mathcal{P}$.
 It is not clear at this stage what structure these subsets should have,
e.g.\ if they should be subalgebras of $\mathcal{A}^M$. In comparison to the classical case, what seems
to be required is a partial ordering on the collection of chosen $\mathcal{P}$, and a concept of causal
separation \cite{Roberts}. An idea could be to choose the $\mathcal{P}$ as sets of (approximate) projections, inspired by
the situation on classical spacetime, where a subset $O$ may be identified with its characteristic function,
which is a projection in the commutative algebra of coordinate functions. The ordering relation may then
be taken as operator ordering. It is more difficult to capture
the concept of causal separation. In the case of classical Minkowski spacetime, the supports of two
$C^\infty_0$ test-functions $c_1$ and $c_2$ are causally separated if and only if 
$i \langle c_1 f, R_m c_2 h \rangle = 0$ for all spinor fields $f$ and $h$, where $R_m$ is the causal
propagator for the Dirac equation 
for any mass term $m$ (corresponding to $R_V$ for $V = m$, cf.\ eqn.\ \eqref{E_propagatorV}).
By assumption, the causal propagator is available in our non-commutative setup as the quadratic form
$(\,.\,,\,.\,)_{(R)}$ on the domain $\mathscr{D} \subset \mathcal{H}$, and thus one may characterize
the causal disjointness of two subsets $\mathcal{P}_1$ and $\mathcal{P}_2$ of $\mathcal{A}^M$
with the help of this quadratic form for any mass term. Ideally, one might want to define $\mathcal{P}_1 \perp \mathcal{P}_2$
to be equivalent to $(c_1 f, c_2 h)_{(R)} = 0$ for all $c_j \in \mathcal{P}_j$ and all $f,h \in
\mathscr{D}$, provided this is compatible with the ordering relation. If that cannot be had, the second best
option would be to define $\mathcal{P}_1 \perp \mathcal{P}_2$ as meaning that 
$(c_1 \,.\,,c_2\,.\,)_{(R)}$ is, in a suitable sense, ``small'' compared to $(\,.\,,\,.\,)_{(R)}$ --- a sort
of ``infinitesimal'' quantity in the sense of spectral geometry. 

Supposing that suitable forms of a partial ordering relation $\mathcal{P}_1 \le \mathcal{P}_2$ and a causal separation
relation $\mathcal{P}_1 \perp \mathcal{P}_2$ have been found for suitably chosen subsets $\mathcal{P}$ of $\mathcal{A}^M$,
it seems well possible that the generalized version of a quantum field theory on non-commutative spacetime
in the operator algebraic setting may take the shape of an assignment $\mathcal{P} \mapsto \mathfrak{R}^M(\mathcal{P})$,
where the $\mathfrak{R}^M(\mathcal{P})$ are operator algebras, subject to the relations of
\begin{eqnarray*}
    \text{Isotony:} & & \mathcal{P}_1 \le {P}_2\ \ \Rightarrow\ \ \mathfrak{R}^M(\mathcal{P}_1) \subset \mathfrak{R}^M(\mathcal{P}_2) \\
    \text{Locality:} & & \mathcal{P}_1 \perp \mathcal{P}_2\ \ \Rightarrow\ \ [F_1,F_2] = 0\quad \text{for} \ \ 
                        F_j \in \mathfrak{R}^M(\mathcal{P}_j) \quad (j =1,2) \,.
\end{eqnarray*}                          
Actually, it could happen that the condition of locality ought to be relaxed requiring only that 
$[F_1,F_2]$ is in a suitable sense ``small'' compared to $F_1$ and $F_2$ if $\mathcal{P}_1 \perp \mathcal{P}_2$,
similar in spirit to the possibly generalized condition of causal separation. Admittedly, this is at present all
speculation, and a careful study of examples is required before a clear picture of the basic structure of
quantum field theory on (Lorentzian) non-commutative spacetime will emerge.  

\section{The Dirac field on Moyal-deformed Minkowski spacetime -- the 
model}\label{S_DiracMoyalSpecial}

Now our intention is to follow the lines of Section~\ref{S_DiracField} under the 
modifications of using $n=1+s=q+p$ dimensional
Moyal-deformed Minkowski spacetime and a suitable different potential term for 
the Dirac operator.

The spacetime of interest (with dimension $n=1+s=q+p$) will be described as in 
section~\ref{S_MoyalMinkST}, with the exception
that we restrict ourselves to Moyal matrices $M$ of the more specialized form
\begin{equation} \label{E_Moyalmatrix}
M=
\left[ 
\begin{array}{c|cc}
0&\cdots&0\\
\hline 
\vdots&M_{(q+p-1)\times (q+p-1)}&\\
0&&\\
\end{array}
\right]_{(q+p)\times(q+p)}
\end{equation}
i.e. the first row and the first column shall vanish.

Nothing is changed (cf.\ Section~\ref{S_DiracField}) in the manner of how we 
define the algebra of Dirac matrices
$(\gamma_0,\gamma_1,\ldots,\gamma_s)$ and the charge conjugation $C$. Again the 
Dirac operator ($m>0$ constant)
acting on $C^\infty_0(\bR)\otimes \sS(\bR^s)\otimes \bC^N$ is denoted by
\begin{equation*}
D_V=(-i\fdag\partial+m)+V.
\end{equation*}
But now the ``potential term'' operator $V$ acting on $f\in 
C^\infty_0(\bR)\otimes \sS(\bR^s)\otimes \bC^N$ is not just
the multiplication operator multiplying $f$ with a scalar function, but one of 
the following operators:
\begin{eqnarray} 
\text{(i) }(Vf)^A(x)&=&(V_\text{(i)}f)^A(x)=(c\star_{(q,p)}f^A)(x)+(f^A\star_{(q
,p)}c)(x)\label{E_DefVMoyal}\\
\text{(ii) }(Vf)^A(x)&=&(V_\text{(ii)}f)^A(x)=(c\star_{(q,p)}f^A\star_{(q,p)}c)(
x),\label{E_DefVMoyalAlt}
\end{eqnarray}
where $c\in C^\infty_0(\bR,\bR)\otimes \sS(\bR^s,\bR)$ is a function of the form
\begin{equation}\label{E_DefcMoyal}
c(x)=a(t)b(\ulx),
\end{equation}
with $a\in C^\infty_0(\bR,\bR)$, $b\in\sS(\bR^s,\bR)$, $t=x^0$, 
$\ulx=(x^1,\ldots,x^s)$.
We aim at presenting an analogue of Proposition~\ref{Prop1} for the potential 
operators $V=V_\text{(i)}$ or
$V=V_\text{(ii)}$ which describe ``scattering by a time-dependent, spatially 
non-commutative potential''.

It is useful, at this point, to consider first the Cauchy-data version of the 
dynamical problem. As in~(\ref{E_DefH0}),
we have the free Hamiltonian
\begin{equation} \label{E_Hzero}
(H_0f)(\ulx)=\left(i\gamma^0\gamma^k\frac{\partial}{\partial 
x^k}+\gamma^0m\right)f(\ulx)
\end{equation}
and the Hamiltonian with time-dependent interaction term,
\begin{equation} \label{E_Hvee}
(H_V(t)f)(\ulx)=\left(i\gamma^0\gamma^k\frac{\partial}{\partial 
x^k}+\gamma^0m+\gamma^0V(t)\right)f(\ulx),
\end{equation}
acting on $f\in\sS(\bR^s)\otimes \bC^N$. Here $V(t)$ stands for the operators
\begin{eqnarray}
V_\text{(i)}(t)&:&f\mapsto V_\text{(i)}(t)f,\quad f\in\sS(\bR^s)\otimes \bC^N,\\
\label{E_Veeone}
(V_\text{(i)}(t)f)^A(\ulx)&=&a(t)(b\star_{(q-1,p)}f^A(\ulx)+f^A\star_{(q-
1,p)}b(\ulx))
\end{eqnarray}
or
\begin{eqnarray}
V_\text{(ii)}(t)&:&f\mapsto V_\text{(ii)}(t)f,\quad f\in\sS(\bR^s)\otimes 
\bC^N,\nonumber\\
\label{E_Veetwo}
(V_\text{(ii)}(t)f)^A(\ulx)&=&a(t)^2(b\star_{(q-1,p)}f^A\star_{(q-1,p)}b(\ulx)).
\end{eqnarray}
By the assumptions made on $a$ and $b$ above, $V(t)=V_\text{(i)}(t)$ and 
$V(t)=V_\text{(ii)}(t)$ are bounded operators
on $L^2(\bR^s,\bC^N)$. As in the case of a scalar potential, we have again that 
$CH_V(t)=-H_V(t)C$, which is a consequence
of the easily checked equation
\begin{equation}\label{E_LeftRightMoyalAndC}
C{\sf R}_c={\sf L}_cC,
\end{equation}
obviously in the case of $V_\text{(i)}$ and under additional usage of the 
associativity of
the Moyal product in the case of $V_\text{(ii)}$. And $H_V(t)$ is a symmetric 
(in fact essentially selfadjoint) operator
on $\sS(\bR^s,\bC^N)\subset L^2(\bR^s,\bC^N)$. As in the case considered before 
in Section~\ref{S_DiracField}, a smooth
function $\varphi\in C^\infty(\bR^s,\bC^N)$ is a solution of
\begin{equation}\label{E_DiracEqOnesMore}
D_V\varphi=0
\end{equation}
if and only if
\begin{equation*}
\frac{1}{i}\frac{d}{dt}P_t\varphi=H_V(t)P_t\varphi.
\end{equation*}
with
\begin{equation*}
P_t\varphi=\varphi|_{\Sigma_t}.
\end{equation*}
Establishing existence and uniqueness of the Cauchy problem for the Dirac 
equation~(\ref{E_DiracEqOnesMore}) is therefore
equivalent to proving existence and uniqueness of solutions for the initial 
value problem
\begin{equation*}
\frac{1}{i}\frac{d}{dt}v_t=H_V(t)v_t,\quad \left.v_t\right|_{t=0}=w.
\end{equation*}
This will be our next auxiliary result.

\begin{proposition}\label{L_ExUniqInitValProbMoyal}
\begin{itemize}
\item[{\rm (a)}]
There is a unique family of unitaries
$T_{t,t'}^{(V)}$ on $L^2(\mathbb{R}^s,\mathbb{C}^N)$, strongly continuous in $t$ and $t'$,
so that 
\begin{equation}
T^{(V)}_{t,t'}\circ T^{(V)}_{t',s}=T^{(V)}_{t,s},\quad T^{(V)}_{t,t}=\eins
\end{equation}
and
\begin{equation}
\frac{1}{i} \frac{d}{dt} T_{t,0}^{(V)}w = H_V(t)w
\end{equation}
for all $w \in L^2(\mathbb{R}^s,\mathbb{C}^N)$.
Moreover, $T_{t,t'}^{(V)}$ maps $\mathscr{S}(\mathbb{R}^s,\mathbb{C}^N)$ into itself.
\item[{\rm (b)}]
Given $w \in \mathscr{S}(\mathbb{R}^s,\mathbb{C}^N)$, the map
$$ (t,t',\underline{x}) \mapsto T_{t,t'}^{(V)}w(\underline{x}) \in \mathbb{C}^N
\quad ((t,t',\underline{x}) \in \mathbb{R} \times \mathbb{R} \times \mathbb{R}^s) $$
is jointly $C^\infty$ in all variables.
\item[{ \rm (c)}]
The Cauchy-problem for the Dirac-equation $D_V\varphi = 0$ with the potential term
$V = V_{\rm (i)}$ or $V = V_{\rm (ii)}$ is well-posed in the following sense. For any given
$w \in \mathscr{S}(\mathbb{R}^s,\mathbb{C}^N)$ and $t' \in \mathbb{R}$ there is a unique $\varphi \in C^\infty(\mathbb{R})
\otimes \mathscr{S}(\mathbb{R}^s) \otimes \mathbb{C}^N$ such that $D_V\varphi = 0$ and
$$ P_{t'}\varphi = w\,.$$
\end{itemize} 
\end{proposition}

\begin{proof}
{\it Part} (a).
We shall work in the interaction picture, i.e. we obtain $T^{(V)}_{t,t'}$ as
\begin{equation}\label{E_DysonInteractEqns}
T^{(V)}_{t,t'}=e^{itH_0}\tilde T^{(V)}_{t,t'}e^{-it'H_0},
\end{equation}
where $\tilde T^{(V)}_{t,t'}$ is the Dyson series for
\begin{equation*}
\tilde U(t)=e^{itH_0}\gamma^0V(t)e^{-itH_0},
\end{equation*}
meaning that
\begin{equation*}
\tilde T^{(V)}_{t,t'}=\sum_{n=0}^\infty \tilde T^{(V)}_n(t,t'),
\end{equation*}
where $\tilde T^{(V)}_n(t,t')$ is iteratively defined by
\begin{equation*}
\tilde T^{(V)}_0(t,t')=\eins,\quad \tilde T^{(V)}_{n}(t,t')=\frac{1}{i}\int_{t'}^t 
\tilde U(r)\tilde T^{(V)}_{n-1}(r,t')dr.
\end{equation*}
Since the operators $V(t)$ and $\gamma^0V(t)$ are bounded operators (with 
uniform bound in $t$) on $L^2(\bR^s,\bC^N)$,
and $H_0$ is essentially selfadjoint, one can rely on Theorem X.69 
in~\cite{ReedSimon2} to see that $T^{(V)}_{t,t'}$ is
a family of unitaries with the required properties, provided that 
$T^{(V)}_{t,t'}$ maps $\sS(\bR^s,\bC^N)$ into itself.
To show this, we note that (cf.~\cite{Thaller}, Theorem 1.2 and Appendix 1.D)
\begin{eqnarray*}
\lefteqn{(e^{itH_0}f)(\ulx)}\\
&&=\int e^{i\underline{k}\cdot\ulx}\left(\cos{(|\widehat{H}(\underline{k})|t)}-
(\gamma^0\gamma^kp_k+i\gamma^0m)
\frac{\sin{( |\widehat{H}(\underline{k}) |t)}}{|\widehat{H}(\underline{k})|}\right)\hat 
f(\ulp)\frac{d^s\ulp}{(2\pi)^s},
\end{eqnarray*}
where $\hat f$ is the Fourier transform of $f$, $\ulp=(p_1,\ldots,p_s)\in\bR^s$, 
and $|\widehat{H}(\underline{k})|=\sqrt{|\underline{k}|^2+m^2}$.
This shows that $e^{itH_0}f$ is in $\sS(\bR^s,\bC^N)$ ($t\in\bR$) for 
$f\in\sS(\bR^s,\bC^N)$ and that, moreover,
$e^{itH_0}f$ is $C^\infty$ in $t$ with respect to the $\sS$-topology. In the 
next step, we note that
\begin{equation*}
\tilde T^{(V)}_{n}(t'',t')=\left(\frac{1}{i}\right)^n\int_{t'}^{t''}\int_{t'}^{t_n}\cdots\int_{t'}^{t_2}
\tilde U(t_n)\cdots\tilde U(t_1)dt_1\cdots dt_n.
\end{equation*}
We set
\begin{eqnarray} \label{E_lcvi}
{\sf v}f&=&{\sf v}_\text{(i)}f=\gamma^0(b\star_{(q-1,p)}f+f\star_{(q-
1,p)}b),\quad V=V_\text{(i)}\\
\label{E_lcvii}
{\sf v}f&=&{\sf v}_\text{(ii)}f=\gamma^0(b\star_{(q-1,p)}f\star_{(q-
1,p)}b),\quad V=V_\text{(ii)}.
\end{eqnarray}
Then $V(t)f=a(t){\sf v}f$, and
\begin{equation*}
\tilde 
T^{(V)}_{n}(t'',t')f=\left(\frac{1}{i}\right)^n\int_{t'}^{t''}\int_{t'}^{t_n}\cdots\int_{t'}^{t_2}
a(t_n)\cdots a(t_1)\boldsymbol{f}^{(n)}(\boldsymbol{t}^{(n)})dt_1\cdots dt_n,
\end{equation*}
where
\begin{equation*}
\boldsymbol{f}^{(n)}(\boldsymbol{t}^{(n)})=e^{it_nH_0}{\sf v}e^{i(t_{n-1}-
t_n)H_0}{\sf v}\cdots e^{i(t_{1}-t_2)H_0}{\sf v}e^{-it_1H_0}f.
\end{equation*}
This implies that, given any pair of multi-indices $\alpha,\beta\in\bN_0^s$, we 
obtain an estimate of the form\footnote{For the remainder of this proof,
$D^\beta = (-i \partial/\partial x^1)^{\beta_1} \cdots (-i \partial/\partial x^s)^{\beta_s}$;
here, $D$ is not to be confused with the Dirac operator.}
\begin{equation}\label{E_L2NormxalphaDbetatildeTnf}
\left\|x^\alpha D^\beta\tilde T^{(V)}_{n}(t'',t')f\right\|_{L^2}\leq 
\frac{C_a(t'',t')^{2n}}{n!}
\max_{t_j\in [t'',t']}{\left\|x^\alpha 
D^\beta\boldsymbol{f}^{(n)}(\boldsymbol{t}^{(n)})\right\|_{L^2}}
\end{equation}
with the constant
\begin{equation*}
C_a(t'',t')=\max_{t\in\bR}{|a(t)||t''-t'|}.
\end{equation*}
Now we will show that there is a constant $C_{b,\alpha,\beta}(t'',t')$ and a 
constant $m(\alpha,\beta)$ such that
\begin{equation}\label{E_n_maxL2Normft}
\max_{t_j\in [t'',t']}{\left\|x^\alpha 
D^\beta\boldsymbol{f}^{(n)}(\boldsymbol{t}^{(n)})\right\|_{L^2}}\leq 
C_{b,\alpha,\beta}(t'',t')^n
\sup_{|\gamma|\leq 2m(\alpha,\beta)}{\|D^\gamma f\|_{L^2}}.
\end{equation}
The proof will be given by induction on $n$, and we will only treat the case 
${\sf v}={\sf v}_\text{(i)}$ (corresponding to
$V=V_\text{(i)}$) as the other case ${\sf v}={\sf v}_\text{(ii)}$ is completely 
analogous.
The inductive proof will only be needed for the special case $x^\alpha=1$ (i.e. 
all $\alpha_j=0$) which makes it more
transparent, and the result will be used in the proof of the general case. Let 
$\beta\in\bN_0^s$ be a multi-index and
define $B_{|\beta|}=\sup_{|\gamma|\leq |\beta|}{\|D^\gamma b\|_{L^2}}$.
We want to prove by induction that there is a constant $F > 0$ with
\begin{equation}\label{E_ast_L2NormDbetaft}
\left\|D^\beta\boldsymbol{f}^{(n)}(\boldsymbol{t}^{(n)})\right\|_{L^2}\leq 
2^n\cdot 2^{|\beta|n}F^nB_{|\beta|}^n
\sup_{|\gamma|\leq |\beta|}{\|D^\gamma f\|_{L^2}}.
\end{equation}
Let us show that this is correct for $n=1$:
\begin{eqnarray*}
\boldsymbol{f}^{(1)}(t_1)&=&e^{it_1H_0}{\sf v}e^{-it_1H_0}f\\
&=&e^{it_1H_0}\gamma^0\left(b\star_{(q-1,p)}(e^{-it_1H_0}f)+(e^{-
it_1H_0}f)\star_{(q-1,p)}b\right).
\end{eqnarray*}
The Leibniz rule for coordinate derivatives applies with respect to the Moyal 
product $\star=\star_{(q-1,p)}$:
\begin{equation}\label{E_LeibnizRuleForMoyal}
\frac{\partial}{\partial x^j}(h\star g)=\left(\frac{\partial}{\partial 
x^j}h\right)\star g+h\star \left(\frac{\partial}{\partial x^j}g\right),
\quad h,g\in\sS(\bR^s).
\end{equation}
Since the coordinate derivatives $\frac{\partial}{\partial x^j}$ commute with 
$e^{itH_0}$, we hence obtain
\begin{eqnarray}
D^\beta\boldsymbol{f}^{(1)}(t_1)&=&\sum_{k=1}^{2^{|\beta|}} e^{it_1H_0}\gamma^0
\left((D^{\beta'(k)}b)\star (D^{\beta''(k)}e^{-it_1H_0}f)\right.\nonumber\\
&&\left.+(D^{\beta'(k)}e^{-it_1H_0}f)\star 
(D^{\beta''(k)}b)\right),\label{E_astast_Dbetaft1}
\end{eqnarray}
where $\beta'(k)$ and $\beta''(k)$ are suitable multi-indices\footnote{The case 
$\beta'(k_1)=\beta'(k_2)$ and
$\beta''(k_1)=\beta''(k_2)$ for some $k_1\neq k_2$ typically occurs in our sum 
decomposition of multiple derivatives
of a product. Usually, this is written as a sum over fewer terms, occurring with 
a multiplicity expressed by binomial
coefficients.} with $\beta'_j(k)+\beta''_j(k)=\beta_j$.
Using the fact that $\|h\star g\|_{L^2}\leq F\|h\|_{L^2}\|g\|_{L^2}$ for 
$h,g\in\sS(\bR^s)$, with $F=(2\pi\theta)^{-p/2}$,
one deduces
\begin{equation}\label{E_astastast_L2NormDbetaft1}
\left\|D^\beta\boldsymbol{f}^{(1)}(t_1)\right\|_{L^2}\leq 
2\cdot 2^{|\beta|}FB_{|\beta|} \sup_{|\gamma|\leq |\beta|}{\|D^\gamma f\|_{L^2}}.
\end{equation}
In order to conclude that the validity of~(\ref{E_ast_L2NormDbetaft}) for some 
$n\in\bN$ implies the
relation~(\ref{E_ast_L2NormDbetaft}) with $n+1$ in place of $n$, we note that
\begin{eqnarray*}
\boldsymbol{f}^{(n+1)}(\boldsymbol{t}^{(n+1)})&=&e^{it_{n+1}H_0}ve^{-
it_{n+1}H_0}\boldsymbol{f}^{(n)}(\boldsymbol{t}^{(n)})\\
&=&e^{it_{n+1}H_0}\gamma^0\left(b\star(e^{-
it_{n+1}H_0}\boldsymbol{f}^{(n)}(\boldsymbol{t}^{(n)}))\right.\\
&&\left.+(e^{-it_{n+1}H_0}\boldsymbol{f}^{(n)}(\boldsymbol{t}^{(n)}))\star 
b\right).
\end{eqnarray*}
Hence, relation~(\ref{E_astast_Dbetaft1}) continues to hold under the 
simultaneous replacements
$\boldsymbol{f}^{(1)}(t_1)\mapsto 
\boldsymbol{f}^{(n+1)}(\boldsymbol{t}^{(n+1)})$,
$e^{\pm it_1H_0}\mapsto e^{\pm it_{n+1}H_0}$ and $f\mapsto 
\boldsymbol{f}^{(n)}(\boldsymbol{t}^{(n)})$.
Therefore,~(\ref{E_astastast_L2NormDbetaft1}) also holds when making these 
replacements, leading to the estimate
\begin{eqnarray*}
\left\|D^\beta\boldsymbol{f}^{(n+1)}(\boldsymbol{t}^{(n+1)})\right\|_{L^2}&\leq& 
2\cdot 2^{|\beta|}FB_{|\beta|}
\sup_{|\gamma|\leq |\beta|}{\|D^\gamma 
\boldsymbol{f}^{(n)}(\boldsymbol{t}^{(n)})\|_{L^2}}\\
&\leq& 2^{n+1}\cdot 2^{|\beta|(n+1)}F^{n+1}B_{|\beta|}^{n+1}
\sup_{|\gamma|\leq |\beta|}{\|D^\gamma f\|_{L^2}},
\end{eqnarray*}
where the induction hypothesis~(\ref{E_ast_L2NormDbetaft}) was used in the 
second inequality. This proves by induction
that~(\ref{E_ast_L2NormDbetaft}) holds for all $n\in\bN$.

Turning to the general case, the first observation is that, given multi-indices 
$\alpha,\beta\in\bN_0^s$ and a finite
real interval $[t'',t']$, there are constants $\Gamma_{\alpha\beta}(t'',t')>0$ 
and $m(\alpha,\beta)>0$ such that
\begin{equation*}
\left\|x^\alpha D^\beta e^{itH_0}\psi\right\|_{L^2}\leq 
\Gamma_{\alpha\beta}(t'',t')\sum_{|\rho|,|\delta|\leq m(\alpha,\beta)}{\|x^\rho 
D^\delta \psi\|_{L^2}}
\end{equation*}
holds for all $\psi\in\sS(\bR^s,\bC^N)$ and all $t\in[t'',t']$, where 
$\rho,\delta$ are multi-indices.
The second observation is that the action of multiplication by a coordinate 
function $x^j$ on a Moyal product can be
decomposed as follows (cf.~\cite{GV}): There are numbers $\varepsilon(j)$ which 
may take the values $0$, $1$ or $-1$, and
for each coordinate $x^j$ there is a coordinate $x^{\iota(j)}$, such that
\begin{eqnarray}
x^j(h\star g)&=&h\star(x^jg)+\frac{i\varepsilon(j)\theta}{2}\frac{\partial 
h}{\partial x^{\iota(j)}}\star g\nonumber\\
&=&(x^jh)\star g-\frac{i\varepsilon(j)\theta}{2}h\star\frac{\partial g}{\partial 
x^{\iota(j)}}\label{E_CoordMultMoyalProd}
\end{eqnarray}
for all $h,g\in\sS(\bR^s)$. Now we define:
\begin{eqnarray*}
M_{\alpha,\beta}&=&\sup_{|\gamma|,|\sigma|\leq 
m(\alpha,\beta)}{\|x^\gamma D^\sigma b\|_{L^2}}\\
C_{b,\alpha,\beta}(t'',t')&=&2\cdot 
2^{2m(\alpha,\beta)}F\Gamma_{\alpha\beta}(t'',t')m(\alpha,\beta)^{2}
M_{\alpha,\beta}\left(1+\frac{|\theta|}{2}\right)^{m(\alpha,\beta)}
\end{eqnarray*}
and we will show that~(\ref{E_n_maxL2Normft}) holds with these definitions for 
all $n\in\bN$. It holds that
\begin{eqnarray}
\lefteqn{\left\|x^\alpha D^\beta 
\boldsymbol{f}^{(n+1)}(\boldsymbol{t}^{(n+1)})\right\|_{L^2}}\nonumber\\
&=&\left\|x^\alpha D^\beta e^{it_{n+1}H_0}{\sf v}e^{-
it_{n+1}H_0}\boldsymbol{f}^{(n)}(\boldsymbol{t}^{(n)})\right\|_{L^2}\nonumber\\
&\leq &\Gamma_{\alpha\beta}(t'',t')\sum_{|\rho|,|\delta|\leq m(\alpha,\beta)}
{\|x^\rho D^\delta {\sf v}e^{-
it_{n+1}H_0}\boldsymbol{f}^{(n)}(\boldsymbol{t}^{(n)})\|_{L^2}}\label{E_L2NormxalphaDbetaftnp1}
\end{eqnarray}
for all $t_{n+1}\in [t'',t']$ and all $n\in\bN_0$ (with 
$\boldsymbol{f}^{(0)}(\boldsymbol{t}^{(0)})=f$). Now we
use~(\ref{E_LeibnizRuleForMoyal}) and~(\ref{E_CoordMultMoyalProd}) to conclude 
that we can write
\begin{eqnarray}
\lefteqn{x^\rho D^\delta {\sf v}e^{-
it_{n+1}H_0}\boldsymbol{f}^{(n)}(\boldsymbol{t}^{(n)})}\nonumber\\
&=&\sum_{l=1}^{2^{|\rho|}}\sum_{k=1}^{2^{|\delta|}} \left(\mu(l,k)\gamma^0
\left(D^{\rho'(l)}D^{\delta'(k)}e^{-
it_{n+1}H_0}\boldsymbol{f}^{(n)}(\boldsymbol{t}^{(n)})\right)
\star \left(x^{\rho''(l)}D^{\delta''(k)}b\right)\right.\nonumber\\
&&\left.+\nu(l,k)\gamma^0\left(x^{\rho'(l)}D^{\delta'(k)}b\right)\star 
\left(D^{\rho''(l)}D^{\delta''(k)}
e^{-
it_{n+1}H_0}\boldsymbol{f}^{(n)}(\boldsymbol{t}^{(n)})\right)\right),\label{E_rhoDdeltavefnt}
\end{eqnarray}
where $\mu(l,k),\nu(l,k)$ are complex numbers and 
$\rho'(l),\rho''(l),\delta'(k),\delta''(k)$ are suitable multi-indices,
where $\delta'_j(k)+\delta''_j(k)=\delta_j$; $|\rho'(l)|,|\rho''(l)|\leq |\rho|$, 
and
$|\mu(l,k)|,|\nu(l,k)|\leq \left(1+\frac{|\theta|}{2}\right)^{|\rho|}$.
Thus we obtain for the sum on the right hand side of~(\ref{E_rhoDdeltavefnt}) 
the $L^2$-norm bound
\[
2\cdot 2^{|\rho|}\cdot 2^{|\delta|}\left(1+\frac{|\theta|}{2}\right)^{|\rho|}F
\sup_{|\gamma|\leq |\rho|,\ |\sigma|\leq |\delta|}{\|x^\gamma D^\sigma b\|_{L^2}}
\cdot \sup_{|\gamma|\leq |\rho|+|\delta|}{\|D^\gamma f\|_{L^2}},
\]
using~(\ref{E_ast_L2NormDbetaft}) again; inserting this 
into~(\ref{E_L2NormxalphaDbetaftnp1}) yields
\begin{eqnarray*}
\lefteqn{\left\|x^\alpha D^\beta 
\boldsymbol{f}^{(n+1)}(\boldsymbol{t}^{(n+1)})\right\|_{L^2}}\\
&\leq &\Gamma_{\alpha\beta}(t'',t')m(\alpha,\beta)^{2}\cdot 2\cdot 
2^{2m(\alpha,\beta)}
\left(1+\frac{|\theta|}{2}\right)^{m(\alpha,\beta)}F\\
&&\cdot \sup_{|\gamma|,|\sigma|\leq m(\alpha,\beta)}{\|x^\gamma D^\sigma 
b\|_{L^2}}
\cdot \sup_{|\gamma|\leq 2m(\alpha,\beta)}{\|D^\gamma f\|_{L^2}}\\
&\leq &C_{b,\alpha,\beta}(t'',t')^{n+1}\sup_{|\gamma|\leq 
2m(\alpha,\beta)}{\|D^\gamma f\|_{L^2}}
\end{eqnarray*}
with the above definitions.
This proves that~(\ref{E_n_maxL2Normft})
holds for all $n\in\bN$.

In combination with~(\ref{E_L2NormxalphaDbetatildeTnf}), we have thus proved 
that for each given $f\in\sS(\bR^s,\bC^N)$,
the Dyson series $\sum_{n=0}^\infty \tilde T^{(V)}_{n}(t'',t')f$ converges in 
all Schwartz norms and thus yields
again an element in $\sS(\bR^s,\bC^N)$. Therefore, $T^{(V)}_{t,t'}$ also maps 
$\sS(\bR^s,\bC^N)$ into itself and thence has (as mentioned, by Thm.\ X.69 in
\cite{ReedSimon2})
the properties claimed in statement (a) of the Lemma.
\\[6pt]
{\it Part} (b).
The arguments showing the claimed property are quite standard in view of the 
estimates given to establish part (a), so we will mainly sketch them.
Let $\mu$ be any $\mathbb{C}$-valued $C_0^\infty$-function on $\mathbb{R} \times
\mathbb{R}$ and denote by $Y$ the function
$$ (t,t',\underline{x}) \mapsto \mu(t,t') T_{t,t'}^{(V)}w(\underline{x}) \equiv
Y(t,t',\underline{x})\,.$$
Since $(t,t') \mapsto T_{t,t'}^{(V)}w \in L^2(\mathbb{R}^s,\mathbb{C}^N)$ is
continuous, $Y$ is in $L^2(\mathbb{R} \times \mathbb{R} \times \mathbb{R}^s,\mathbb{C}^N)$.
Thus we need only show that, if $\Delta_{s +2}$ denotes the Laplacian in
$s + 2$ dimensions,
$ (1 -\Delta_{s + 2})^J Y $ is again in $L^2(\mathbb{R} \times \mathbb{R} \times \mathbb{R}^s,
\mathbb{C}^N)$ for all $J \in \mathbb{N}$; the claimed statement on smoothness then follows by Sobolev's Lemma
(cf.\ Thm.\ IX.24 in \cite{ReedSimon2}). In turn, the required property follows from the
fact that $T_{t,t'}^{(V)}$ maps $\mathscr{S}(\mathbb{R}^s,\mathbb{R}^N)$ into itself
and that, as established in the proof of (a) or following immediately thereof, 
\begin{eqnarray*}
 {\rm (i)} & & \frac{\partial}{\partial t} T_{t,t'}^{(V)} w = i H_V(t) T_{t,t'}^{(V)} w\,, \\
    & & \frac{\partial}{\partial t'}T_{t,t'}^{(V)} w = -i T_{t,t'}^{(V)}H_V(t')w\,, \\
 {\rm (ii)} & & x^\alpha D^\beta H_V(t) x^{\alpha'}D^{\beta'} \ \text{is a continuous operator on}
\\ & & \mathscr{S}(\mathbb{R}^s,\mathbb{C}^N)\ \text{uniformly in}\ t\ \text{ranging over compact intervals.}
\end{eqnarray*}
{\it Part} (c).
It follows form parts (a) and (b) that there exits for any given Cauchy-datum $w \in 
\mathscr{S}(\mathbb{R}^s,\mathbb{C}^N)$ a solution $\varphi \in C^\infty(\mathbb{R}) \otimes \mathscr{S}(\mathbb{R}^s) \otimes \mathbb{C}^N$
of $D_V\varphi = 0$ with $P_{t'}\varphi = w$.

Recalling the definition 
$$ (v,w)_{\mathcal{D}} = \int_{\mathbb{R}^s} \delta_{AB}\overline{v}^A(\ulx)w^b(\ulx)\,d^s\ulx $$
for $v,w \in L^2(\mathbb{R}^s,\mathbb{C}^N)$, one finds
$$ \frac{d}{dt} (P_t \varphi,P_t \psi)_{\mathcal{D}} = (iH_V(t)P_t\varphi,P_t\psi)_{\mathcal{D}}
 + (P_t\varphi,iH_V(t)P_t\psi)_{\mathcal{D}} = 0$$
for all $\varphi,\psi \in C^\infty(\mathbb{R}) \otimes \mathscr{S}(\mathbb{R}^s) \otimes \mathbb{C}^N$
which are solutions of the equations $D_V\varphi = 0$ and $D_V\psi = 0$. Hence, in particular, 
if for two solutions $\varphi$ und $\psi$ there holds $(P_{t'}(\varphi - \psi),P_{t'}(\varphi - \psi))_{\mathcal{D}} = 0$
for some real $t'$, then it follows that  $(P_{t}(\varphi - \psi),P_{t}(\varphi - \psi))_{\mathcal{D}} = 0$ for all real $t$.
This shows that, if $\varphi$ and $\psi$ have the same Cauchy-datum on some Cauchy-hyperplane $\Sigma_{t'}$, then
actually $\varphi = \psi$. 
${}$ \hfill $\Box$    
 \end{proof}
On $C^\infty_0(\bR)\otimes \sS(\bR^s)\otimes \bC^N$ we can introduce the 
sesquilinear form
\begin{eqnarray*}
\langle f,h\rangle&=&\int_{\bR^n}\gamma_{0AB}(\bar 
f^B\star_{(q,p)}h^A)(x)d^nx\nonumber\\
&=&\int_{\bR^n}\gamma_{0AB}\bar f^B(x)h^A(x)d^nx,
\end{eqnarray*}
where the last equality follows from the tracial property of the Moyal product 
for $q=0$ (Lemma 2.1 (v) in~\cite{GGISV}) and its
obvious generalization to arbitrary $(q,p)$ due to the trivial case $p=0$ and 
the tensor product structure~(\ref{E_qpMoyalTensProd}).
Therefore we still have
\begin{equation*}
\langle Cf,Ch\rangle=-\langle h,f\rangle\quad (f,h\in C^\infty_0(\bR)\otimes 
\sS(\bR^s)\otimes \bC^N).
\end{equation*}
We recall the definitions $\Sigma_t=\{x=(x^0,x^1,\ldots,x^s)\in\bR^n:x^0=t\}$, 
$\cD_t=L^2(\Sigma_t,\bC^N)$,
\begin{eqnarray*}
(v,w)_\cD&=&\int_{\Sigma_t}\delta_{AB}(\bar 
v^A\star_{(q,p)}w^B)(\ulx)d^s\ulx\nonumber\\
&=&\int_{\Sigma_t}\bar v^A(\ulx)\delta_{AB}w^B(\ulx)d^s\ulx\quad (v,w\in\cD_t).
\end{eqnarray*}
and the property $(Cv,Cw)_\cD=(w,v)_\cD$ of a conjugation $C$ induced on each 
$\cD_t$ by the charge conjugation $C$ (same symbol).

For a subset $G$ of $n$ dimensional Moyal-deformed Minkowski spacetime the sets 
$J^\pm(G)$ and the notion of hyperbolicity are
defined in exactly the same way as in section~\ref{S_DiracField}.

Now we transfer the results collected in Proposition~\ref{Prop1} to our new 
setting. For this purpose the following definition is needed.
\begin{definition}
{\rm 
Let $K$ be a non-empty  subset of $\mathbb{R}^n$. Then let
$$ \kappa_-(K) = \inf\{x^0 : x = (x^0,\underline{x}) \in K\}\,, \ \ \
\kappa_+(K) = \sup\{x^0 : x = (x^0,\underline{x}) \in K\}\,,$$
and define
\begin{eqnarray*}
 \cT^+(K) & = & \{(x^0,\underline{x}) \in \mathbb{R}^n: x^0 \ge \kappa_-(K)\}\,,\\
\cT^-(K) & = & \{(x^0,\underline{x}) \in \mathbb{R}^n: x^0 \le \kappa_+(K)\}\,.
\end{eqnarray*} }

\end{definition}
\begin{center}
\begin{picture}(300,50)
\thicklines
\put(50,25.5){\oval(50,20)}
\put(200,24.5){\oval(50,20)}
\thinlines
\put(0,15){\line(1,0){100}}
\put(150,35){\line(1,0){100}}
\put(45,20){$\bf K$}
\put(195,20){$\bf K$}
\put(0,15){\line(1,3){8}}\put(10,15){\line(1,3){8}}\put(20,15){\line(1,3){8}}\put(30,15){\line(1,3){8}}
\put(40,15){\line(1,3){8}}\put(50,15){\line(1,3){8}}\put(60,15){\line(1,3){8}}\put(70,15){\line(1,3){8}}
\put(80,15){\line(1,3){8}}\put(90,15){\line(1,3){8}}
\put(150,35){\line(1,-3){8}}\put(160,35){\line(1,-3){8}}\put(170,35){\line(1,-3){8}}\put(180,35){\line(1,-3){8}}
\put(190,35){\line(1,-3){8}}\put(200,35){\line(1,-3){8}}\put(210,35){\line(1,-3){8}}\put(220,35){\line(1,-3){8}}
\put(230,35){\line(1,-3){8}}\put(240,35){\line(1,-3){8}}
\put(77,25){$\cT^+(K)$}
\put(227,15){$\cT^-(K)$}
\end{picture}
${}$\\
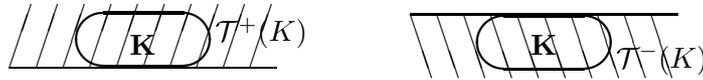
{\small \bf Figure 1: Sketch of the regions $\cT^\pm(K)$}
\end{center}

\begin{proposition}\label{Prop1forMoyal}
\hfill
\begin{enumerate}
\item
$\langle D_Vf,h\rangle=\langle f,D_Vh\rangle\quad(f,h\in C^\infty_0(\bR)\otimes 
\sS(\bR^s)\otimes \bC^N)$
\item\label{Prop1forMoyal_b}
There is a unique pair of continuous linear maps
\[
R^\pm_V:C^\infty_0(\bR)\otimes \sS(\bR^s)\otimes \bC^N\rightarrow 
C^\infty(\bR)\otimes \sS(\bR^s)\otimes \bC^N
\]
having the properties
\begin{eqnarray*}
&&\quad\quad D_VR^\pm_Vf=f=R^\pm_VD_Vf\text{ and }\nonumber\\
&&\supp R^\pm_Vf \subset \cT^\pm(\supp f)\quad(f\in 
C^\infty_0(\bR)\otimes \sS(\bR^s)\otimes \bC^N).
\end{eqnarray*}
\item
$CR^\pm_V=R^\pm_VC$
\item
Writing $R_V=R^+_V-R^-_V$, the form
\begin{equation*}
(f,h)_V=\langle f,iR_Vh\rangle
\end{equation*}
is a sesquilinear form on $C^\infty_0(\bR)\otimes \sS(\bR^s)\otimes \bC^N$, and 
$C$ is a conjugation for this form:
\begin{equation*}
(Cf,Ch)_V=(h,f)_V=\overline{(f,h)_V}\quad (f,h\in C^\infty_0(\bR)\otimes 
\sS(\bR^s)\otimes \bC^N).
\end{equation*}
\item\label{I_e_Prop1forMoyal}
For each $t\in\bR$ it holds that
\begin{equation*}
(f,h)_V=(P_tR_Vf,P_tR_Vh)_{\cD},\quad (f,h\in C^\infty_0(\bR)\otimes 
\sS(\bR^s)\otimes \bC^N),
\end{equation*}
where $P_t:C^\infty(\bR)\otimes \sS(\bR^s)\otimes \bC^N\rightarrow 
\sS(\Sigma_t,\bC^N)$ is the map given by
\begin{equation*}
P_t:\varphi\mapsto \varphi(t,\cdot)
\end{equation*}
for $\varphi:(x^0,\ulx)\mapsto \varphi(x^0,\ulx)$ in $C^\infty(\bR)\otimes 
\sS(\bR^s)\otimes \bC^N$, $\ulx=(x^1,\ldots,x^s)$. Hence, $(\cdot,\cdot)_V$
is positive-semidefinite on $C^\infty_0(\bR)\otimes \sS(\bR^s)\otimes \bC^N$.
\item\label{I_f_Prop1forMoyal}
Let $E_V$ be the subspace of all $f\in C^\infty_0(\bR)\otimes \sS(\bR^s)\otimes 
\bC^N$ so that $(f,f)_V=0$, and let
$\cK_V$ be the Hilbert space arising as completion of
$(C^\infty_0(\bR)\otimes \sS(\bR^s)\otimes \bC^N)/E_V$ with respect to the 
scalar product induced by $(\cdot,\cdot)_V$ (which
will be denoted by the same symbol). The quotient map
$C^\infty_0(\bR)\otimes \sS(\bR^s)\otimes \bC^N\rightarrow 
(C^\infty_0(\bR)\otimes \sS(\bR^s)\otimes \bC^N)/E_V$ will
be written
\begin{equation*}
f\mapsto [f]_V.
\end{equation*}
Then for each $t\in\bR$, the map
\begin{equation*}
Q_{V,t}:[f]_V\mapsto P_tR_Vf
\end{equation*}
extends to a unitary map from $\cK_V$ onto $\cD_t$.
\item\label{I_g_Prop1forMoyal}
Let $G$ be an open time-slice of $n$ dimensional Moyal-deformed Minkowski 
spacetime, i.e. $G=\{(x^0,x^1,\ldots,x^s):
\lambda_1<x^0<\lambda_2\}$ for some real numbers (or infinite) 
$\lambda_1<\lambda_2$, and suppose that $V_1$ and $V_2$ are potentials
of the form described by~(\ref{E_DefVMoyal}),~(\ref{E_DefcMoyal}), and that 
$V_1=V_2$ on $G$. Then
\begin{equation*}
R^\pm_{V_1}f=R^\pm_{V_2}f\text{ on }G\text{ for all }f\in 
C^\infty_0((\lambda_1,\lambda_2))\otimes \sS(\bR^s)\otimes \bC^N.
\end{equation*}
\end{enumerate}
\end{proposition}

\begin{proof}[Sketch of proof]
\hfill
\begin{enumerate}
\item
Clearly the only difference compared to Proposition~\ref{Prop1}~(\ref{I_a_Prop1}) 
is the partial claim
$\langle Vf,h\rangle=\langle f,Vh\rangle$. To show this, calculate
(only for the case $V=V_\text{(i)}$; the other one is completely analogous)
\begin{eqnarray*}
\langle 
Vf,h\rangle&=&\int_{\bR^n}\gamma_{0AB}\overline{(Vf)^B}(x)h^A(x)d^nx=\gamma_{0AB
}((Vf)^B,h^A)_{L^2}\\
&=&\gamma_{0AB}(c\star_{(q,p)}f^B+f^B\star_{(q,p)}c,h^A)_{L^2}\\
&=&\gamma_{0AB}(f^B,c\star_{(q,p)}h^A+h^A\star_{(q,p)}c)_{L^2}\\
&=&\int_{\bR^n}\gamma_{0AB}\bar f^B(x)(Vh)^A(x)d^nx=\langle f,Vh\rangle,
\end{eqnarray*}
since $c$ is real-valued, using Lemma~\ref{Lemma_HermPropMoyal}.
\item
To prove this, the fundamental solutions will be constructed explicitly.
Using the notation $f_{t'}(\,.\,) = f(t',\,.\,)$, define for $f \in 
C_0^\infty(\mathbb{R}) \otimes \mathscr{S}(\mathbb{R}^s) \otimes \mathbb{C}^N$, 
  \[
(R^\pm_Vf)(t,\ulx)=\pm i\gamma^0\int  \theta(\pm(t - t'))T^{(V)}_{t',t}f_{t'}(\ulx)\,dt'\,,
   \]
where $\theta$ is the Heaviside step function.
Note that the integral is well-defined since $f_{t'}(\ulx)$ has compact support in $t'$.
By Lemma \ref{L_ExUniqInitValProbMoyal} it follows that 
$(t,\ulx) \mapsto (R_V^\pm f)(t,\ulx)$ is $C^\infty$, and of Schwartz type with respect to
$\ulx$. Using standard arguments, and exploiting the properties of $T^{(V)}_{t',t}$ given in
Lemma \ref{L_ExUniqInitValProbMoyal}, one proves that
$D_V R^\pm_V f = f = R^\pm_V D_V f$. The next step consists in showing that 
${\rm supp}(R_V^\pm f) \subset \cT^\pm({\rm supp}\,f)$. To this end, suppose that
$(x^0,\ulx) \notin \cT^+({\rm supp}\,f)$. Then $x^0 < \kappa_-({\rm supp}\,f)$ and therefore
$$ \theta(x^0 - t')T_{t',t}^{(V)}f_{t'}(\ulx) = 0$$
for all values of $t,t'$ and $\ulx$. To see this, note that if $t' \ge x^0$, then
$\theta(x^0 - t') = 0$, and if $t' < x^0 < \kappa_-({\rm supp}\,f)$, then
$f_{t'}(\,.\,) = 0$ by the definition of $\kappa_-({\rm supp}\,f)$. Hence, it holds that
$(R^+_Vf)(x^0,\ulx) = 0$ if $(x^0,\ulx) \notin \cT^+({\rm supp}\,f)$, implying that
${\rm supp}(R^+_V f) \subset \cT^+({\rm supp}\,f)$. The inclusion
${\rm supp}(R_V^-f) \subset \cT^-({\rm supp}\,f)$ can be shown in an analogous manner.
The uniqueness property of the fundamental solutions follows by a standard argument owing
to the well-posedness of the Cauchy-problem for the Dirac-equation $D_V\varphi = 0$.
\item
This is a consequence of the uniqueness of the $R^\pm_V$ together with 
$CD_V=D_VC$.
\item
Analogous to Proposition~\ref{Prop1}; the crucial point is the validity of
$\langle R_Vh,f\rangle =-\langle h,R_Vf\rangle $, for which the argument is
again similar to the proof of Thm.\ 2.1 in \cite{Dimock}.
\item
The argument is the same as in Proposition~\ref{Prop1}.
\item
The proof is identical to the corresponding statement (g) of Proposition~\ref{Prop1}.
The modification lies in the generalized class
of Cauchy data. The choice of a partition of unity
$\chi_\pm$ depending only on the time-coordinate
$x^0$ is needed at this point.
\item
Apart from the modified assumption on the subset $G$, providing an adjusted 
time-direction behaviour for the
non-commutative case, this can obviously proved the same way as 
Proposition~\ref{Prop1},~(\ref{I_g_Prop1}).\qed
\end{enumerate}
\end{proof}

Analogous to Section~\ref{S_DiracField} the self-dual CAR-algebra $\fF(\cK_V,C)$ 
is generated by the $\bC$-linear
``abstract field operators'' $\Psi(f)=\Psi_V(f)=B_V([f]_V)$, for
$f\in C^\infty_0(\bR)\otimes \sS(\bR^s)\otimes \bC^N$, obeying the relations
\begin{eqnarray*}
\Psi(f)^*&=&\Psi(Cf),\nonumber\\
\{\Psi(f)^*,\Psi(h)\}&=&2(f,h)_V\eins,\nonumber\\
\Psi(D_Vf)&=&0.
\end{eqnarray*}
Note that because of the trivial action of $V$ with respect to the first 
coordinate (the time) it holds that
$D_Vf\in C^\infty_0(\bR)\otimes \sS(\bR^s)\otimes \bC^N$.

Again this construction can be carried out for ``local subspaces'' of $\cK_V$ as 
well (cf.\ Section~\ref{S_DiracField}),
and we get $\fF(\cK_V^G,C)$ generated by $B_V^G([f]_V^G)$,
but this time only for an open time-slice $G$ of $n$ dimensional Moyal-deformed 
Minkowski spacetime and no longer
for arbitrary hyperbolic subsets.
Being mainly a consequence of Proposition~\ref{Prop1}~(\ref{I_g_Prop1}), 
Lemma~\ref{Lemma2} can be transferred
almost unchanged.

\begin{lemma}\label{Lemma2forMoyal}
Suppose that $G$ is a hyperbolic neighbourhood of a Cauchy hyperplane in $n$ 
dimensional Moyal-deformed Minkowski spacetime,
of the form as in Proposition~\ref{Prop1forMoyal}~(\ref{I_g_Prop1forMoyal}). 
Moreover,
suppose that $V_1$ and $V_2$ are two potentials of 
type~(\ref{E_DefVMoyal}),(\ref{E_DefcMoyal}), which coincide on the region $G$. 
Then
\begin{enumerate}
\item
The map
\begin{equation*}
u^G_{V_1,V_2}:[f]^G_{V_1}\mapsto [f]_{V_2},\quad f\in 
C^\infty_0((\lambda_1,\lambda_2))\otimes \sS(\bR^s)\otimes \bC^N
\end{equation*}
extends to a unitary between $\cK^G_{V_1}$ and $\cK_{V_2}$ commuting with the 
charge conjugation $C$.
\item
There is a $*$-algebra isomorphism
\begin{equation*}
\alpha^G_{V_1,V_2}:\fF(\cK^G_{V_1},C)\rightarrow \fF(\cK_{V_2},C)
\end{equation*}
induced by
\begin{equation*}
\alpha^G_{V_1,V_2}\left(B^G_{V_1}([f]^G_{V_1})\right)=B_{V_2}([f]_{V_2}),\quad 
f\in C^\infty_0((\lambda_1,\lambda_2))\otimes \sS(\bR^s)\otimes \bC^N.
\end{equation*}
\end{enumerate}
\end{lemma}

Since our potential operator $V$ (see~(\ref{E_DefVMoyal}),(\ref{E_DefcMoyal})) 
was chosen to be a compactly supported multiplication
operator with respect to the time coordinate, we can maintain exactly the same 
geometrical setting as in Section~\ref{S_DiracField}
involving the same time-slice $\{(x^0,x^1,\ldots,x^s):\lambda_-<x^0<\lambda_+\}$ 
for some real numbers $\lambda_-<\lambda_+$ and the
same regions $G_+,G_-$ being hyperbolic neighbourhoods of the Cauchy hyperplanes 
$\Sigma_+,\Sigma_-$. As a result we again arrive at an
automorphism
\begin{eqnarray}
&&\beta_V:\fF(\cK_0,C)\rightarrow \fF(\cK_0,C),\nonumber\\
&&\beta_V=\alpha_{0,-}\circ \alpha_{V,-}^{-1}\circ \alpha_{V,+}\circ 
\alpha_{0,+}^{-1}.\label{E_DefBetaVforMoyal}
\end{eqnarray}

As in Sec.~\ref{S_DiracField} we can again define $T_{sc}^{(V)}$ as the scattering 
transformation on
$\mathcal{D}_0 \simeq L^2(\Sigma_0,\mathbb{C}^N)$ ($\Sigma_0 \simeq 
\mathbb{R}^s$), the space of Cauchy data for the
Dirac equation at coordinate-time $t = 0$, by setting
$$ T_{sc}^{(V)} = T_t^{-1} \circ T_{t,t'}^{(V)} \circ T_{t'} $$
for $t> \lambda_+$, $t' < \lambda_-$ (recall that the interval $[\lambda_-
,\lambda_+]$ is the time-support of
the potential term $V$).
As before in Sec.~\ref{S_DiracField}, $T_t$ denotes the ``free'' evolution of the Dirac equation 
without potential
term, coinciding with $
\left. T_{t,0}^{(V)}\right|_{V = 0}$. In consequence one obtains, exactly as
in eqn.~\eqref{E_DeftauScMor}, an induced automorphism $\tau_{sc}^{(V)}$ on the CAR-algebra 
$\fF(\cD_0,C)$, given by
$$ \tau^{(V)}_{sc}(B_{\cD_0}(v)) = B_{\cD_0}(T_{sc}^{(V)}v)\,, \quad v \in 
\cD_0\,,$$
where the $B_{\cD_0}(v)$ are the generators of $\fF(\cD_0,C)$. Again, there is a 
canonical identification
between the CAR algebras $\fF(\cK_0,C)$ and $\fF(\cD_0,C)$,
$$ \varrho(B_0([f]_0)) = B_{\cD_0}(Q_0([f]_0))\,.$$
\par
In the next section we will study the problem of unitary implementability of 
$\beta_V$ in the
Fock-vacuum-representation of $\fF(\cK_0,C)$, and the following Lemma, which
is the counterpart of Lemma~\ref{L_betaV_tauscV_intertw} for the case of potential $V$ defined as
in~(\ref{E_DefVMoyal}) and~(\ref{E_DefVMoyalAlt}),
guarantees that unitary implementability of $\tau_{sc}^{(V)}$ in the vacuum 
representation is just
the equivalent problem.

\begin{lemma}\label{L_betaV_tauscV_intertwMoyal}
The morphism $\beta_V$ of $\fF(\cK_0,C)$ defined in~(\ref{E_DefBetaVforMoyal}) 
and the scatte\-ring morphism $\tau^{(V)}_\text{sc}$
(defined like~(\ref{E_DeftauScMor})) describing the potential scattering of the 
quantized Dirac field at the level of the Cauchy-data
CAR-algebra $\fF(\cD_0,C)$ (with $\cD_0=L^2(\Sigma_0,d^sx)$) are intertwined by 
the CAR-algebra isomorphism
$\varrho:\fF(\cK_0,C)\rightarrow \fF(\cD_0,C)$ defined
in~(\ref{E_varrhoDef}), i.e. it holds that
\begin{equation} \label{E_rhoInterwRel}
\varrho\circ \beta_V=\tau^{(V)}_\text{sc}\circ \varrho.
\end{equation}
\end{lemma}

\section{Moyal-deformed Minkowski spacetime: Scattering of the Dirac field in 
the vacuum representation
and implementability of the scattering transformation}\label{S_MoyalMink_ImplScatt}

In the present section we will prove unitary implementability of the scattering 
transformation
$\beta_V$ on $\fF(\cK_0,C)$ for the ``Moyal-Minkowski-potentials'' $V$ defined 
in~(\ref{E_DefVMoyal}) and~(\ref{E_DefVMoyalAlt}),
in the Fock-vacuum representation 
$(\cH^\text{vac},\pi^\text{vac},\Omega^\text{vac})$ of the vacuum
state $\omega^{\rm vac}$ on $\fF(\cK_0,C)$. Owing to Lemma~\ref{L_betaV_tauscV_intertwMoyal}, this is 
equivalent to the 
problem of unitary implementability of the scattering transformation 
$\tau_{sc}^{(V)}$ on 
$\fF(\cD_0,C)$ in the Fock-vacuum representation 
$(\cH^{p_+},\pi^{p_+},\Omega^{p_+})$ of
$\omega^{p_+}$ (where $\omega^{\rm vac} = \omega^{p_+} \circ \varrho$), the pure, 
quasifree
ground state on $\fF(\cD_0,C)$ with respect to the time-evolution induced by the 
Hamiltonian 
$H_0$ of~(\ref{E_Hzero}) on the domain $\sS(\bR^s,\bC^N)\subset L^2(\bR^s,\bC^N)$. 
Recall that
$p_+$ is the spectral projection of $H_0$ corresponding to the spectral interval 
$[m,\infty)$, and that
the conjugation $C$ intertwines $p_+$ and $1 - p_+$ and hence $p_+$ is a basis 
projection on
$(\cD_0,C)$ according to \cite{Araki}. 
A well-established criterion for unitary implementability of $\tau^{(V)}_{sc}$ 
in the
Fock-vacuum representation is that $[p_+,T_{sc}^{(V)}]$ is Hilbert-Schmidt
\cite{Araki,ShaleStinespring}. If and only if this is the case, then there is a unitary operator
$S_{\tau_{sc}^{(V)}}$ on $\mathcal{H}^{p_+}$ such that 
\begin{equation} \label{E_Stau}
S_{\tau^{(V)}_\text{sc}} \pi^{p_+}(B_{\cD_0}(v)) S_{\tau^{(V)}_\text{sc}}^{-1} =
 \pi^{p_+}\left( \tau_\text{sc}^{(V)}(B_{\cD_0}(v)) \right) = 
\pi^{p_+}(B_{\cD_0}(T_\text{sc}^{(V)}v))\,,
\quad v \in \cD_0\,,
\end{equation}
which is just what it means to say that $\tau_{sc}^{(V)}$ is unitarily 
implementable.
The condition that $[p_+,T_{sc}^{(v)}]$ is Hilbert-Schmidt is equivalent to the 
condition
that $[\varepsilon,T^{(V)}_{sc}]$ is Hilbert-Schmidt as an operator on
$L^2(\mathbb{R}^s,\mathbb{C}^N)$, where 
$\varepsilon = {\rm sign}(H_0) = H_0/|H_0|$ is the sign function of $H_0$ in the
sense of the functional calculus, since $p_+=(\eins+\varepsilon)/2$. In an interesting work, Langmann and 
Mickelsson
\cite{LangmannMickelsson} have shown that certain conditions on the potential 
term
$V(t)$ (cf.\ eqn.\ \eqref{E_Hvee}) are sufficient to conclude that 
$[\varepsilon,T_{sc}^{(V)}]$ is Hilbert-Schmidt. Their argument is interesting 
as it
involves a non-local regularization of the interacting dynamics which 
nevertheless leads to 
the same scattering transformation $T_{sc}^{(V)}$. We refer to
\cite{LangmannMickelsson} for details and present only the relevant conditions,
adapted to our notation.

Let the interaction potential $W(t) = \gamma^0 V(t)$ in the Hamiltonian 
$H_V(t)$ of eqn.\ \eqref{E_Hvee} have the following properties:
\begin{itemize}
\item[(I)] $W(t)$ is a bounded operator on $L^2(\mathbb{R}^s,\mathbb{C}^N)$ 
for each $t \in \mathbb{R}$, such that $t \mapsto W(t)$ is
$C^\infty$. 
\item[(II)]
There is a core for $H_0$, contained in the $C^\infty$-domain of $H_0$, which
is left invariant by all $W(t)$ and $(\partial_t)^kW(t)$ $(k \in \mathbb{N})$
and by all $C^\infty$ functions of $H_0$ which, together with all their 
derivatives,
are polynomially bounded.
\item[(III)] 
There is some $\nu \in \mathbb{N}_0$ so that $|H_0|^{-\nu}W(t)$ and
$|H_0|^{-\nu} (\partial_t)^k W(t)$, $k= 1,\ldots,\nu,$ are Hilbert-Schmidt 
operators for
all $t$.
\item[(IV)]
Defining $\delta_{|H_0|}(A) = [ |H_0|,A]$, it holds that 
$\delta^n_{|H_0|}(W(t))$ and $\delta^n_{|H_0|}((\partial_t)^kW(t))$, $k = 
1,\ldots,\nu$,
are bounded operators on $L^2(\mathbb{R}^s,\mathbb{C}^N)$ for all $n \in \mathbb{N}$
$(t \in \mathbb{R})$.
\item[(V)]
$|H_0|^{-\nu}\delta^n_{|H_0|}(W(t))$ and $|H_0|^{-
\nu}\delta^n_{|H_0|}((\partial_t)^kW(t))$, $k = 1,\ldots,\nu$, are Hilbert-
Schmidt operators for all $n \in \mathbb{N}$ $(t \in \mathbb{R})$.
\item[(VI)]
$W(t) = 0$ if $t < \lambda_-$ and if $t> \lambda_+$ for some 
real numbers $\lambda_- < \lambda_+$.
\end{itemize}
We cite the result relevant for our purposes.
\begin{theorem} {\bf (Langmann and Mickelsson \cite{LangmannMickelsson})}
If the interaction term $W(t) = \gamma^0V(t)$ in $H_V(t)$ (cf.\ \eqref{E_Hvee})
satisfies the conditions ${\rm (I)}\ldots {\rm (VI)}$, then 
$[\varepsilon,T_{sc}^{(V)}]$ is Hilbert-Schmidt, and hence 
$\tau^{(V)}_{sc}$ is implementable in the vacuum-representation
$(\mathcal{H}^{p_+},\pi^{p_+},\Omega^{p_+})$ of the Dirac field.
\end{theorem}
Consequently, what we will now set out to demonstrate is
\begin{proposition} \label{E_Smatrix}
Let $V(t)$ be any of the $V_{\rm (i)}(t)$ or $V_{\rm (ii)}(t)$ defined in
\eqref{E_Veeone} and \eqref{E_Veetwo}, with $a \in 
C_0^\infty(\mathbb{R},\mathbb{R})$
and $b \in \mathscr{S}(\mathbb{R}^s,\mathbb{R})$. Then $W(t) = \gamma^0V(t)$
fulfills the criteria ${\rm (I)}\ldots {\rm (VI)}$ above. Therefore,
$\tau_{sc}^{(V)}$ is unitarily implementable in the vacuum representation, so 
that
there is a unitary $S_{\tau_{sc}^{(V)}}$ on $\mathcal{H}^{p_+}$ such that
\eqref{E_Stau} holds.
\end{proposition}
{\bf Proof}\ \ Observing that (cf.\ \eqref{E_Veeone},\eqref{E_Veetwo})
$$ W(t)f = \tilde{a}(t){\sf v}f \quad (f \in L^2(\mathbb{R}^s,\mathbb{R}^N))$$ 
with $\tilde{a}(t) = a(t)$ for $V=V_{( \rm i)}$ and $\tilde{a}(t) = a(t)^2$ for
$V = V_{({\rm ii})}$,
the time-dependence of $W(t)$ is trivial in the context of conditions
(I)$\ldots$(VI), and they need only be checked for the time-independent
operator ${\sf v}$, which is
\begin{eqnarray*}
{\sf v}f & = & {\sf v}_{(i)}f = \gamma^0({\sf L}_b f + {\sf R}_b f) \ \ \text{or} 
\\
{\sf v}f & = & {\sf v}_{(ii)}f = \gamma^0({\sf L}_b{\sf R}_b f)\,.
\end{eqnarray*}
We note also that the multiplication with $\gamma^0$ is irrelevant for
checking the conditions (I)$\ldots$(VI) since $\gamma^0$ commutes with
$|H_0|$. Thus, the conditions need only be checked for 
${\sf L}_b$, ${\sf R}_b$ and ${\sf L}_b{\sf R}_b$. A quick inspection shows that
the conditions are algebraic in the sense that, if they hold for ${\sf L}_b$ and
${\sf R}_b$, then they hold also for the operator product ${\sf L}_b{\sf R}_b$.
As will become clear from the Fourier-representations of ${\sf L}_b$ and ${\sf 
R}_b$
(see below), checking the conditions for ${\sf R}_b$ is completely analogous to 
the
case of ${\sf L}_b$, so it is sufficient to show that the conditions are 
fulfilled for the
operator ${\sf L}_b$.

Let, for $g \in L^2(\mathbb{R}^s,\mathbb{C})$, the Fourier-transform be defined 
by
\begin{equation} \label{E_FT}
(F g)(\underline{k}) = \hat{g}(\underline{k}) =
\frac{1}{(2\pi )^{s/2}} \int_{\mathbb{R}^s} g(\underline{y}){\rm e}^{-
i\underline{y} \cdot
\underline{k}}\,d^s\underline{y}\,\text{;}
\end{equation}
this definition is extended componentwise to elements in 
$L^2(\mathbb{R}^s,\mathbb{C}^N)$.
It is easy to see that
\begin{eqnarray} \label{E_FLbFin}
F {\sf L}_b F^{-1} \hat{g} (\underline{k}) & = &
 \frac{1}{(2\pi)^{s/2}} \int_{\mathbb{R}^s} \hat{b}(\underline{k} -\underline{u})
{\rm e}^{i\underline{u} \cdot
\underline{M}\underline{k}} \hat{g}(\underline{u})\,d^s\underline{u}\,,\\
\label{E_FRbFin} 
F {\sf R}_b F^{-1} \hat{g} (\underline{k}) & = &
 \frac{1}{(2\pi)^{s/2}} \int_{\mathbb{R}^s} \hat{b}(\underline{k} -\underline{u})
{\rm e}^{-i\underline{u} \cdot
\underline{M}\underline{k}} \hat{g}(\underline{u})\,d^s\underline{u}\,,
\end{eqnarray}
where $\underline{M}$ is the block-entry $M_{(q +p -1) \times (q + p - 1)}$ in 
the
matrix \eqref{E_Moyalmatrix}. Moreover, one finds
\begin{eqnarray} \label{E_FHFin}
F H_0 F^{-1} \hat{g}(\underline{k}) & = &\widehat{H}(\underline{k}) 
\hat{g}(\underline{k})\,,\\ \nonumber
F |H_0 |^{\kappa} F^{-1} \hat{g}(\underline{k}) & = 
&|\widehat{H}(\underline{k})|^{\kappa} \hat{g}(\underline{k})\, \quad (\kappa 
\in \mathbb{R}),
\end{eqnarray}
with
\begin{equation} \label{E_Hhat}
 \widehat{H}(\underline{k}) = -i\gamma^0 \gamma^j\underline{k}_j + \gamma^0 
m\,, \quad
| \widehat{H}(\underline{k})| = ( |\underline{k}|^2 + m^2)^{1/2} \quad 
(\underline{k} \in
\mathbb{R}^s)\,.
\end{equation}
This implies 
\begin{equation} \label{E_kerofdn}
F \delta_{|H_0|}^n({\sf L}_b)F^{-1} \hat{g}(\underline{k})
= 
\frac{1}{(2\pi)^{s/2}} \int_{\mathbb{R}^s} ( |\widehat{H}_0(\underline{k})| -
|\widehat{H}_0(\underline{u})|)^n \hat{b}(\underline{k} - \underline{u})
 {\rm e}^{i \underline{u} \cdot \underline{M}\underline{k}} 
\hat{g}(\underline{u})\,
 d^s\underline{u}
\end{equation} 
for all $n \in \mathbb{N}$ and all $\hat{g} \in 
\mathscr{S}(\mathbb{R}^s,\mathbb{C}^N)$,
and similarly 
\begin{eqnarray} 
\lefteqn{F |H_0|^{-\nu}\delta_{|H_0|}^n({\sf L}_b)F^{-1} \hat{g}(\underline{k})}\nonumber\\\label{E_kerofdnnu}
&=& 
\frac{1}{(2\pi)^{s/2}} \int_{\mathbb{R}^s} 
\frac{( |\widehat{H}_0(\underline{k})| -
|\widehat{H}_0(\underline{u})|)^n}{|\widehat{H}_0(\underline{k})|^\nu} 
\hat{b}(\underline{k} - \underline{u})
 {\rm e}^{i \underline{u} \cdot \underline{M}\underline{k}} 
\hat{g}(\underline{u})\,
 d^s\underline{u}\,.
\end{eqnarray} 
The discussion in Sec.\ 4 shows that ${\sf L}_b$ is bounded. Moreover, we see 
from 
the Fourier-representation of $H_0$ that 
$\mathscr{S}(\mathbb{R}^s,\mathbb{C}^N)$ is a core
with the properties demanded in (II). In view of what we observed previously,
(IV) is proved once we have shown that
$\delta^n_{|H_0|}({\sf L}_b)$ is a bounded operator for all $n \in \mathbb{N}$. 
We use that,
given $n$,
there are constants $\alpha,\beta >0$ such that 
\begin{equation} \label{E_bound1}
|\,|\widehat{H}_0(\underline{k})| - |\widehat{H}_0(\underline{u})|\,|^n
\le \alpha|\underline{k} - \underline{u}|^{2n} + \beta \quad 
(\underline{k},\underline{u} \in
\mathbb{R}^s)\,.
\end{equation}
Now the integral kernel in \eqref{E_kerofdn} is actually a matrix, and owing to 
\eqref{E_bound1},
each of its entries has a modulus which can be bounded by
\begin{equation}
\alpha |\widehat{(-\Delta^n b)}(\underline{k} - \underline{u})| + \beta 
|b(\underline{k} - \underline{u})|\,,
\end{equation}
where $\Delta$ denotes the Laplace operator.
This shows that $F \delta_{|H_0|}^n({\sf L}_b)F^{-1}$ has an operator norm which 
can be
dominated by a constant times
\begin{equation}
\sup_{\underline{u}}\,\left( \int_{\mathbb{R}^s} [\alpha |\widehat{(-\Delta^n 
b)}(\underline{k} - \underline{u})| + \beta |b(\underline{k} - 
\underline{u})|]^2
 \,d^s\underline{k}\right)^{1/2}
\end{equation}
which is finite since $b$ is of Schwartz type. It remains to check conditions 
(III) and (V). To this end, we observe that the integral
kernel in \eqref{E_kerofdnnu} is a matrix where each entry has a modulus which,
for given $n$ and $\nu$, can be
estimated by a constant times
\begin{equation}
\frac{1}{(|\underline{k}|^2 + m^2)^{\nu/2}}[\alpha |\widehat{-\Delta^n 
b}(\underline{k} - \underline{u})| + \beta |\widehat{b}(\underline{k} - 
\underline{u})|]\,.
\end{equation}
One can obviously choose $\nu$ large enough so that this expression is, for each 
$n$, 
square integrable over $(\underline{k},\underline{u}) \in \mathbb{R}^s \times 
\mathbb{R}^s$,
using again that $b$ is a Schwartz function.
This shows that a number $\nu$ can be chosen large enough so that
 $F |H_0|^{-\nu}\delta_{|H_0|}^n({\sf L}_b)F^{-1}$
is Hilbert-Schmidt for all $n$ (including $n=0$, ie.\ $F |H_0|^{-\nu}{\sf 
L}_b F^{-1}$).

Finally we remark that condition (VI) is clearly fulfilled since it was assumed 
that
$a \in C_0^\infty(\mathbb{R},\mathbb{R})$.
${}$ \hfill $\Box$
\\[10pt] 
We will also show in this chapter that the generator of the S-matrix $S_{\tau_{sc}^{(V)}}$ with
respect to variations of $V$ exists as an essentially selfadjoint operator in the sense of
derivations. Using the fact that $S_V^M$ and $S_{\tau_{sc}^{(V)}}$ are intertwined by a unitary
establishing the equivalence between $\pi^{\rm vac}$ and
$\pi^{p_+} \circ \varrho$, this will allow the conclusion that also the generator
of the S-matrix $S_V^M$ with respect to variations of $V$ exists as an essentially selfadjoint operator on
a suitable domain. 

In preparing the proof of the assertion we aim to establish, we need an auxiliary result.
\begin{proposition} \label{Prop_dT}
Let $V$ be any of the operators $V_{(0)}$, $V_{(\rm i)}$ or $V_{(\rm ii)}$, where
$$ V_{(0)}f (x)  =  c(x) f(x) \quad (f \in \mathscr{S}(\mathbb{R}^n,\mathbb{C}^N),\ x \in \mathbb{R}^n)\,,$$
and $V_{(\rm i)}$ and $V_{(\rm ii)}$ are defined as \eqref{E_DefVMoyal} and
\eqref{E_DefVMoyalAlt}, with $c(x) = a(x^0)b(\underline{x})$ for 
any $a \in C_0^\infty(\mathbb{R},\mathbb{R})$ and $b \in \mathscr{S}(\mathbb{R}^s,\mathbb{R})$, 
$x = (x^0,\underline{x}) \in \mathbb{R}^{1 + s} = \mathbb{R}^n$.
\\[4pt]
Then 
\begin{itemize}
\item[{\rm (a)}]
Defining
\begin{equation}
dT^{(V)}_{sc} v  =  -i\left. \frac{d}{d\lambda}\right|_{\lambda = 0} T_{sc}^{(\lambda V)}v \quad (v \in \mathscr{S}(\mathbb{R}^s,\mathbb{C}^N))\,,
\end{equation}
it holds that 
\begin{equation}
dT^{(V)}_{sc} v(\underline{x}) =  -\int_{-\infty}^{\infty} \tilde{a}(t){\rm e}^{iH_0t} {\sf v} {\rm e}^{-iH_0t}v (\underline{x})\,dt
\quad (\underline{x} \in \mathbb{R}^s)\,,
\end{equation}
where ${\sf v}$ is either of the operators ${\sf v}_{(0)}$, ${\sf v}_{(\rm i)}$ or
${\sf v}_{(\rm ii)}$, with ${\sf v}_{(\rm i)}$ and ${\sf v}_{(\rm ii)}$ defined in \eqref{E_lcvi} and \eqref{E_lcvii} and 
$$ {\sf v}_{(0)} v(\underline{x}) = \gamma^0 b(\underline{x})v(\underline{x}) \,,$$
and with $\tilde{a}(t) = a(t)$ in the cases $V = V_{(0)},V_{(\rm i)}$, while
$\tilde{a}(t) = a(t)^2$ in case $V = V_{(\rm ii)}$.
The operator $dT^{(V)}_{sc}$ is bounded and selfadjoint on $L^2(\mathbb{R}^s,\mathbb{C}^N)$.
\item[{\rm (b)}]
The commutator
$$ [dT_{sc}^{(V)},p_+] $$
is a Hilbert-Schmidt operator on $L^2(\mathbb{R}^s,\mathbb{C}^N)$. Here, $p_+$ denotes again the spectral projection of
$H_0$ corresponding to the spectral interval $[m,\infty)$.
\end{itemize}
\end{proposition}
Before we give the proof of that Proposition (see towards the end of this Section), we explain how this result allows it to conclude the statements
made in Sec.\ 5, which re-appear below in the eqns.\
\eqref{E_Gamling}, \eqref{E_88a} and \eqref{E_89a}.

Recall that $\cH^{p_+} = \mathcal{F}_+(p_+L^2(\mathbb{R}^s,\mathbb{C}^N))$ is the Fermionic Fock-space with one-particle space
$p_+L^2(\mathbb{R}^s,\mathbb{C}^N)$. For $v \in L^2(\mathbb{R}^s,\mathbb{C}^N)$, define the field operators
$$ \underline{\psi}(v) = A(p_+ C v) + A^+(p_+v) \, , \quad v \in \mathcal{D}_0 \equiv L^2(\mathbb{R}^s,\mathbb{C}^N) \,,$$
where $A$ and $A^+$ denote the Fermionic annihilation and creation operators. In other words,
\begin{equation} \label{E_underlinepsi}
\underline{\psi}(v) = \pi^{p_+}(B_{\mathcal{D}_0}(v)) \,.
\end{equation}
By $\underline{\mathcal{F}}$ we denote the $*$-algebra generated by all $\underline{\psi}(v)$ and the unit operator, and
we set $\underline{\mathcal{W}} = \underline{\mathcal{F}}\Omega^{p_+}$. Now consider an orthonormal basis $\{\chi^+_j\}_{j \in \mathbb{N}}$
of $p_{+} L^2(\mathbb{R}^s,\mathbb{C}^N)$;
then $\chi_j^- = C \chi_j^+$ is an orthonormal basis of $p_{-} L^2(\mathbb{R}^s,\mathbb{C}^N)$. 
If $[dT_{sc}^{(V)},p_+]$ is Hilbert-Schmidt, we can form the operator
$$ :{\sf G}(dT_{sc}^{(V)}): = \lim_{k \to \infty} {\sf G}_k(dT_{sc}^{(V)}) - (\Omega^{p_+},{\sf G}_k(dT_{sc}^{(V)})\Omega^{p_+}) $$
upon defining
$$ {\sf G}_k(dT_{sc}^{(V)}) = \sum_{j = 1}^k \left( \underline{\psi}(dT_{sc}^{(V)} \chi^+_j)^*\underline{\psi}(\chi^+_j)
 + \underline{\psi}(dT^{(V)}_{sc}\chi^-_j)^*\underline{\psi}(\chi^-_j) \right)\,.$$
According to Sec.\ 10 in \cite{Thaller} (cf.\ also \cite{CareyRuij}),
$:{\sf G}(dT_{sc}^{(V)}):$ defines an essentially selfadjoint operator on $\underline{\mathcal{W}}$ (to see the hermiticity,
use that $C dT^{(V)}_{sc} = - dT^{(V)}_{sc} C$), and it holds that
\begin{equation}
[:{\sf G}(dT_{sc}^{(V)}):,\underline{\psi}(v)] = \underline{\psi}(dT_{sc}^{(V)}v)
\end{equation}
for all $v \in L^2(\mathbb{R}^s,\mathbb{C}^N)$. 
Moreover, $:{\sf G}(dT_{sc}^{(V)}):$ is actually independent of the choice of $\{\chi^\pm_j\}_{j \in
\mathbb{N}}$.
(Note that in the notation of \cite{Thaller} and 
\cite{CareyRuij}, $:{\sf G}(dT_{sc}^{(V)}):$ would be written $: dT_{sc}^{(V)} \underline{\psi}^*\underline{\psi} :$.)
On the other hand we have, by eqn.\ \eqref{E_Stau} and owing to \eqref{E_underlinepsi}, the relation
\begin{equation} \label{E_Gamling}
 \left. \frac{d}{d\lambda} \right|_{\lambda = 0}
 S_{\tau_{sc}^{(\lambda V)}} \underline{\psi}(v) S_{\tau_{sc}^{(\lambda V)}}^{-1} =
  \underline{\psi} \left( \left. \frac{d}{d\lambda} \right|_{\lambda = 0} T_{sc}^{(\lambda V)}v\right)
 = \underline{\psi}(idT_{sc}^{(V)}v) 
\end{equation}
for all $v \in L^2(\mathbb{R}^s,\mathbb{C}^N)$,
resulting in
\begin{equation} \label{E_88a}
 \left. \frac{d}{d\lambda} \right|_{\lambda = 0}
 S_{\tau_{sc}^{(\lambda V)}} \underline{\psi}(v) S_{\tau_{sc}^{(\lambda V)}}^{-1}
 = [i:{\sf G}(dT_{sc}^{(V)}):,\underline{\psi}(v)] 
 \end{equation}
for all $v \in L^2(\mathbb{R}^s,\mathbb{C}^N)$.

In view of $\omega^{\rm vac} = \omega^{p_+} \circ \varrho$ with the morphism $\varrho$ in \eqref{E_varrhoDef} and
\eqref{E_rhoInterwRel}, there is a canonical unitary operator $\Larry : \mathcal{H}^{\rm vac} \to \mathcal{H}^{p_+}$ so that
\begin{equation} \label{Larrys}
 \Larry \pi^{\rm vac}(A) \Larry^{-1} = \pi^{p_+} \circ \varrho(A) \quad (A \in \mathfrak{F}(\mathcal{K}_0,C))\ \quad \text{and} \ \
 \Larry \Omega^{\rm vac} = \Omega^{p_+}\,.
\end{equation}
It is easy to see that $\Larry \mathcal{W} = \underline{\mathcal{W}}$ where $\mathcal{W}$ has been introduced as domain
for $:\bm\psi^+\bm\psi:(c)$ in Sec.\ 5. Furthermore, defining
$$ S_V^M = \Larry^{-1} S_{\tau_{sc}^{(V)}} \Larry\,,$$
and using \eqref{E_rhoInterwRel}, one can see that \eqref{E_Stau}, which was proven in
Prop.\ \ref{E_Smatrix}, is equivalent to
$$ S_V^M \pi^{\rm vac}(A) (S_V^M)^{-1} = \pi^{\rm vac}(\beta_V(A)) \quad (A \in \mathfrak{F}(\mathcal{K}_0,C))\,.$$
Setting
 $$ \Phi(c) = \Larry^{-1} :{\sf G}(dT^{(V)}_{sc}): \Larry \,,$$
it holds that $\Phi(c)$ is essentially selfadjoint on $\mathcal{W}$, and by eqn.\ \eqref{E_Gamling},
there holds
\begin{equation} \label{E_89a}
 \left. \frac{d}{d \lambda} \right|_{\lambda = 0} S_{\lambda V}^M \bm\psi(f) S_{\lambda V}^M{}^{-1} = [i\Phi(c),\bm\psi(f)]
= \bm\psi(VR_0 f) 
\end{equation}
for all $f \in C_0^\infty(\mathbb{R}) \otimes \mathscr{S}(\mathbb{R}^s) \otimes \mathbb{C}^N$. Thus we have established the
statements announced in Sec.\ 5.

A further remark is in order here. We do not directly establish the relation
$-i\left. \frac{d}{d \lambda} \right|_{\lambda = 0} S_{\lambda V}^M = \Phi(c)$. Notice that
the unitary $S_{\lambda V}^M$ implementing the
scattering transformation is not uniquely determined, but only determined up to a phase,      i.e.\ if $\tilde{S}_{\lambda V}^M$ is another choice of unitary implementer of the 
scattering matrix, then $\tilde{S}_{\lambda V}^M (S_{\lambda V}^M)^{-1} = {\rm e}^{ir(\lambda)}$
with some real-valued function $r(\lambda)$. 
However, if $\lambda \mapsto S_{\lambda V}^M$
is indeed differentiable at $\lambda = 0$ (upon
a suitable choice of $\lambda \mapsto r(\lambda)$), then it follows that its derivative with respect to $\lambda$ at $\lambda = 0$ in fact equals the above defined $\Phi(c)$ up to an
additive multiple of the unit operator. This is
a consequence of \eqref{E_89a} together with the
fact that the $*$-algebra generated by $\eins$ and the $\bm\psi(f)$ acts irreducibly in the vacuum representation. The additive constant may be compensated for by a re-definition of the
phase function $r(\lambda)$.     
${}$\\[10pt]
{\bf Proof of Proposition \ref{Prop_dT}}\\
In order to show that $[dT_{sc}^{(V)},p_+]$ is Hilbert-Schmidt, it is sufficient to
prove that $p_+ dT_{sc}^{(V)} p_-$ is Hilbert-Schmidt. It holds that
$$ p_+ dT_{sc}^{(V)} p_- = \int_{-\infty}^{\infty} \tilde{a}(t) {\rm e}^{iH_0t} p_+
 {\sf v} p_- {\rm e}^{-iH_0t}\,dt\,,$$
 implying
\begin{eqnarray*}
 & & \hspace*{-0.5cm}
F p_+ dT_{sc}^{(V)} p_- F^{-1}\\ &=& \int_{-\infty}^{\infty} \tilde{a}(t) F {\rm e}^{-iH_0t} p_+
F^{-1} F {\sf v}  F^{-1} F p_- {\rm e}^{iH_0t}F^{-1}\,dt\,,
\end{eqnarray*}
where $F$ is the Fourier-transform as in \eqref{E_FT}. In view of the Fourier-expressions for
$H_0$ (cf.\ \eqref{E_FHFin},\eqref{E_Hhat}), one has
$$
F {\rm e}^{iH_0t}p_+ F^{-1} \hat{g}(\underline{k}) = 
{\rm e}^{i\hat{H}_0(\underline{k})t} \hat{p}_+(\underline{k}) \hat{g}(\underline{k}) $$
where $\hat{p}_+(\underline{k})$ is an $N \times N$-matrix valued projection.
By the properties of $H_0$ and the resulting $\hat{H}_0(\underline{k})$,
it follows that there is a smooth family of unitary matrices
${\rm U}(\underline{k})$ diagonalizing $\hat{H}_0(\underline{k})$ and with the property
that
$$ {\rm U}(\underline{k}) \hat{p}_+(\underline{k}) {\rm U}(\underline{k})^{-1}
 = \left( \begin{array}{cc}
           \eins & 0 \\
           0 & 0 \end{array} \right) \equiv P_+ \,, \quad
  {\rm U}(\underline{k}) \hat{p}_-(\underline{k}) {\rm U}(\underline{k})^{-1}
 = \left( \begin{array}{cc}
           0 & 0 \\
           0 & \eins \end{array} \right) \equiv P_- \,, $$
where $\eins$ denotes the $N/2 \times N/2$-unit matrix (recall that $N$ is even). Since
${\rm U}(\underline{k}) \hat{H}_0(\underline{k}) {\rm U}(\underline{k})^{-1}$ is
diagonal and $\hat{H}_0(\underline{k})^*\hat{H}_0(\underline{k}) = |\hat{H}_0(\underline{k})|^2
\eins_{N \times N}$ is a multiple of the $N \times N$-unit matrix, the eigenvalues of
${\rm U}(\underline{k})\hat{H}_0(\underline{k}){\rm U}(\underline{k})^{-1}$ are $\pm |\hat{H}_0(\underline{k})|$, and one obtains
\begin{eqnarray*}
{\rm U}(\underline{k}) {\rm e}^{i\hat{H}_0(\underline{k})t}
 \hat{p}_+(\underline{k}) {\rm U}(\underline{k})^{-1} & = & {\rm e}^{-i|\hat{H}_0(\underline{k})|t} P_+\,, \\
{\rm U}(\underline{k}) {\rm e}^{-i\hat{H}_0(\underline{k})t}
 \hat{p}_-(\underline{k}) {\rm U}(\underline{k})^{-1} & = & {\rm e}^{-i|\hat{H}_0(\underline{k})|t} P_-\,.  
\end{eqnarray*}   
Using the Fourier-representations of ${\sf R}_b$ and ${\sf L}_b$ given in
\eqref{E_FLbFin} and \eqref{E_FRbFin}, it furthermore follows
that 
$$ F {\sf v} F^{-1}\hat{g}(\underline{k}) = \int \hat{{\sf v}}(\underline{k},\underline{\ell})
 \hat{g}(\underline{\ell})\,d^s\underline{\ell} \quad (g \in L^2(\mathbb{R}^s,\mathbb{C}^N)) $$
with a smooth, bounded, $N \times N$-matrix valued function $\hat{\sf v}(\underline{k},\underline{\ell})$. Taking together all these observations, we find 
\begin{eqnarray*}
 & & \hspace*{-0.5cm}
F p_+ dT_{sc}^{(V)} p_- F^{-1}\hat{g}(\underline{k})\\ &=& \int_{-\infty}^{\infty}
\int_{\mathbb{R}^s} \tilde{a}(t) {\rm e}^{-it(|\hat{H}_0(\underline{k})| +
|\hat{H}_0(\underline{\ell})|)} \hat{p}_+(\underline{k}) \hat{\sf v}(\underline{k},
\underline{\ell})\hat{p}_-(\underline{\ell})\hat{g}(\underline{\ell}) \,d^s\underline{\ell}
\,dt \\
&  = & \sqrt{2\pi}\int_{\mathbb{R}^s} (F\tilde{a})(|\hat{H}_0(\underline{k})| +
|\hat{H}_0(\underline{\ell})|) \hat{p}_+(\underline{k}) \hat{\sf v}(\underline{k},
\underline{\ell})\hat{p}_-(\underline{\ell})\hat{g}(\underline{\ell}) \,d^s\underline{\ell}\,.
\end{eqnarray*}
The Fourier-transform $(F \tilde{a})$ of $\tilde{a}$ is in the Schwartz-class while the
modulus of $\hat{p}_+(\underline{k}) \hat{\sf v}(\underline{k},
\underline{\ell})\hat{p}_-(\underline{\ell})$ is continuous and uniformly bounded; therefore
the integral kernel of the last integral is clearly $L^2$ in the $(\underline{k},\underline{\ell})$ 
variables, proving that 
 $F p_+ dT_{sc}^{(V)} p_- F^{-1}$ and hence $p_+ dT_{sc}^{(V)} p_-$ is Hilbert-Schmidt.
${}$ \hfill $\Box$

\section{Bogoliubov's formula}\label{S_BogForm}

In this section, we will derive the expressions for $d/d\lambda|_{\lambda =0} \beta_{\lambda V}$ that we
alluded to in Sec.\ 5 (cf.\ eqns.\ \eqref{E_Scatder} and \eqref{E_derivationM}).
We proceed using the geometrical setting from the last part of Section~\ref{S_DiracMoyalSpecial}
(which has also been investigated already in the commutative case in Section~\ref{S_DiracField}).
Under consideration is, respectively, one of the following ``potential'' operators:
\begin{eqnarray}
\text{(0) }(Vf)^A(x)&=&(V_\text{(0)}f)^A(x)=c(x)f^A(x) \\ 
\text{(i) }(Vf)^A(x)&=&(V_\text{(i)}f)^A(x)=(c\star_{(q,p)}f^A)(x)+(f^A\star_{(q,p)}c)(x)\\ 
\text{(ii) }(Vf)^A(x)&=&(V_\text{(ii)}f)^A(x)=(c\star_{(q,p)}f^A\star_{(q,p)}c)(x), 
\end{eqnarray}
where $c\in C^\infty_0(\bR,\bR)\otimes \sS(\bR^s,\bR)$ is a function of the form
\begin{equation}
c(x)=a(t)b(\ulx),
\end{equation}
with $a\in C^\infty_0(\bR,\bR)$, $b\in\sS(\bR^s,\bR)$, $t=x^0$, $\ulx=(x^1,\ldots,x^s)$. These operators act on
$f\in C^\infty_0(\bR)\otimes \sS(\bR^s)\otimes \bC^N$ and $(q+p)\times (q+p)$ is the dimension of the matrix $M$, which
still shall be of the form
\[
M=
\left[ 
\begin{array}{c|cc}
0&\cdots&0\\
\hline 
\vdots&M_{(q+p-1)\times (q+p-1)}&\\
0&&\\
\end{array}
\right]_{(q+p)\times(q+p)}
\]
These potentials act non-trivially only inside the time-slice
$\{(x^0,x^1,\ldots,x^s):\lambda_-<x^0<\lambda_+\}$ for some real numbers $\lambda_-<\lambda_+$. We recall
the definitions of the regions $G_+,G_-$,
\begin{eqnarray*}
G_+&=&\{(x^0,x^1,\ldots,x^s):x^0>\lambda_++\frac{1}{2}\}\\
G_-&=&\{(x^0,x^1,\ldots,x^s):x^0<\lambda_--\frac{1}{2}\},
\end{eqnarray*}
forming hyperbolic neighbourhoods of the Cauchy hyperplanes 
\begin{eqnarray*}
\Sigma_+&=&\{(x^0,x^1,\ldots,x^s):x^0=\lambda_++1\},\\
\Sigma_-&=&\{(x^0,x^1,\ldots,x^s):x^0=\lambda_--1\}
\end{eqnarray*}
respectively.

With these assumptions, we obtain the following result, the proof of which makes use
of Proposition~\ref{Prop1}, Lemma~\ref{Lemma2} (commutative case) and
Proposition~\ref{Prop1forMoyal}, Lemma~\ref{Lemma2forMoyal} (non-commutative case).

\begin{theorem}
 It holds that
\[
\left.\frac{d}{d\lambda}\right|_{\lambda=0}\beta_{\lambda V}(\Psi_0(f))=
\begin{cases}
\Psi_0(cR_0f)\\
\Psi_0(c\star_{(q,p)}R_0f+(R_0f)\star_{(q,p)}c)\\
\Psi_0(c\star_{(q,p)}(R_0f)\star_{(q,p)}c),
\end{cases}
\]
$\lambda$ being a real parameter, for the respective choices of the operator
{\rm $V = V_\text{(0)},V_\text{(i)},V_\text{(ii)}$}.
\end{theorem}

\begin{proof}
The priority of this proof lies on the non-commutative cases. The commutative case can be carried out along the same
lines. We recall the automorphism $\beta_{\lambda V}$ from~(\ref{E_DefBetaVforMoyal})
\begin{eqnarray}
&&\beta_{\lambda V}:\fF(\cK_0,C)\rightarrow \fF(\cK_0,C),\nonumber\\
&&\beta_{\lambda V}=\alpha_{0,-}\circ \alpha_{{\lambda V},-}^{-1}\circ \alpha_{{\lambda V},+}\circ \alpha_{0,+}^{-1},\nonumber
\end{eqnarray}
together with
\begin{equation*}
\beta_{\lambda V}(B_0([f]_0))=B_0(U_{\lambda V}[f]_0),
\end{equation*}
where $U_{\lambda V}$ is the unitary given by
\begin{equation*}
U_{\lambda V}=u_{0,-}\circ u_{{\lambda V},-}^{-1}\circ u_{{\lambda V},+}\circ u_{0,+}^{-1},
\end{equation*}
and where, similarly as for the isomorphisms above, we have used the abbreviations
\begin{equation*}
u_{0,\pm}=u^{G_\pm}_{0,0},\quad u_{{\lambda V},\pm}=u^{G_\pm}_{0,{\lambda V}}.
\end{equation*}
These equations arise from Lemma~\ref{Lemma2forMoyal}.
The  proof relies now on exactly the same strategy as the one for the very similar Theorem 4.3 of~\cite{BFV}.
With that in mind we start by finding more explicit expressions for the inverses in the chain of mappings
\begin{equation*}
U_{\lambda V}:\xymatrix{
[f]_0 \ar@{|->}[r]^{u_{0,+}^{-1}}& [f^{G_+}]^{G_+}_0 \ar@{|->}[r]^{u_{{\lambda V},+}}& [f^{G_+}]_{\lambda V}
\ar@{|->}[r]^{u_{{\lambda V},-}^{-1}}& [f^{G_-}]^{G_-}_0 \ar@{|->}[r]^{u_{0,-}}& [f^{G_-}]_0
},
\end{equation*}
where $f^{G_+}$ is any element in $C^\infty_0((\lambda_+,\infty))\otimes \sS(\bR^s)\otimes \bC^N$
such that $R_0(f-f^{G_+})=0$, and $f^{G_-}$ is any element in
$C^\infty_0((-\infty,\lambda_-))\otimes \sS(\bR^s)\otimes \bC^N$ such that $R_{\lambda V}(f^{G_+}-f^{G_-})=0$.
According to~\cite{BFV}, we choose
\begin{equation}\label{E_Def_fGpm}
f^{G_+}=-D_0\chi^\text{ret}_+R_0f,\quad
f^{G_-}=-D_{\lambda V}\chi^\text{ret}_-R_{\lambda V}f^{G_+},
\end{equation}
where $\chi^\text{ret}_\pm$ are defined as follows: It has been demanded that the open regions $G_\pm$ contain
Cauchy hyperplanes $\Sigma_\pm$. Then there are two pairs of further Cauchy surfaces in $G_\pm$, namely $\Sigma_\pm^\text{adv}$
in the timelike future of $\Sigma_\pm$ and $\Sigma_\pm^\text{ret}$ in the timelike past of $\Sigma_\pm$.
Thus
\begin{equation*}
\cO_\pm=[J^-(\Sigma_\pm^\text{adv})\cap J^+(\Sigma_\pm^\text{ret})]^\circ 
\end{equation*}
are open neighbourhoods of $\Sigma_\pm$ and $\overline{\cO_\pm}\subseteq G_\pm$. Now a partition of unity
$\{\chi^\text{adv}_\pm,\chi^\text{ret}_\pm\}$ is introduced on $\bR^n$ with $\chi^\text{adv}_\pm=0$ on
$J^-(\Sigma_\pm^\text{ret})$ and $\chi^\text{ret}_\pm=0$ on $J^+(\Sigma_\pm^\text{adv})$.

The crucial point lies in having to ensure that $f^{G_\pm}$ both lie in the the domain of
$R_{\lambda V}$ in view of
 the weakened support properties of the fundamental solutions $R^\pm_V$
(cf. Proposition~\ref{Prop1forMoyal}(\ref{Prop1forMoyal_b})) compared to the  situation in~\cite{BFV}.

Obviously it holds 
\begin{equation}\label{E_T_RelCauchyEvolNk_P_Dchi}
D_{\lambda V}\chi^\text{adv}_-R_{\lambda V} f
=-D_{\lambda V}\chi^\text{ret}_-R_{\lambda V}f.
\end{equation}
since $D_{\lambda V}R_{\lambda V}f=0$ and $\chi^\text{adv}_-+\chi^\text{ret}_-=1$.
The left hand side of~(\ref{E_T_RelCauchyEvolNk_P_Dchi}) vanishes on $J^-(\Sigma^\text{ret}_-)$
and the right hand side on $J^+(\Sigma^\text{adv}_-)$. 
Thus we can conclude from \eqref{E_T_RelCauchyEvolNk_P_Dchi} that both
$D_{\lambda V}\chi^\text{adv}_-R_{\lambda V} f$ and $-D_{\lambda V}\chi^\text{ret}_-R_{\lambda V}f$ lie in
$C_0^\infty(\mathbb{R}) \otimes \mathscr{S}(\mathbb{R}^s)$ and, hence, in the domain
of $R_{\lambda' V}$ for any $\lambda'$. This shows immediately that $f^{G_+}$ lies in the domain
of any $R_{\lambda' V}$, and, iterating the argument, the same holds for $f^{G_-}$.

Putting the definitions of~(\ref{E_Def_fGpm}) into the chain of mappings composing $U_{\lambda V}$ results in
\begin{equation*}
U_{\lambda V}[f]_0=[D_{\lambda V}\chi^\text{ret}_-R_{\lambda V}D_0\chi^\text{ret}_+R_0f]_0.
\end{equation*}
In the following we would like to abbreviate formally ``$\delta=\left.\frac{d}{d\lambda}\right|_{\lambda=0}$''.
We calculate
\begin{eqnarray*}
\delta U_{\lambda V}[f]_0&=&\delta [D_{\lambda V}\chi^\text{ret}_-R_{\lambda V}D_0\chi^\text{ret}_+R_0f]_0\\
&=&-[\delta D_{\lambda V}\chi^\text{ret}_-R_0f]_0+[D_0\chi^\text{ret}_-(\delta R_{\lambda V})D_0\chi^\text{ret}_+R_0f]_0,
\end{eqnarray*}
since $R_0D_0\chi^\text{ret}_+\varphi=-\varphi$. It is easy to see that $\delta D_{\lambda V}$ and
$\chi^\text{ret}_-$ have disjoint supports, and thus
\[
\delta U_{\lambda V}[f]_0=[D_0\chi^\text{ret}_-\delta R_{\lambda V}D_0\chi^\text{ret}_+R_0f]_0.
\]
$R_{\lambda V}=R_{\lambda V}^+-R_{\lambda V}^-$ implies
\[
\chi^\text{ret}_-R_{\lambda V}D_0\chi^\text{ret}_+\varphi
=\chi^\text{ret}_-R_{\lambda V}^+D_0\chi^\text{ret}_+\varphi-\chi^\text{ret}_-R_{\lambda V}^-D_0\chi^\text{ret}_+\varphi,
\]
whereof the first term on the right hand side vanishes, since $\supp \chi^\text{ret}_-\subseteq J^-(G_-)$,
$\supp R_{\lambda V}^+D_0\chi^\text{ret}_+\subseteq \overline{\cT^+(G_+)}$ and
$\overline{\cT^+(G_+)}\cap J^-(G_-)=\emptyset$. Hence
\[
\delta U_{\lambda V}[f]_0
=[-D_0\chi^\text{ret}_-\delta R_{\lambda V}^-D_0\chi^\text{ret}_+R_0f]_0.
\]
And this equals
\[
[D_0\chi^\text{ret}_-R_0^-\delta D_{\lambda V}R_0^-D_0\chi^\text{ret}_+R_0f]_0,
\]
because of the following deduction:
\begin{eqnarray*}
R_{\lambda V}^-D_{\lambda V}=\eins &\Rightarrow& (\delta R_{\lambda V}^-)D_0+R_0^-(\delta D_{\lambda V})=0\\
&\Rightarrow& \delta R_{\lambda V}^-D_0R_0^-=-R_0^-\delta D_{\lambda V}R_0^-\\
&\Rightarrow& \delta R_{\lambda V}^-=-R_0^-\delta D_{\lambda V}R_0^-.
\end{eqnarray*}
Support arguments lead to $\chi^\text{ret}_-R_0^+\delta D_{\lambda V}=0$ and
$\delta D_{\lambda V}R_0^+D_0\chi^\text{ret}_+\varphi=0$ and therefore
\[
\delta U_{\lambda V}[f]_0
=[D_0\chi^\text{ret}_-R_0\delta D_{\lambda V}R_0D_0\chi^\text{ret}_+R_0f]_0
=[\delta D_{\lambda V}R_0f]_0,
\]
since $R_0D_0\chi^\text{ret}_\pm=-\eins$.\\
Obviously $\delta D_{\lambda V}=V$, which is just one of the three choices of a potential operator.
Hence
\begin{eqnarray*}
\delta \beta_{\lambda V}(\Psi_0(f))
&=&\delta \beta _{\lambda V}(B_0([f]_0))
=\delta B_0(U_{\lambda V}[f]_0)\\
&=&B_0([VR_0f]_0)
=\Psi_0(VR_0f).
\end{eqnarray*}
\qed
\end{proof}

\section{Conclusion and outlook}

We have given a brief sketch of how a simple quantum
field theory on Moyal-Minkowski spacetime can be constructed in such a way as providing a model for the
construction of quantum field theories on more general
Lorentzian non-commutative spacetimes in a setting inspired by spectral geometry. For this quantum field
theory --- which is the quantized Dirac field ---
we have seen that a construction of observable field
operators labelled by elements of the deformed function algebra of
Moyal-Minkowski space can be derived, via Bogoliubov's
formula, from the S-matrix describing scattering of the
usual Dirac field on Minkowski spacetime by a non-commutative scalar potential. Again, it is feasible
that this procedure can be generalized to more general
Lorentzian non-commutative spacetimes.

However, it is certainly inappropriate, at this stage, to
judge the ge\-ne\-ra\-lity of the method. Moyal-Minkowski spacetime with commutative time is a very simple and
very special non-commutative geometry whose physical relevance is not compelling, to say the least. On the other hand, due to the quite unusual and counter-intuitive
properties of non-commutative geometries, one surely
needs examples as one's guidance towards developing
physical theories in non-commutative geometries, such as
quantum field theory. This clearly shows a dilemma: 
The examples for non-commutative geometries that are
manageable are probably un-physical and may therefore do a very poor job as a guidance when attempting to find some central principles, while without such principles, it is hard to judge which non-commutative geometries are related to
physics. Nevertheless, one can probably do better, and
try and investigate our method of construction of quantum
field theories and their observables for non-commutative
spacetimes that have a greater physical appeal, such as
developed in \cite{DFR} and \cite{WaldmannBahns}, for
example. Even if this appears to be a considerably 
more difficult task, we think it is worth an attempt.   

\appendix

\section{The Action of the Wick square}
In this Appendix we will prove that the Wick-square acts as a derivation on the
Dirac field operators in the same way as the derivative of the S-matrix with respect to
the scalar scattering potential. The assumptions, wherever not spelled out in detail, are
those stated in Sec.\ 5.
\begin{proposition}
Let $e_\mu$ and $\eta_\mu$ ($\mu = 1,\ldots,L$) be elements in $\mathbb{C}^N$, chosen
such that $\sum_{\mu = 1}^L \overline{e_\mu}{}^A \eta_{\mu}^B = \frac{1}{4}\gamma_0{}^{AB}$. Define
for $q_1, q_2 \in C_0^\infty(\mathbb{R}^n,\mathbb{R})$,
$$ \bm \psi^+\bm \psi(q_1 \otimes q_2) = \sum_{\mu = 1}^L \bm \psi
(q_1 e_\mu)^* \bm \psi(q_2 \eta_\mu) $$
and whence,
$$ :\bm \psi^+ \bm \psi :(c) = \lim_{\epsilon \to 0} \ \bm \psi^+\bm \psi(F_\epsilon) -
   (\Omega^{\rm vac},\bm \psi^+\bm \psi(F_\epsilon)\Omega^{\rm vac}) \eins $$
   on $\mathcal{W}$
with $F_\epsilon(x,y) = q_1(x)q_2(y) j_\epsilon(x-y)$ and $c(x) = q_1(x)q_2(x)$ $(x,y \in \mathbb{R}^N)$, where
$\lim_{\epsilon \to 0}\int q(x)j_\epsilon(x -y)\,d^nx = q(y)$ $(q \in C_0^\infty(\mathbb{R}^n))$. 
\\[6pt]
Then 
$$ [:\bm \psi^+\bm \psi:(c),\bm \psi(f)] = -i\bm \psi(c R_0 f) \quad (f \in C_0^\infty(\mathbb{R}) \otimes 
\mathscr{S}(\mathbb{R}^s)\otimes \mathbb{C}^N) \,.$$
Moreover, $:\bm \psi^+\bm \psi:(c)$ is independent of the choice of families $e_\mu , \eta_\mu \in \mathbb{C}^N$
fulfilling $\sum_\mu \overline{e_\mu}{}^A \eta_\mu^B = \frac{1}{4}\gamma_0{}^{AB}$.
\end{proposition}
{\bf Proof} \ \ Using the relations of the generators of the CAR algebra, it follows that 
\begin{eqnarray} \label{E_ap1}
 & & {} \hspace*{-1cm} \sum_\mu (\chi_1, [\bm \psi(q_1 e_\mu)^*\bm \psi(q_2 \eta_\mu), \psi(f)] \chi_2) \\
 & & {} \hspace*{-0.5cm} = 2 \sum_\mu \left\{ (\chi_1,\bm \psi (C q_1 e_\mu)\chi_2) (Cq_2 \eta_\mu,f)_{0} -
                     (\chi_1,\bm \psi(q_2 \eta_\mu)\chi_2) (q_1 e_\mu,f)_{0} \right\} \nonumber
\end{eqnarray}                     
holds for all vectors $\chi_1,\chi_2$ in the dense domain $\mathcal{W} \subset \mathcal{H}^{\rm vac}$ and for all
$f \in C_0^\infty(\mathbb{R}) \otimes 
\mathscr{S}(\mathbb{R}^s) \otimes \mathbb{C}^N$.
We recall here the definition
$$ (f,h)_{0} = i \int \gamma_{0AB} \overline{f}{}^B(x)(R_0 h)^A(x)\, d^nx\,$$
(for $V=0$, see Prop.~\ref{Prop1forMoyal} or respectively~\ref{Prop1} and eqn.~(\ref{E_SesqLinFormLRangle})). One can show
either directly, or by falling back onto general arguments \cite{FreHer}, that for each pair of vectors
$\chi_1,\chi_2 \in \mathcal{W}$ there are smooth functions $\xi_A$ on $\mathbb{R}^n$ ($A = 1,\ldots,N$) such that
\begin{equation} \label{E_ap2_first}
(\chi_1,\bm \psi(f) \chi_2) = \int \xi_A(y) f^A(y)\,d^n y \quad (f \in C_0^\infty(\mathbb{R}) \otimes 
\mathscr{S}(\mathbb{R}^s) \otimes \mathbb{C}^N) 
\,.
\end{equation}
With this notation, the right hand side of \eqref{E_ap1} assumes the form
\begin{eqnarray} \label{E_ap2}
 & & {}\hspace*{-1cm}  2i \sum_{\mu} \iint \overline{q}_1(y) q_2(x) \xi_A(y) (C e_\mu)^A(y) \gamma_{0BD}
          \overline{C\eta_\mu}{}^B(R_0f)^D(x)\,d^nx\,d^ny \\
 & & {}\hspace*{-0.5cm} - 2i \sum_{\mu} \iint \overline{q}_1(x) q_2(y) \xi_A(y)\eta_\mu^A
 \gamma_{0 D'B'} \overline{e_\mu}{}^{B'}(R_0f)^{D'}(x) \,d^nx\, d^ny \nonumber 
\end{eqnarray}
Using the defining property $\sum_{\mu} \overline{e_\mu}{}^A\eta_\mu^B = \frac{1}{4}\gamma_0{}^{AB}$ and $\gamma_0{}^2 = \eins$,
the second integral of \eqref{E_ap2} simplifies to
\begin{equation} \label{E_ap3}
-\frac{i}{2}  \iint \overline{q}_1(x) q_2(y) \xi_A(y)(R_0f)^A(x)\,d^nx \,d^ny\,.
\end{equation}
In a similar manner, using also the relations $C^2 = \eins$ and~\eqref{E_CSkewForSesqu}, one can check that the first integral in
\eqref{E_ap2} simplifies to
\begin{equation} \label{E_ap4}
-\frac{i}{2} \iint \overline{q}_1(y) q_2(x) \xi_A(y)(R_0f)^A(x)\,d^nx \,d^ny\,,
\end{equation}          
so that \eqref{E_ap2} becomes equal to
\begin{equation}\label{E_ap5}
-\frac{i}{2} \iint (q_1(y)q_2(x) + q_2(y)q_1(x))\xi_A(y)(R_0f)^A(x)\,d^nx \, d^ny\,,
\end{equation}
observing that $q_1$ and $q_2$ are real-valued.
Replacing here $q_1 \otimes q_2$ by $F_\epsilon$ and taking the limit $\epsilon \to 0$ turns the last expression
into
\begin{equation} \label{E_ap6}
-i\int  \xi_A(x) c(x)(R_0f)^A(x)\, d^nx\,.
\end{equation}
On using \eqref{E_ap2_first}, we have therefore proved the first claim of the Proposition.

To see the independence of the definition of $:\bm\psi^+\bm\psi:(c)$ of the mentioned choices,
we note that the commutator formula for $:\bm\psi^+\bm\psi:(c)$, together with the fact that the
$*$-algebra generated by $\eins$ and all the $\bm\psi(f)$ acts irreducibly, fixes
$:\bm\psi^+\bm\psi:(c)$ up to addition of a multiple of the unit operator $\eins$. On the other
hand, by construction we have $(\Omega^{\rm vac},:\bm\psi^+\bm\psi:(c)\Omega^{\rm vac}) = 0$,
so that the scalar multiple in question must vanish in the vacuum state, which implies that it is
zero. This demonstrates the claimed independence of the definition of $:\bm\psi^+\bm\psi:(c)$ of
the possible choices for $e_\mu$ and $\eta_\mu$.  
${}$ \hfill $\Box$ 
${}$\\[30pt]
{\bf Acknowledgments}
The second named author would like to thank Mario Paschke for many discussions on the topics
presented here and for continuing collaboration on Lorentzian spectral geometry of which 
some ideas have been sketched in the text. Thanks are also extended to Sergio Doplicher,
Klaus Fredenhagen and Raimar Wulkenhaar for discussions and comments. The first named author
gratefully acknowledges financial support by the DFG.

\end{document}